%% file: CG-QCA-gamma7_amend.tex
\documentclass[a4paper,11pt,accepted=2025-05-16]{quantumarticle}
\pdfoutput=1
\usepackage{epsfig,amsmath,amssymb,mathrsfs,mathtools,amscd}
\usepackage[numbers]{natbib}
\usepackage[linkcolor=red]{hyperref}
\usepackage[dvipsnames]{xcolor}
\usepackage{soul}
\usepackage{multirow}
\usepackage{enumerate}
\usepackage{amsmath} 
\usepackage{amsthm}
\usepackage{verbatim} 
\usepackage{mathtools} 
\usepackage{float} 
\usepackage{tikz}
\usepackage{tikzit}
\usepackage{tabularx}
\usepackage{array}
\usepackage{subcaption}
\usepackage{graphicx}

\newcolumntype{Y}{>{\centering\arraybackslash}m{\hsize}}
\newtheorem{lemma}{Lemma}
\newtheorem{Corollary}{Corollary}
\newtheorem{definition}{Definition}

\newtheorem{proposition}{Proposition}
\newtheorem{remark}{Remark}

\def\ind{\ensuremath{\operatorname{ind}}}

\def\transf#1{\ensuremath{\mathcal{#1}}}
\def\tH{\transf H}
\def\tI{\transf I}
\def\tJ{\transf V ^\dag}
\def\tK{\transf K}

\def\tP{\transf P}

\def\tS{\transf S}
\def\tT{\transf T}
\def\tU{\transf U}
\def\tV{\transf V}

\def\tZ{\transf Z}
\def\tor{\ensuremath{\bigtimes_{t=1}^s\mathbb Z_{f_t}}}

\def\Jop{J_{\Lambda_{{x}}}}
\def\MdC{\mathsf M_d(\mathbb C)}
\def\J{J}
\def\Tred{(\tT^N)_r}

\def\Tredmat{(U_{\tT^N})_r}
\def\Tmat{(U_\tT)}
\def\Tnmat{U_{\tT^N}}

\def\Jopconj#1{\left(\bigotimes_{x\in \mathcal{L'}} J^{#1}_{\Lambda_{{x}}}\right)}
\def\Aqloc#1{\ensuremath{\mathsf{A}(\mathbb{Z}^{#1})}}

\def\Aplocc{\ensuremath{\mathsf{A}\rvert_{\Pi}(\mathcal{L})}}

\def\Aredloc{\ensuremath{\mathsf A(\mathcal{L'})}}

\mathchardef\minus="002D

\def\<{\langle}
\def\>{\rangle}
 \def\ket#1{| #1 \rangle}
 \def\dket#1{| #1 \rangle\!\rangle}
 
\def\bra#1{\langle #1 |}
\def\ketbra#1#2{| #1 \rangle\!\langle#2 |}

\def\braket#1#2{\langle #1 | #2 \rangle}

\def\vc#1{\boldsymbol{\mathrm #1}}

\def\L{\ensuremath{\mathcal{L}}}
\def\uon#1{\ensuremath{\lVert #1\rVert}}

\def\Mod{\operatorname{mod}}

\newcommand\Item[1][]{%
  \ifx\relax#1\relax  \item \else \item[#1] \fi
  \abovedisplayskip=0pt\abovedisplayshortskip=0pt~\vspace*{-\baselineskip}}

\newcommand{\Tr}{\mathop{\mathrm{Tr}}}
\input{tikzstyle.tex}
\input{sample.tikzstyles}

\makeatletter
\expandafter\let\csname qfix@@array\endcsname\@array
\expandafter\let\csname qfix@@@array\endcsname\@@array
\expandafter\let\csname qfix@@tabular\endcsname\@tabular
\expandafter\let\csname qfix@@tabarray\endcsname\@tabarray
\expandafter\let\csname qfix@array\endcsname\array
\expandafter\let\csname qfix@endarray\endcsname\endarray
\expandafter\let\csname qfix@endtabular\endcsname\endtabular
\expandafter\let\csname qfix@@mkpream\endcsname\@mkpream
\expandafter\let\csname qfix@@classx\endcsname\@classx
\expandafter\let\csname qfix@@classv\endcsname\@classv
\expandafter\let\csname qfix@@arrayacol\endcsname\@arrayacol
\expandafter\let\csname qfix@@tabacol\endcsname\@tabacol
\expandafter\let\csname qfix@insert@column\endcsname\insert@column
\expandafter\let\csname qfix@@arraycr\endcsname\@arraycr
\expandafter\let\csname qfix@@xarraycr\endcsname\@xarraycr
\expandafter\let\csname qfix@@xargarraycr\endcsname\@xargarraycr
\expandafter\let\csname qfix@@yargarraycr\endcsname\@yargarraycr
\expandafter\let\csname qfix@hline\endcsname\hline
\expandafter\let\csname qfix@@tabularcr\endcsname\@tabularcr
\expandafter\let\csname qfix@@xtabularcr\endcsname\@xtabularcr
\expandafter\let\csname qfix@multicolumn\endcsname\multicolumn
\expandafter\let\csname qfix@@addamp\endcsname\@addamp
\expandafter\let\csname qfix@@arrayclassz\endcsname\@arrayclassz
\expandafter\let\csname qfix@@tabclassiv\endcsname\@tabclassiv
\expandafter\let\csname qfix@@arrayclassiv\endcsname\@arrayclassiv
\expandafter\let\csname qfix@@tabclassz\endcsname\@tabclassz
\expandafter\let\csname qfix@endtabular*\expandafter\endcsname\csname endtabular*\endcsname
\AtBeginDocument{%
  \expandafter\let\expandafter\@array\csname qfix@@array\endcsname
  \expandafter\let\expandafter\@@array\csname qfix@@@array\endcsname
  \expandafter\let\expandafter\@tabular\csname qfix@@tabular\endcsname
  \expandafter\let\expandafter\@tabarray\csname qfix@@tabarray\endcsname
  \expandafter\let\expandafter\array\csname qfix@array\endcsname
  \expandafter\let\expandafter\endarray\csname qfix@endarray\endcsname
  \expandafter\let\expandafter\endtabular\csname qfix@endtabular\endcsname
  \expandafter\let\expandafter\@mkpream\csname qfix@@mkpream\endcsname
  \expandafter\let\expandafter\@classx\csname qfix@@classx\endcsname
  \expandafter\let\expandafter\@classv\csname qfix@@classv\endcsname
  \expandafter\let\expandafter\@arrayacol\csname qfix@@arrayacol\endcsname
  \expandafter\let\expandafter\@tabacol\csname qfix@@tabacol\endcsname
  \expandafter\let\expandafter\insert@column\csname qfix@insert@column\endcsname
  \expandafter\let\expandafter\@arraycr\csname qfix@@arraycr\endcsname
  \expandafter\let\expandafter\@xarraycr\csname qfix@@xarraycr\endcsname
  \expandafter\let\expandafter\@xargarraycr\csname qfix@@xargarraycr\endcsname
  \expandafter\let\expandafter\@yargarraycr\csname qfix@@yargarraycr\endcsname
  \expandafter\let\expandafter\hline\csname qfix@hline\endcsname
  \expandafter\let\expandafter\@tabularcr\csname qfix@@tabularcr\endcsname
  \expandafter\let\expandafter\@xtabularcr\csname qfix@@xtabularcr\endcsname
  \expandafter\let\expandafter\multicolumn\csname qfix@multicolumn\endcsname
  \expandafter\let\expandafter\@addamp\csname qfix@@addamp\endcsname
  \expandafter\let\expandafter\@arrayclassz\csname qfix@@arrayclassz\endcsname
  \expandafter\let\expandafter\@tabclassiv\csname qfix@@tabclassiv\endcsname
  \expandafter\let\expandafter\@arrayclassiv\csname qfix@@arrayclassiv\endcsname
  \expandafter\let\expandafter\@tabclassz\csname qfix@@tabclassz\endcsname
  \expandafter\let\csname endtabular*\expandafter\endcsname\csname qfix@endtabular*\endcsname
}
\makeatother
\begin{document}
\title{Renormalisation of quantum cellular automata}

\author{Lorenzo Siro Trezzini} \email[]{lorenzosiro.trezzini01@universitadipavia.it}
\affiliation{Dipartimento di Fisica dell'Universit\`a di Pavia, via
  Bassi 6, 27100 Pavia} \affiliation{Istituto Nazionale di Fisica
  Nucleare, Gruppo IV, via Bassi 6, 27100 Pavia} 
\author{Alessandro  Bisio} \email[]{alessandro.bisio@unipv.it}
\affiliation{Dipartimento di Fisica dell'Universit\`a di Pavia, via
  Bassi 6, 27100 Pavia} \affiliation{Istituto Nazionale di Fisica
  Nucleare, Gruppo IV, via Bassi 6, 27100 Pavia} 
\author{Paolo Perinotti} \email[]{paolo.perinotti@unipv.it}
\affiliation{Dipartimento di Fisica dell'Universit\`a di Pavia, via
  Bassi 6, 27100 Pavia} \affiliation{Istituto Nazionale di Fisica
  Nucleare, Gruppo IV, via Bassi 6, 27100 Pavia} 

\begin{abstract} 
  We study a coarse-graining procedure for quantum
  cellular automata on hypercubic lattices that
  consists in grouping neighboring cells into tiles
  and selecting a subspace within each tile. This
  is done in such a way that multiple evolution
  steps applied to this subspace can be viewed as
  a single evolution step of a new quantum
  cellular automaton, whose cells are the
  subspaces themselves. We derive a necessary and
  sufficient condition for renormalizability and
  use it to investigate the renormalization flow
  of cellular automata on a line, where the cells
  are qubits and the tiles are composed of two
  neighboring cells. The problem is exhaustively
  solved, and the fixed points of the
  renormalization flow are highlighted.
\end{abstract} 

\maketitle
\section{Introduction}
Quantum Cellular Automata (QCA) have been developed in the years since the very advent of quantum computation~\cite{Feynman:1982aa}, as the quantum counterpart of classical cellular automata~
The notion of a QCA was made more and more rigorous over the years~
until the algebraic formulation that is still presently used~\cite{schumacher2004reversiblequantumcellularautomata,10.1007/978-3-540-88282-4_8}. 
The interest in QCA is due to many reasons, the main ones being that they constitute a computational model that is equivalent to a universal quantum computer~\cite{PhysRevLett.97.020502,Arrighi:2009aa}, that
their geometric structure suggests their use in the simulation and modelling of quantum physical systems~\cite{PhysRevA.90.062106,bisio2016quantum,Mallick:2016aa,Arrighi:2020aa,bisio2018thirring,bisio2021scattering,Eon2023relativistic}, and that they represent interesting many body systems~\cite{PhysRevLett.125.190402,Jahn2022tensornetworkmodels,hillberry2024integrabilitygoldilocksquantumcellular}, with a peculiar role in the study of the topological phases of matter~\cite{stephen2019subsystem,po2016chiral,fidkowski2019interacting,kitaev2006topological}. The characteristic trait of QCAs, as of cellular automata in general~\cite{Perinotti2020cellularautomatain} is that they are defined by a simple local rule, and yet they can give rise to arbitrarily complex and interesting large-scale dynamics, as witnessed by their computational universality. As a consequence, the connection between the local rule and the large-scale phenomenology is clearly a problem as difficult as crucial, that calls for understanding. 

A typical feature of large-scale behaviour of physical models is that it is rather 
insensitive to most of the details of small-scale dynamics. This statement reflects 
e.g.~the notion of universality classes of statistical mechanical models \footnote{We 
stress that the universality notion invoked here is different from computational 
universality, which in simple words refers to the ability of a computational model to 
simulate any other conceivable computer.}, where the relevant physical parameters 
take on values in an astonishingly limited range, collecting very different models in the 
same large-scale equivalence class~\cite{10.1093/acprof:oso/9780199227198.001.0001}. In 
this context, the renormalisation procedure due to 
Kadanoff~\cite{PhysicsPhysiqueFizika.2.263}, and perfected by 
Wilson~\cite{PhysRevB.4.3174,PhysRevB.4.3184}, represents the basic idea underpinning all 
the techniques for recovering the large-scale features of a physical model starting from 
its microscopic description, and is based on the progressive erasure 
of degrees of freedom that are supposed to be irrelevant to the large-scale picture. The 
validation of the approach comes from the convergence of the recursive application of the 
procedure, and thus to the existence of fixed point models that become stable in the 
process. This technique led to the discovery of critical exponents in phase transitions 
and is intimately connected to renormalisation techniques in quantum field 
theories~\cite{10.1093/acprof:oso/9780199227198.001.0001}. A similar approach to the 
large scale behaviour of \emph{classical} cellular automata was explored in 
Ref.~\cite{PhysRevLett.92.074105,PhysRevE.73.026203}, where the renormalisation flow was 
interestingly connected to the idea that a cellular automaton $A$ that is renormalised to 
another one $B$ has a larger complexity than $B$, in the sense that it can simulate $B$ 
and then it is at least as complex as $B$. Although our focus will be on such Kadanoff-like renormalisation, it is worth noting that the challenge of analysing the large-scale behaviour of QCA and Quantum Walks has already been explored in the literature (see, e.g., \cite{Duranthon_2021, ScalingQuantumWalks}). The key distinction from previous works lies in the requirement that the model describing the large-scale behaviour of a QCA must be a QCA itself. This entails the possibility of defining a notion of ``renormalisation flow'' in the space of QCAs and, consequently, to identify fixed points.

The connection of small-scale dynamics with its large-scale \emph{emergent} phenomenology 
is of particular interest for application of the theory of QCAs to quantum simulations,
where QCAs serve both as primitive models~\cite{PhysRevA.90.062106} and as approximations of Floquet systems or Trotterised approximations of local Hamiltonians~\cite{PhysRevB.106.184304}. The paradigm of quantum computation---even more so in the present, so-called NISQ era~\cite{Preskill2018quantumcomputingin}---is based on the run of algorithms on a large noisy physical qubit register where only a smaller logical nearly error-free register is considered. This is an extreme summary of the notion of \emph{error correction}. A renormalisation technique, discarding detailed and fast variables in favour of a fraction of the total pool of degrees of freedom, where an effective dynamics occurs, can then be seen as an error correction technique for quantum simulations through QCAs: a fine-grained QCA where we do not have full control can be used to run a simulation on its renormalised QCA. This allows us to interpret the fine-grained QCA as a noise-resilient variant of the coarse-grained one. The issue of quantum error correction in the context of QCA has also been discussed in relation to the implementation of computational tasks, e.g.~density classification on a noisy register~\cite{ErrorCorrection}.  In this respect, we observe that renormalisability can be seen as an instance of \emph{intrinsic simulation}~\cite{van1996quantum}: a QCA $\tT$ that can be renormalised to another one $\tS$  actually \emph{simulates} $\tS$---provided that one programs it in the suitable degrees of freedom. In particular, a QCA that could be renormalised to any other QCA would be \emph{intrinsically universal}~\cite{Arrighi:2009aa}.

In the present work, we adapt the Kadanoff renormalisation idea to the coarse-graining of 
quantum cellular automata. Unlike the case of classical cellular automata, it is so hard 
to include irreversible evolutions in the picture of quantum cellular automata that the 
very definition of a QCA entails reversibility. This fact determines the consequence that, 
despite the initial formulation of the problem of coarse-graining is very close to that of 
classical cellular automata~\cite{PhysRevLett.92.074105,PhysRevE.73.026203} or even of 
classical networks~\cite{10.1093/comnet/cnx055}, the subsequent solution has to face more 
constraints and has a radically different family of solutions. We start analysing the case 
of QCAs on a hypercubic lattice, and pose the question as to what QCAs $A$ admit a 
block-tiling and, on each block, a projection $P$ on a subspace that is equivalent to that 
of a single cell, with the requirement that projecting after $N$ steps is equivalent to 
projecting right away, then applying one step of a different QCA $B$ on the lattice whose 
elementary cells correspond to the tiles. After writing the general equations whose 
solution gives both the projection $P$ and the renormalised QCA $B$, we restrict to the 
case of 1-dimensional QCAs that can be implemented as a Finite Depth Quantum Circuit 
(FDQC), with tiles made of two neighbouring nodes and coarse-graining two time-steps into 
one. The problem in this case can be manipulated, revealing a general algebraic structure 
that can be further simplified in the case where the cells are qubits. The main reason why 
the simplification occurs is that for qubit QCAs a full classification of FDQCs has been 
provided~\cite{schumacher2004reversiblequantumcellularautomata}. We then solve the problem 
exhaustively, providing the renormalisation flow for the 1-d qubit case, along with a 
classification of the fixed points.

As a remark, we can observe that very few qubit QCAs are renormalisable in this sense, and 
those that are  {have a renormalized evolutions that fall into three classes: (i) local rotations on each qubit (ii) shifts (iii) unitaries that are diagonal in the computational basis}. This fact can reasonably be ascribed to the very 
limited choice of degrees of freedom due to the smallness of both the cells and the tiles 
that we are considering. The solution, however, provides technical tools that pave the way 
to similar analysis in more general cases, even considering special case-studies. It is interesting to notice that the renormalisation scheme \emph{can} change the topological index, as shown by particular solutions of the qubit case.

The manuscript is structured as follows: In section \eqref{sec:qca} we introduce the 
notion of Quantum Cellular Automaton following the definition proposed by Schumacher and 
Werner \cite{schumacher2004reversiblequantumcellularautomata}. 
We will then focus on two important tools that we will use extensively: the Wrapping Lemma 
and the index. The former is a procedure to represent the evolution of a QCA defined over 
an infinite lattice of cells (say, for concreteness, $\mathbb{Z}^s$) on a finite lattice. 
The latter is a topological invariant quantity that classifies the QCA modulo FDQC. 
In section \eqref{sec:CG} we introduce the notion of coarse-graining and renormalisation, 
and we derive the renormalisability conditions. 
In section \eqref{sec:FDQC} we focus our attention on the case of one-dimensional FDQC. 
Finally, in section \eqref{sec:Qubits} we solve the above equations in the case
of QCA with qubit cells. After the classification of all such QCAs that are 
renormalisable, we study the renormalised evolutions and the fixed points of the 
renormalisation flow.

\section{Quantum Cellular Automata}\label{sec:qca}
This section is devoted to review the existing litterature about Quantum Cellular Automata.
A Quantum Cellular Automaton $\tT$ is a reversible, local and homogeneous
discrete time evolution of a lattice $\mathcal L$ of cells, each
containing a finite-dimensional quantum system.
The locality condition means that, at each step of
the evolution, a cell $x \in \mathcal L$ interacts
with a finite set $\mathcal{N}_x$ of
cells, called the \emph{neigbhourhood} of $x$. Therefore, the speed of
propagation of information in a QCA is bounded
and we have a natural notion of future and past ``light 
cone'' or, more precisely, cone of causal influence.  
Homogeneity---or translation invariance---on the other hand,
is the requirement that every cell is updated
according to the same rule. The latter is equivalent to the 
requirement that the evolution commutes with the
lattice translations~\footnote{Sometimes in the literature QCAs can be defined 
without the homogeneity property, and what we call a QCA would then be a
\emph{homogeneous} QCA (see e.g.~Ref.~\cite{Index}).}.

Let us now make these
concepts mathematically precise.
First of all, for the purpose of the present study we will focus on the case where the 
lattice $\mathcal L$ is a square lattice, i.e.
\begin{align}
\L \coloneqq \tor, \quad  f_i \in \mathbb{N} \cup \infty, \quad s < \infty, \label{eq:latticedefinition}
\end{align}
where $\mathbb Z_{f_i}$ denotes the additive group
modulo $f_i$, and $\mathbb Z_\infty=\mathbb Z$
\footnote{The most general structure on which
  homogeneous QCAs can be defined is the
  \emph{Cayley graph} of some finitely presented
  group $G$~\cite{arrighi:hal-01785458,
    hadeler2017cellular,Perinotti2020cellularautomatain}. However
  here we will restrict to the case where $G$ is
  abelian.}.  Therefore, each cell of the lattice
is represented by a vector $x \in \tor $ of
integer coefficients.

This notation allows us to treat the case in which
$\mathcal L$ is infinite and that where
$\mathcal L$ is finite in a unified way.
Typically, one considers the lattice $\mathcal L$ as a graph, the edges corresponding to 
neighbourhood relations. One can prove that the homogeneity requirement makes the graph
of causal connections the \emph{Cayley graph} of some 
group~\cite{PhysRevA.90.062106,DAriano:2016aa,Perinotti2020cellularautomatain}, that in 
the present case is $\mathcal L$ imagined as an abelian additive group.

Since we will deal with lattices of infinitely many
quantum systems, it is convenient to formulate the
QCA dynamics in the Heisenberg picture, namely defining
the evolution by its action on the algebra of operators 
rather than on states. This approach, introduced in 
Ref.~\cite{schumacher2004reversiblequantumcellularautomata} and inspired by the algebraic approach to the
statistical mechanics of spin systems~\cite{bratteli1987operator} or to quantum field theory~\cite{haag1992local} has the advantage that the notion of local observable is
always well defined, whereas defining locality through the action on 
the local state is highly problematic. This
construction is based on the concept of
\textit{quasi-local algebra} of operators
$\mathsf{A}(\mathcal L)$.  If each cell contains a
finite dimensional quantum system of dimension
$d\ge2$, we can associate to each cell $x$ the
observable algebra $\mathsf A_x$ which is
isomorphic to the $C^*$-algebra $\MdC$ of $d\times d$
complex matrices.

When $\Lambda \subset \L$ is a
finite subset of $\L$ we denote with
$\mathsf{A}(\Lambda)$ the algebra of operators
acting on all cells in $\Lambda$, namely
$\mathsf{A}(\Lambda)=\bigotimes_{x\in\Lambda}\mathsf{A}_x$. If
we consider two different subsets
$\Lambda_1,\Lambda_2\subset\L$, we can consider
$\mathsf{A}(\Lambda_1)$ as a subalgebra of
$\mathsf{A}(\Lambda_1\cup\Lambda_2)$ by tensoring
each element $O\in\mathsf{A}(\Lambda_1)$ with the
identity over $\Lambda_2\setminus \Lambda_1$,
namely
$O\otimes I_{\Lambda_2 \setminus
  \Lambda_1}\in\mathsf{A}(\Lambda_1\cup\Lambda_2)$. In
this way the product $O_1O_2$ with
$O_i\in\mathsf{A}(\Lambda_i)$ is a well defined
element of
$\mathsf{A}(\Lambda_1\cup\Lambda_2)$. Each local
algebra is naturally endowed with the operator
norm $\uon \cdot$. We denote with
$\mathsf{A}^\mathrm{loc}(\L)$ the \textit{local algebra}
of operators, i.e. the algebra of all the
operators acting non-trivially over any finite set
of cells.
Since tensoring with the identity does
not change the norm of an operator, every local operator has a well defined norm, 
which is independent of the region on which we are considering the 
action of the operator. Thus, $\mathsf{A}^\mathrm{loc}(\L)$, is a
normed algebra, whose completion is the
quasi-local algebra $\mathsf{A}(\L)$. For
a detailed construction of $\mathsf{A}(\L)$ see
Appendix \eqref{sec:qloc}.  Clearly, for finite
lattices $\mathsf{A}^\mathrm{loc}(\L)$ is a finite
dimensional algebra and
$\mathsf{A}(\L)=\mathsf{A}^\mathrm{loc}(\L)$.  We could
regard $\mathsf{A}(\L)$ as the $C^*$-algebra of
operators that can be approximated arbitrarily
well by local operators.

A reversible dynamics then 
corresponds to an automorphism
$\tT : \mathsf{A}(\L) \to \mathsf{A}(\L)$ of the
quasi local algebra. The algebraic approach allows
us not to choose a representation of the observable
algebra. When dealing with finitely many systems
(i.e. when $\L$ is a finite lattice) this is rather
negligible advantage, since all the representation
are unitarily equivalent. In this case, then $\tT $ 
would be the conjugation by some unitary operator, i.e.
$\tT(O) = U^\dag OU$. However, when $\L$ is an
infinite lattice we have that $\mathsf{A}(\L)$ is
an infinite dimensional $C^*$-algebra, which
generally decomposes into inequivalent representations. 
Thus working in a representation-independent framework is highly
advantageous.

In order to define a homogeneous dynamics, we have to define the
operators that implement the lattice
translations.
\begin{definition}[Shift]\label{def:shift} For any $x\in \L$ we
  define the shift $\tau_x$ as the extension to
  the whole $\mathsf{A}(\L)$ of the isomorphism
  from each $\mathsf{A}_y$ to
  $\mathsf{A}_{y+x}$.
\end{definition}
If $O \in \mathsf{A}_\Lambda$ is an operator
localized in the region $\mathsf{A}_\Lambda$, then
$\tau_x(O) \in \mathsf{A}_{\Lambda + x}$ is an
operator localized in the region
$\Lambda + x := \{y+x | y \in
\Lambda \}$.
For $\L = \mathbb{Z}$ one can refer to the pictorial
representation of e.g. $\tau_{1}$:
\begin{align}
\begin{split}
\tau_{1 } = \input{Leftsh.tikz} .
\end{split}
\end{align}

We are now in position to define a Quantum Cellular Automaton. 
 \begin{definition}\label{def:QCA} 
   A $\emph{Quantum Cellular Automaton}$ (QCA)
   over a lattice $\L$ with finite neighbourhood
   scheme $\mathcal{N} \subset\L$ is an
  automorphism
   $\tT:\mathsf{A}(\L) \longrightarrow
   \mathsf{A}(\L)$ of the quasi-local
   algebra such that
   \begin{itemize}
   \item  $\tT(\mathsf{A}(\Lambda))\subset
   \mathsf{A}(\Lambda+\mathcal{N}) \hspace{5pt}
   \forall \Lambda\subset \L$ (locality),
   \item $\tT\circ\tau_x=\tau_x\circ\tT$
   $\forall x\in \L$ (homogeneity),
 \end{itemize}
 where $\Lambda+\mathcal{N} := \{ y+x | y \in
\Lambda, x \in \mathcal{N}  \}$.
   The homomorphism $\tT_0:\mathsf{A}_0\longrightarrow
   \mathsf{A}(\mathcal{N})$ given by the
   restriction of the QCA to the cell labelled $x=0$ is called the
   $\emph{transition rule}$.
\end{definition}
The identity and the shifts $\tau_x$ are (trivial)
examples of QCA. Another example of a QCA is a finite depth quantum circuit,
whose action can be decomposed in terms of local gates as follows
\begin{equation}
\resizebox{0.85\hsize}{!}{
\input{QCAexample.tikz}}.
\end{equation}
The boxes represent unitary completely positive maps, e.g.~those labelled $\mathcal U$ represent the maps $\mathcal U(X)\coloneqq UXU^\dag$.
\begin{remark} \label{rem:abstran} Notice that the algebra $\mathsf A_0$ is isomorphic to $\MdC$ (say via $\transf H_\mathcal L:\mathsf A_0\to\MdC$), and $\mathsf A(\mathcal N)$ is isomorphic to $\MdC^{\otimes|\mathcal N|}$ (via $\tK_\mathcal L:\mathsf A(\mathcal N)\to\MdC^{\otimes|\mathcal N|}$). Thus, the transition rule can be equivalently identified by the homomorphism
\begin{align}
\tT_*\coloneqq \tK_\mathcal L\circ\tT_0\circ\tH^{-1}_\mathcal L:\MdC\to\MdC^{\otimes|\mathcal N|}.
\label{eq:isoabs}
\end{align}
In this way, the restriction of $\tK_\mathcal L$ to $\tT_0(\mathsf A_0)$ is an isomorphism to $\tT_*(\MdC)$.
\end{remark}

It is possible to prove (\cite{schumacher2004reversiblequantumcellularautomata}) that the
local transition rule and the global evolution are
in one-to-one correspondence. This result is
stated in the following Lemma:
\begin{lemma}\label{lem:globloc}
  A homomorphism
  $\tT_0 : \mathsf{A}_0 \longrightarrow
  \mathsf{A}(\mathcal{N})$ is the transition rule
  of a cellular automaton if and only if for all
  $x \in \L$ such that
  $\mathcal{N}\cap\mathcal{N}_x\neq \emptyset$,
  $\mathcal N_x := \{x\}+\mathcal N$, the
  algebras $\tT_{0}(\mathsf A_0)$ and
  ${\tau}_x(\tT_{0}(\mathsf A_0))$ commute
  element-wise.
  Given a local transition rule $\tT_0$
  the global homomorphism $\tT$
  is given by
\begin{equation}
    	\tT(\bigotimes_{x\in \mathcal L}A_x)=\prod_{x\in \mathcal L}\tT_x(A_x)
    	\label{eq:autom}
\end{equation}
where
$\tT_x({A}_x)={\tau}_x\tT_0{\tau}_{-x}({A}_x)$.
 \end{lemma}

 Since the image $\tT_0(\mathsf A_0)$ is a
 subalgebra of the full matrix algebra
 $\mathsf A_\mathcal N$, from the characterisation
 theorem of
 finite dimensional $C^*$-algebra (see
 e.g.~\cite{schumacher2004reversiblequantumcellularautomata})
we have
\begin{equation}
    \tT_0(X)=U^\dag(I_{\mathcal N\setminus \{0\}}\otimes X)U,\quad\forall X\in\mathsf A_0
    \label{eq:ev}
\end{equation} 
for a suitable unitary $U$ in $\mathsf A_{\mathcal
  N\cup\{0\}}$, and analogously for $\tT_*$. In the following we will use 
  eq.~\eqref{eq:ev} both for $\tT_0$ and $\tT_*$, leaving to the context the specification 
  as to which is the case.
\subsection{The wrapping Lemma}
Our definition of QCA encompasses both 
infinite lattice and finite lattices embedded
in a torus. However, a result known as Wrapping Lemma (see Ref. \cite{Perinotti2020cellularautomatain}) relates the 
evolution of the QCAs over the infinite lattice
with the one over a finite one. The idea is
to ``wrap'' a sufficiently large finite
sub-lattice of a QCA over the infinite
lattice $\mathbb{Z}^s$ by imposing suitably periodic
boundary conditions, obtaining a discrete torus. 
The ``unwrapped'' QCA and the ``wrapped'' one then have
the same local transition rule. The concept of
``sufficiently large'' is grasped by the
requirement that, after wrapping, still
$\tT_0(A)$ should commute with $\tT_x(B)$ for all $A\in\mathsf A_0$ and $B\in\mathsf A_x$, 
for any $x$ such that $\mathcal{N}\cap\mathcal{N}_x\neq\emptyset$.
Therefore, a ``sufficiently large'' wrapping is such that the
set of elements $x\in\L$ such that
$\mathcal{N}\cap\mathcal{N}_x\neq\emptyset$ is the
same as in the infinite case, i.e. no new
overlapping between neighbourhoods is introduced
after wrapping the lattice.
Since we will heavily exploit the Wrapping Lemma in the next sections, 
we now provide its formal statement. However, a reader who is already familiar with this concept can easily skip this section.
Let us begin defining a wrapping,
i.e. a homomorphism of the infinite lattice $\mathbb Z^s$ into
a discrete torus.
\begin{definition}[Wrapping]
A wrapping of the additive group $\mathbb Z^s$ is a homomorphism 
\begin{align*}
\varphi:\mathbb Z^s\to\mathcal L\coloneqq\bigtimes_{t=1}^s\mathbb Z_{f_t}, \hspace{5pt} f_t<\infty \hspace{5pt} \forall t,
\end{align*}
where $\mathbb Z_f$ stands for the additive group 
$\mathbb Z$ modulo $f$.
\end{definition}
More concretely, if an element $x\in \mathbb Z^s$ is identified by the s integer numbers $x=(x_1,...,x_s)$ the action of the wrapping is:
\begin{align*}
\varphi(x)=((x_1\Mod{f_1}) ,..., (x_s\Mod{f_s})).
\end{align*}
We now define a regular neighbourhood scheme for a wrapping $\varphi$,
that pins down the idea of a ``sufficiently large'' wrapping.
\begin{definition}[Regular neighbourhood scheme]\label{def:regnei}
Let $\mathcal N\subset\mathbb Z^s$ be a finite set. We say that $\mathcal N$ is \emph{regular} for the wrapping $\varphi$ if 
\begin{align*}
&\{\varphi^{-1}(\mathcal N_{\varphi(x)})\cap \varphi^{-1}(\mathcal N_{\varphi(y)})\}\cap
\{\mathcal N_x\cup \mathcal N_y\}=\\
&=\mathcal N_x\cap \mathcal N_y 
\end{align*}
for every pair $x,y\in\mathbb Z^s$ such that $\mathcal N_x\cap\mathcal N_y\neq\emptyset$.
\end{definition}

The above definition captures the idea that the neighbourhoods of close cells $x$ and $y$,
after wrapping, should have the same intersection as before wrapping. The intersection with $\mathcal N_x\cup \mathcal N_y$ on the r.h.s.~in the definition is required because the inverse image of $\mathcal N_{\varphi(x)}\cap \mathcal N_{\varphi(y)}$ contains infinitely many intersections, due to periodicity of the set $\varphi^{-1}(x)$. 
Now, we can finally understand the meaning of ``sufficiently large'' wrapping, with respect to a neighbourhood scheme: the periodic boundary conditions should be such that the neighbourhood scheme is regular. For a single step of evolution of a nearest neighbour automaton, this is true for all lattices with more than four cells in each coordinate direction.

With the above notions in mind, we can now state the wrapping lemma.
\begin{lemma}[Wrapping Lemma]\label{lem:wrap}
QCAs with a regular neighbourhood scheme $\mathcal N$ on a lattice 
$\mathcal{L}$ with given periodic boundary conditions are in one-to-one correspondence 
with those on the infinite lattice, and the correspondence is defined by ``having the same 
transition rule''.
\end{lemma}
\begin{remark}
Notice that we did not give a formal definition of the (equivalence) relation ``having the same transition rule''. However, this notion is straightforward if we think of the transition rule as $\tT_*$ rather than $\tT_0$.
\end{remark}

Given a wrapping $\varphi:\mathbb Z^s\to\mathcal L$, let us now introduce the following isomorphisms, referring to those introduced in Remark~\eqref{rem:abstran}
\begin{align*}
&\mathcal{U}_\varphi\coloneqq\tH^{-1}_{\mathbb Z^s}\circ\tH_\mathcal L:\mathsf{A}_0^w\rightarrow\mathsf{A}_0^\infty \\
&\mathcal{W}_\phi\coloneqq\tK^{-1}_\mathcal L\circ\tK_{\mathbb Z^s}:\mathsf{A}_{\mathcal{N}}^\infty\rightarrow\mathsf{A}_{\mathcal{N}}^w.
\end{align*}
The map $\mathcal{U}_\varphi$ maps the single cell algebra $\mathsf{A}^w_0$ of the wrapped lattice and into the single cell algebra $\mathsf{A}^\infty_0$ of the algebra of the infinite lattice, while the map $\mathcal{W}_\varphi$ maps the algebra of the neighbourhood $\mathsf{A}^\infty_{\mathcal{N}}$ of the infinite lattice into that of the wrapped lattice denoted $\mathsf{A}_{\mathcal{N}}^w$. 
The regularity requirement in def.~\eqref{def:regnei} ensures that the application of 
$\varphi$ to $\mathcal{N}_x$ does not change the neighbourhood structure: $\varphi(\mathcal N_x)=\mathcal N_{\varphi(x)}\leftrightarrow\mathcal N_x$. 

If we now consider a QCA $\tT$ over $\mathbb{Z}^s$, with transition rule $\tT^{(\infty)}_0$, we can define its wrapped version as follows.
\begin{definition}[Wrapped evolution]\label{def:wrapped} The wrapped evolution $\tT_w$ on the wrapped lattice $\mathcal{L}=\varphi(\mathbb Z^s)$ of a QCA $\tT$ over $\mathbb{Z}^s$ is the QCA over the lattice $\mathcal{L}$ with local rule given by:
\begin{align} \label{eq:wrapping}
\tT^{(w)}_0=\mathcal{W}_\varphi\tT^{(\infty)}_0\mathcal{U}_\varphi.
\end{align}
\end{definition} 
Notice that, in this case,
\begin{align*}
\tK_\mathcal L\circ\tT^{(w)}_0\circ\tH_{\mathcal L}^{-1}=\tK_{\mathbb Z^s}\circ\tT^{(\infty)}_0\circ\tH_{\mathbb Z^s}^{-1}=\tT_*.
\end{align*}
Consequently, we can say that wrapped QCAs over different lattices share the same local rule if their local rule can be described as in eq.~\eqref{eq:wrapping} starting from the same infinite QCA $\tT$, namely $\tT^{(w)}$ and $\tT^{(w')}$ share the same local rule if:
\begin{align*}
\tK_\mathcal L\circ\tT^{(w)}_0\circ\tH_{\mathcal L}^{-1}=\tK_{\mathcal L'}\circ\tT^{(w')}_0\circ\tH_{\mathcal L'}^{-1}=\tT_* .
\end{align*}

 \subsection{One dimensional QCA: Index Theory}
 One key feature of QCAs on $\mathbb Z$ is their
 \emph{index}, which was introduced in
 Ref. \cite{Index} for QCAs without the
 translation-invariance hypothesis. The index is a
 locally computable invariant, meaning it can be
 determined from a finite portion of the automaton
 and will yield the same result regardless of the
 portion chosen, even if the QCA is not
 translation-invariant.
 Intuitively, the index theory for one dimensional
 QCAs suggests that information can be
 thought as an incompressible fluid, with the index
 representing its flow rate.

 Let us now review the main results of index
 theory.  The whole construction is based on the
 notion of support algebras~\cite{Index} (also called interaction algebras~\cite{PhysRevA.63.012301}).
\begin{definition} [Support algebra]
  Let $\mathsf B_1$ and $\mathsf B_2$ be
  $C^*$-algebras, and consider a $*$-subalgebra
  $\mathsf{A}\subset \mathsf B_1\otimes \mathsf
  B_2$. Each element $a\in\mathsf{A}$ can be
  expanded uniquely in the form
  $a=\sum_{\mu}a_\mu\otimes e_{\mu}$, where
  $\{e_{\mu}\}$ is a fixed basis of $\mathsf
  B_2$. The \emph{support algebra}
  $\mathsf S(\mathsf A,\mathsf B_1)$ of
  $\mathsf A$ in $\mathsf B_1$ is the smallest
  $C^*$-algebra generated by all $a_\mu$ in such
  an expansion \footnote{It is easy to prove that the support algebra
    does not depend on the choice of the basis $\{e_\mu\}$}.
\end{definition}

Let us now consider a nearest neighbour QCA on $\mathbb Z$,
i.e.~with $\mathcal N_x = \{ x-1,x,x+1\}$ for all
$x$. Two comments are now in order. First,
any QCA on $\mathbb Z$ can be expressed as a nearest-neighbour QCA
by grouping multiple cells in a single one. In the following
discussion we will 
assume that all the QCAs that we introduce are
nearest-neighbour QCAs.
Second, since
translation invariance is not assumed, local
algebras at different sites may not be
isomorphic.  Let us now define the left and right
support algebras around cells $2x$ and $2x+1$ for a QCA with evolution $\tT$ as
\begin{equation}
\begin{split}
&\mathsf{L}_{2x}= \mathsf{S}(\tT(\mathsf{A}_{2x} \otimes \mathsf{A}_{2x+1}),(\mathsf{A}_{2x-1}\otimes \mathsf{A}_{2x})),\\
&\mathsf{R}_{2x}= \mathsf{S}(\tT(\mathsf{A}_{2x} \otimes \mathsf{A}_{2x+1}),(\mathsf{A}_{2x+1}\otimes \mathsf{A}_{2x+2})).
\end{split}
\end{equation}
The first result of index theory is the
following theorem.
\begin{proposition}[Index of a QCA] For any 
  QCA $\tT$ on $\mathbb Z$, the quantity 
\begin{equation}
\ind(\tT)=\sqrt{\frac{\operatorname{dim}[\mathsf{L}_{2x}]}{\operatorname{dim}[\mathsf{A}_{2x}]}}
\end{equation}
does not depend on $x$.
We call $\ind(\tT )$ the \emph{index} of the QCA.
\end{proposition}

The following Proposition presents the most
important consequences of index theory.
\begin{proposition}
  \label{prp:indextheorymainprop}
  For any
  QCAs $\tT$ and $\tS$ we have that
$\ind[\tT\tS] = \ind[\tT]\ind[\tS]$ and
 $\ind[\tT] = 1$ $\textit{iff}$ $\tT$ has a
  \emph{Margolus partitioning scheme}, 
namely
\begin{equation}
\begin{split}
&\tT = \mathcal{W} \circ \mathcal{V}\\
&\mathcal{V}(\cdot) = (\prod_{n\in\mathbb{Z}}V^\dag_{2n})(\cdot)(\prod_{m\in\mathbb{Z}}V_{2m})\\
&\mathcal{W}(\cdot) = (\prod_{k\in \mathbb{Z}} W^\dag_{2k+1})(\cdot)(\prod_{l\in\mathbb{Z}}W_{2l+1})
\end{split}
\end{equation}
where $V_{2n}$ are unitaries acting on sites $2n$ and $2n + 1$, and $W_{2m+1}$ are unitaries acting on sites $2m+1$ and $2m+2$.
\end{proposition}

 \subsection{One Dimensional Nearest Neighbour Qubit Automata}\label{NNQCA}
 We now restore the translation invariance
 assumption. It is possible to show
 \cite{schumacher2004reversiblequantumcellularautomata} that a nearest-neighbor (i.e.
 $\mathcal{N}=\{-1,0,1\}$) QCA $\tT$ on
 $\mathbb{Z}$ where each cell is a qubit, has
 $\ind( \tT ) = 1$ if and only if
 \begin{equation}
   \label{eq:10}
\tT=\input{Qubitev.tikz}
\end{equation}
with $\mathcal U(X)\coloneqq U^\dag X U$, $U$ an arbitrary unitary gate acting
on a single cell, and $\mathcal C_\phi(X)\coloneqq C_\phi^\dag XC_\phi$, where $C_\phi$ is 
a controlled-phase gate, i.e.
\begin{align*}
C_\phi=\begin{pmatrix}
1&0&0&0\\
0&1&0&0\\
0&0&1&0\\
0&0&0&e^{i\phi}\\
\end{pmatrix}. 
\end{align*}
The expression~\eqref{eq:ev} of the local rule of
the QCA is
\begin{equation}
\begin{split}
    &\tT_0({A}_0)=X^{\dag}(I\otimes U^{\dag}A_0 U\otimes I)X,\\
    &X=(C_\phi\otimes I_3)(I_1\otimes C_\phi).
\end{split}
    \label{eq:x}
\end{equation}
From now on, we will set the basis in which $C_\phi$ is diagonal as the basis $\{\ket{a}\}_{a=0,1}$ such that $\sigma_z\ket{a}=(-1)^{a}\ket{a}$. We will refer to this basis as the \textit{computational basis}.

\section{QCA renormalisation}\label{sec:CG}
Following the idea presented in \cite{PhysRevE.73.026203} we
now define a procedure to provide a coarse-grained
description of the dynamics of a QCA, in such a
way that the large scale description that we
eventually obtain is still a QCA. Intuitively
speaking, the procedure consists in choosing a
subset of degrees of freedom of the original
evolution through a projection on a subspace of
the Hilbert space of $N$ neighbouring cells, and
requiring that the restriction of $N$ steps of the
original QCA to such a subspace is still a QCA.
Naively speaking, one could try to restrict the domain of
the QCA to the operators whose support is
contained in the support of
$\bigotimes_y \Pi_y$, where each $\Pi_y$ identifies a subset of degrees of freedom of $\mathsf{A}_{2y}\otimes\mathsf{A}_{2y+1}$. However, this infinite
tensor product is ill-defined, lying outside the
quasi local algebra. We overcome such problem exploiting the
wrapping lemma to reduce the coarse-graining
problem of an infinite lattice to the corresponding problem 
over all possible wrappings $\mathcal L$. In order to do that, 
we have to define a partition
$\{\Lambda_x\}_{x\in\mathcal{L}'}$ of the original
lattice $\mathcal{L}$ . The indices $x$ of such
a partition will be the cells of a new lattice
$\mathcal{L}'$: this will establish the relation
between the wrapped lattice $\mathcal{L}$
and the coarse-grained one $\mathcal{L'}$. Once
this is done, we have a formulation of the
problem in terms of coarse-graining of the family of
evolution rules over all the possible wrapped lattices
$\mathcal{L}$. We will then find conditions for the coarse-graining
to exist, and show that these
conditions do not depend on the choice of the
lattice $\mathcal{L}$. This means that our
coarse-graining condition ultimately boils down to
a single condition over a minimal regular lattice
$\mathcal{L}$. In particular, in the first
subsection we will introduce the technical tools
to properly formulate the problem of coarse-graining 
a finite lattice $\mathcal{L}$ in another
finite lattice $\mathcal{L'}$, while in the second
subsection we will apply these tools to the study
of the coarse-graining condition for a  QCA over an infinite lattice.

A QCA that admits a solution to the coarse-graining equation is \emph{renormalisable}. The
coarse-graining of QCAs determines a displacement in the space of QCAs as specified by 
their defining parameters, that represents a \emph{renormalisation flow}, whose fixed 
points play a distinguished role as in any renormalisation scenario.

\subsection{Coarse-grained algebra}
Let $\tT$ be a QCA on $\mathbb Z^s$ such that the
algebra of each cell is isomorphic
to the $C^*$-algebra of $d\times d$ complex
matrices, and let us fix some $N \in \mathbb{N}$.
Let us now consider a regular
tessellation  of $\mathbb Z^s$ such that each tile 
is a $Ns$-dimensional cube i.e.
\begin{align}
&\bigcup_{{x}\in \mathbb{Z}^s}\Lambda_{x}= \mathbb{Z}^s,\label{eq:part}\\
&\Lambda_0\coloneqq [0,N-1]^s,&&\Lambda_{{x}}\coloneqq\Lambda_0+Nx,\label{eq:transpal}
\end{align}
where $[a,b]$ denotes the set
$\{a,a+1,\ldots,b-1,b\}\subset\mathbb Z$ (see Fig. \eqref{fig:Tess} ). 
\begin{figure}
\centering
(a) \resizebox{0.25\paperwidth}{!}{
  \input{Lattice.tikz}}
\\
\medskip
(b) \resizebox{0.25\paperwidth}{!}{
   \input{LatticePart.tikz}}
\\
\medskip
(c) \resizebox{0.2\paperwidth}{!}{
   \input{Latticeprimbis.tikz}}
        \caption{Example of tessellation of $\mathcal{L}$ and mapping in $\mathcal{L'}$ in $\mathbb{Z}^2$. (a) The wrapped lattice $\mathcal{L}\simeq\mathbb Z_6\times\mathbb Z_6$. (b) Tessellation $\mathcal L'\simeq \mathbb Z_3\times\mathbb Z_3$ of the lattice $\mathcal{L}$ through $\{\Lambda_x\}$ with $\Lambda_{0,0}=[0,1]^2$. (c) The new lattice $\mathcal{L'}$ is given by the identification $\Lambda_{\vc{x}}\rightarrow \vc{x}=(x,y)$.}
\label{fig:Tess}
\end{figure}
 From
eq.~\eqref{eq:transpal} we have that the cubes
$\{\Lambda_{{x}}\}$ are related to each other by a
translation.  Consider now a projection $\Pi_y$ in
the algebra $\mathsf{A}(\Lambda_y)$ of each of the
cubic tiles. As already noted, the projector $\bigotimes_y \Pi_y$ lies outside the quasi local algebra.
We then have to exploit the wrapping
lemma and consider the wrapped evolution $\tT_w$
on a finite lattice
\begin{align}
  \label{eq:2}
\mathcal{L} = \mathbb{Z}_{Nf}^s \quad f \in \mathbb{N}  
\end{align}
where $f$ is sufficiently large such
that the regular neighborhood condition is
satisfied. Clearly,  the lattice $\mathcal{L}$ admits a tessellation in terms of
the same $s$-dimensional hypercubes. Then we have
\begin{align}
  \label{eq:1}
  \bigcup_{{y}\in \mathcal{L}'}\Lambda_{y}=
  \mathcal{L}, \quad \mathcal{L}' :=
  \mathbb{Z}_{f}^s
\end{align}
If $\Pi_{{0}}$ is a given a projection belonging to
the local algebra of the tile $\Lambda_0$, we can
define the projection $\Pi_{{{x}}}$ on the tile $\Lambda_x$ by translating 
$\Pi_{{0}}$, namely $ \Pi_{x}\coloneqq\tau_{Nx}(\Pi_{0})$.
We now define the coarse-grained algebra $\mathsf A\vert_\Pi(\mathcal{L})$  over $\mathcal{L}$  
as follows: 
\begin{align}
  \begin{aligned}
&\mathsf
A\lvert_{\Pi}(\mathcal{L})\coloneqq
\{ A \in \mathsf{A}(\mathcal{L}) \mbox{ s.t. }
\Pi A \Pi = A\} , \\
&\Pi\coloneqq\bigotimes_{x\in \mathcal{L'}}\Pi_{x}.
\label{eq:restop}    
  \end{aligned}
\end{align}
i.e. $\mathsf A\vert_\Pi(\mathcal{L})$ is the
algebra of operators with support and range
contained in those of $\Pi$. The algebra
$\mathsf A\vert_\Pi(\mathcal{L})$
is isomorphic to the
algebra $\mathsf A(\mathcal{L'})$ on the lattice
$\mathcal{L'}$ where each local system has
dimension $d'\coloneqq\operatorname{rank}(\Pi)$.
The node $x\in \mathcal L'$ corresponds to
the tile $\Lambda_x$ in $\mathcal L$, and operators
$A\in\mathsf A(\mathcal{L'})$ are in one-to-one
correspondence with operators in
$\mathsf A\lvert_{\Pi}(\mathcal{L}) $.
Let us explicitly construct such isomorphisms.
Let $\{\ket{\psi_k}\}$ be a basis on the
Hilbert space of the tile $\Lambda_x$ such that 
\begin{equation}
\Pi_{{{x}}}=\sum_{k=0}^{d'-1} \ket{\psi_k}\bra{\psi_{k}}.
\end{equation}
We then define
\begin{equation}
\J\coloneqq\Jopconj{ }, \quad \Jop\coloneqq\sum_{k=0}^{d'-1}\ket{k}\bra{\psi_k},
\label{eq:Jop}
\end{equation}
where $\{\ket{k}\}_{k=0}^{d'-1}$ is a basis of
$\mathbb{C}^{d'}$, i.e. the Hilbert space of the
cell $x \in \mathcal{L}'$.  This operator applies
the $\Jop$ over the support of each $\Lambda_x$.
It is easy to verify that
\begin{align}
  \label{eq:3}
  \Jop^\dag\Jop=\Pi_{x},\; \Jop\Jop^\dag= I,
  \;
\J^\dag\J=\Pi, \; \J\J^\dag=I.
\end{align}
Then we define
\begin{align}
  \begin{aligned}
&\tV:\Aredloc \to \mathsf{A}(\mathcal{L})
               ,\quad \tJ: \mathsf{A}(\mathcal{L}) \to\Aredloc , \\
&\tV(A_r) := J^\dag A_r J \\
 & \tJ(A):= J \mathcal{P} (A) J^\dag,
  \quad
   \mathcal{P} (A):= \Pi A \Pi    
  \end{aligned}
   \label{eq:Jiso3}
\end{align}
and from Equation \eqref{eq:3} we have
\begin{align}
\tJ\circ\tV=\tI, \quad
\tV\circ\tJ=\tP
\end{align}
i.e. $\tV$ is an
isomorphism between
$\Aredloc$ and its image $\Aplocc$ and  the
restriction of $\tJ$ to
$\Aplocc$ is the inverse isomorphism.
\begin{figure}
\input{J.tikz}
\caption{Representation of the maps $\tJ, \tV$.}
\end{figure}
With the choice of basis in which each $\Pi_x$ is diagonal, the
matrices representing operators
$A\in\mathsf A\lvert_{\Pi}(\mathcal{L})$, are
obtained by padding with zeros the matrices that
represent their correspondent
$A_r\in\mathsf A(\mathcal{L'})$, and conversely,
given the matrix of an operator
$A\in\mathsf A\lvert_{\Pi}(\mathcal{L})$,
the corresponding operator on
$A_r\in\mathsf A(\mathcal{L'})$ is obtained
by removing the
null blocks with image or support on the kernel of
$\Pi$. In other words, we have:
\begin{align}
A_r\in \mathsf{A}(\mathcal{L'}) \leftrightarrow \begin{pmatrix}A_r &0\\0&0\end{pmatrix}\in \mathsf A\lvert_{\Pi}(\mathcal{L}).
\end{align}

We can now define the map $(\tT)_r$ induced by $\tT_w$ on $\Aredloc$ as follows
\begin{definition}[Induced map]
  Let $\tT$ be a QCA on $\mathbb{Z}^s$
  and let $\mathcal{L} = \mathbb{Z}_{Nf}^s$,
   be a wrapped lattice such that the regular
   neighbourhood condition is satisfied.
   Let $\mathcal{L}' = \mathbb{Z}_{f}^s$
   and let $\tJ$ and $\tV$ be defined as in eq.~\eqref{eq:Jiso3}.
   The \emph{induced map} on the lattice
   $\mathcal{L'}$ is defined as follows:
\begin{align}
  \begin{aligned}
  &\tT_r:\Aredloc \to\Aredloc ,\quad
\tT_r\coloneqq \tJ\circ\tT_{w}\circ\tV.    
  \end{aligned}
\label{eq:induced}
\end{align}
\end{definition}
Notice that the induced map depends on the choice of the
lattices $\mathcal{L}$ and $\mathcal L'$, and we have a
\textit{family} of induced maps
$\{\tT_r\}_{\mathcal L,\mathcal{L'}}$.  In order to lighten
the notation, we will not explicitly write
dependance of the induced map on the lattices
$\mathcal{L}$ and $\mathcal{L}'$. This simplification is also justified by
the fact that there is no
dependence of the local rule on the choice of the
lattices, as we will prove later.
We remark that, in general, we need to differentiate between the map induced by $k$ 
steps of the QCA $\tT_w$ and $k$ subsequent applications of the map induced by $\tT_w$. 
The above two transformations of the quasi-local algebra are given by
\begin{align}
&(\tT^k)_r= \tJ\circ\tT_w^k\circ\tV && (\tT_r)^k= (\tJ\circ\tT_w\circ\tV)^k,
\end{align}
respectively.
In the following we will see how the induced map pins down the idea of coarse-graining.

\subsection{Renormalisable QCAs}

We are now in position to provide a definition of renormalisable QCA.
\begin{definition}[Renormalisable QCA]
  Let $\tT$ and $\tS$ be two QCAs on $\mathbb{Z}^s$,
  $N$ be a natural number,
  and $\Pi_0 \in \mathsf{A}(\Lambda_0)$ be a
  projection on the local algebra of the region
  $\Lambda_0 \coloneqq [0,N-1]^s$. 
  We say that $\tT$  is 
\emph{renormalisable through the coarse-graining} $(\Pi_0,\Lambda_0)$
to $\tS$  if
\begin{equation}
\tT^N_{w}\circ \tV=\tV\circ\tS_w
\label{eq:cg}
\end{equation}
holds for any wrapping $\tT_w$ on $\mathcal{L} = \mathbb{Z}_{Nf}^s$
  such that the neighbourhood of $\tT^N$ is regular over $\mathcal{L}$.
  We say that $\tS$ is a size-$N$ \emph{renormalisation} of $\tT$ through the 
  coarse-graining $(\Pi_0,\Lambda_0)$. We will also  say that  
  $\tT$ \emph{admits a size-$N$ renormalisation} if there exists
  a QCA $\tS$ and a coarse-graining $(\Pi_0,\Lambda_0)$ such that $\tS$ is a size-$N$ 
  renormalisation of $\tT$ through $(\Pi_0,\Lambda_0)$.
\end{definition}
\begin{remark}
Notice that this definition involves the map $\tV$, whose role is to embed the single cell algebra into the N-cell algebra. This definition may seem in contrast with the one in~\cite{PhysRevE.73.026203}, where the coarse-graining map takes the state of the $N$-supercell and projects it into the state of a single cell. However, we have to remember that we are working in the Heisenberg picture here. The above definition represent the faithful transposition of that in~\cite{PhysRevE.73.026203} to the quantum case.
\end{remark}

  \begin{figure}[!]
\resizebox{0.45\textwidth}{!}{
  \begin{tikzpicture} [remember picture]
  \node (A) at (-2,2.25) {\includegraphics[scale=0.2]{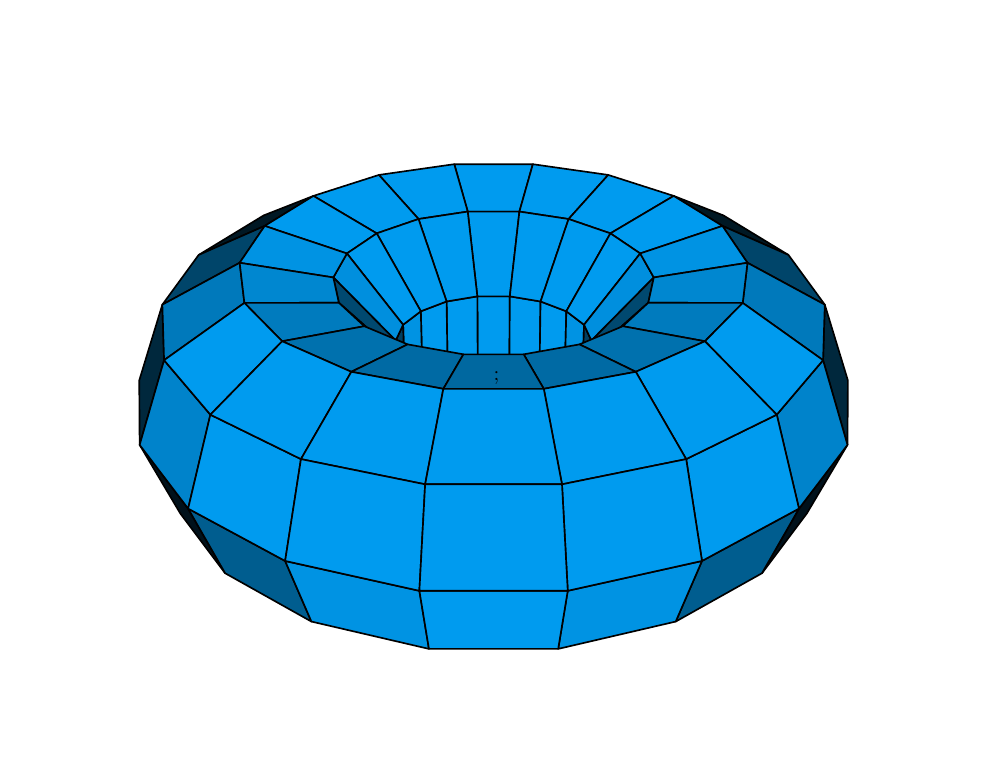}};
  \node (B) at (-2,-2.25) {\includegraphics[scale=0.2]{toruscg}};
  \draw[->,line width=2 pt] ($(A)-(0,1.25)$) to node[anchor=east] {\huge $\tS$} ($(B)+(0,1)$) ;
  \node (C) at (3.2, 2.25 ) {\includegraphics[scale=0.2]{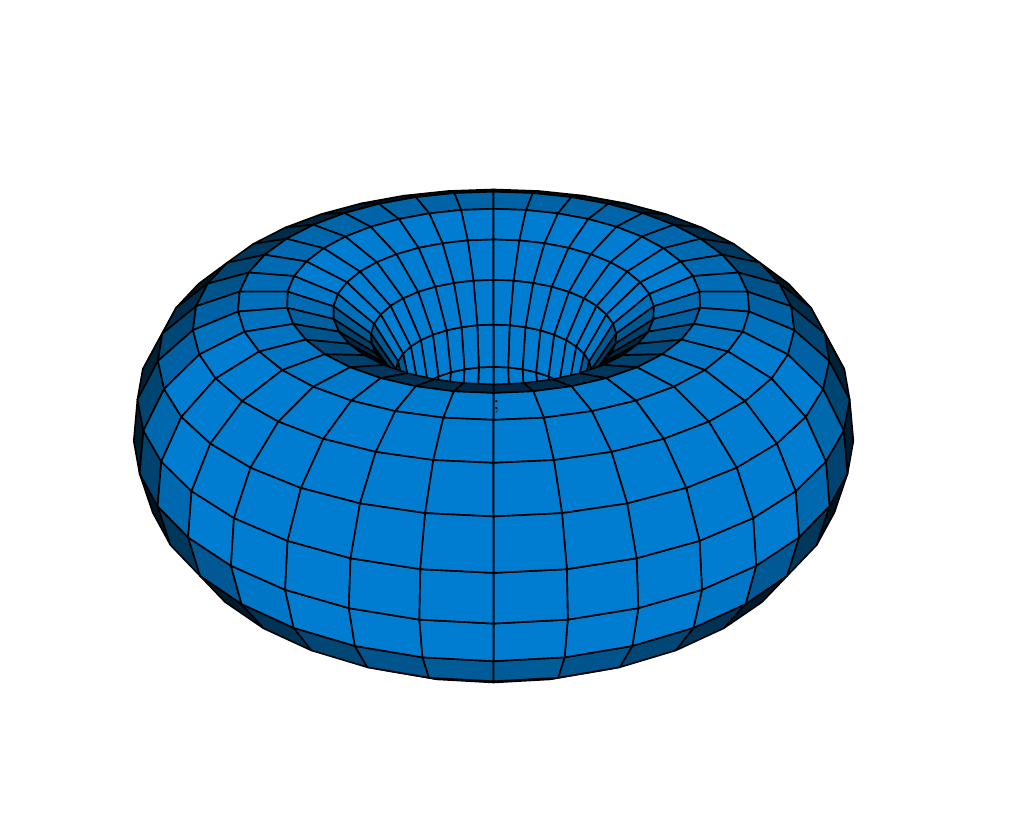}};
  \node (D) at (3.2,-2.25) {\includegraphics[scale=0.2]{torus}};
  \draw[->,line width=2 pt] ($(C)-(0,1.25)$) to node[anchor=west] {\huge $\tT^N$} ($(D)+(0,1)$) ;
  \draw[->,line width=2 pt]  ($(A)+(1.75,-0.25)$) to node[anchor=south] {\huge $\tV$}  ($(C)-(1.75,0.25)$) ;
  \draw[->,line width=2 pt]  ($(B)+(1.75,-0.25)$) to node[anchor=north] {\huge $\tV$} ($(D)-(1.75,0.25)$);
  \end{tikzpicture}}
  \caption{The condition in eq.~\eqref{eq:cg} amounts to ask the commutativity of the above diagram: if $\tS$ is a renormalisation of $\tT$ then evolving with a single step of $\tS$ in the coarse-grained algebra and then embed the coarse-grained algebra in the original one through the map $\tV$ is the same than embed the coarse-grained algebra in the original one and consider N steps of $\tT$ . In this way, a single step of $\tS$ mimics the behaviour of $\tT^N$ if reduced to the chosen degrees of freedom.}
  \end{figure}

This definition is modelled on the analogous one
from~\cite{PhysRevE.73.026203}.
However it is
not particularly amenable for the analysis of the
problem due to the reference to all possible wrappings $\tT^N_w$. The following lemmas 
will show that it is sufficient to consider just one wrapping and check the renormalisation condition there. Z
\begin{lemma}\label{lem:CgT}
  If a QCA $\tT$ on $\mathbb{Z}^s$ is 
renormalisable through $(\Pi_0,\Lambda_0)$ to some QCA $\tS$ as in eq.~\eqref{eq:cg},
then every induced map $(\tT^N)_r$ is a QCA. 
\end{lemma}
\begin{proof}
Let us begin remembering the expression of the induced map of N-steps of $\tT$:
\begin{equation}
(\tT^N)_r= \tJ\circ\tT^N_{w}\circ\tV.
\label{eq:inducedN}
\end{equation}
Composing \eqref{eq:cg} on the left with $\tJ$ we have
\begin{equation}
\tJ\circ\tT^N_{w}\circ\tV=\tJ\circ\tV\circ\tS_w=\tS_w,
\end{equation}
which reads
\begin{equation}
\tS_w=(\tT^N)_r,
\end{equation}
for every possible choice of the lattice $\mathcal{L}$.
This means that, if $\tS$ is a renormalisation of $\tT$, then $(\tT^N)_r$ is a QCA. 
\end{proof}

Notice that, if $(\tT^N)_r$ is a QCA for every
choice of the lattice $\mathcal{L}$, the Wrapping
Lemma allows us to extend it to a QCA $\tS$ on the
\textit{quasi-local} algebra over the infinite
lattice $\Aqloc s$. However, the local rule of $\tS$ may depend on the specific wrapping. 
In the following, we prove that this is not the case.

Let us now study when the
condition \eqref{eq:cg}  holds for a given
wrapped lattice $\mathcal{L}$. First of all, since $\tT_w$
is a $C^*$-algebra automorphism, 
we can
always find a unitary $U_\tT$ such that
\begin{equation}\label{eq:evunit}
\tT_w(A)=\Tmat^\dag A \Tmat \hspace{10pt} \forall A\in \mathsf{A}(\mathcal{L}),
\end{equation}
and the same holds for $\tS_\mathcal{L'}$:
\begin{equation}
\tS_w(A)=(U_\tS)^\dag A (U_\tS)\hspace{10pt} \forall A\in \mathsf{A}(\mathcal{L'}).
\end{equation}
Moreover, eq.~\eqref{eq:inducedN} allows us to write:
\begin{align}
&\Tred (A)=\Tredmat^\dag A \Tredmat\label{eq:tred}\\
&\Tredmat \coloneqq \J  \Tnmat \J^\dag \label{eq:tredmat}
\end{align}
Let us begin providing some trivial examples of renormalisation to develop the intuition.
\begin{itemize}
\item \textbf{Local Unitaries}: Let $\tT$ be a QCA
  over $\mathbb{Z}$ whose effect is to apply the
  same local transformation $U_x$ on each cell
  $x$, namely:
\begin{equation*}
\tT(\mathsf{A}_{\Lambda})=\left(\bigotimes_{x\in\Lambda} U^\dag_x\right)\mathsf{A}_{\Lambda}\left(\bigotimes_{y\in\Lambda}U_y\right).
\end{equation*}
Wrapping $\tT$ over a lattice $\mathcal{L}$ we obtain 
the induced map:
\begin{align*}
\begin{aligned}
&\Tredmat=\J \left(\bigotimes_{x\in\L} U^N_x \right) \J^\dag=\\
&=\left(\bigotimes_{x\in\mathcal{L}'}\J_{\Lambda_x} \right) \left(\bigotimes_{x\in\L} U^N_x \right) \left(\bigotimes_{y\in\mathcal{L}'}\J_{\Lambda_y} \right)^\dag.
\end{aligned}
\end{align*}
 Since the
transformation is completely factorized, we can
evaluate the expression over the support of a
single tile $\Lambda_x$, i.e.:
\begin{align*}
\begin{aligned}
&(\Tredmat)_x=\\
&=\J_{\Lambda_x} \left(\bigotimes_{x'\in\Lambda_x} U^N_{x'} \right) \J^\dag_{\Lambda_x}=\sum_{k,k'} u_{k,k'}\ketbra{k}{k'},
\end{aligned}
\end{align*}
with 
\begin{align*}
u_{k,k'}\coloneqq \bra{\psi_k^{(x)}}\left(\bigotimes_{x'\in\Lambda_x} U^N_{x'} \right)\ket{\psi_{k'}^{(x)}}.
\end{align*}
The total matrix is then $\Tredmat=\bigotimes_{x\in\L'}(\Tredmat)_x$.
In particular, if we choose the $\{\ket{\psi^{(x)}_k}\}_k$ to be eigenvectors of $\left(\bigotimes_{x\in\Lambda_x} U^N_x \right)$ we have:
\begin{equation*}
(\Tredmat)_x=\sum_k u_{k,k}\ketbra{k}{k}.
\end{equation*}
That is, the renormalised evolution is still a local unitary with eigenvalues given by a subset of the eigenvalues of $\left(\bigotimes_{x\in\Lambda_x} U^N_x \right)$ and eigenvector corresponding to the chosen basis $\{\ket{k}\}_{k=0}^{d'-1}$. A pictorial representation of the size-2 renormalisation is shown in fig.\eqref{fig:CG}.
\item \textbf{Shift}: Consider the shift $\tau_{\pm 1}$ as in Definition \eqref{def:shift}. Following the same steps as before we can compute the matrix implementing the induced map as:
\begin{equation*}
\begin{aligned}
&\Tredmat=\J \left(\bigotimes_{x\in\L} \tau_{\pm 1}^N \right) \J^\dag=\\
&=\left(\bigotimes_{x\in\mathcal{L}'}\J_{\Lambda_x} \right) \left(\bigotimes_{x\in\L} \tau_{\pm N} \right) \left(\bigotimes_{y\in\mathcal{L}'}\J_{\Lambda_{y}} \right)^\dag.
\end{aligned}
\end{equation*}
Exploiting the fact that the size of each $\Lambda_x$ is N, we have $\tau_{\pm N}\ket{\psi_{k}^{(x)}}=\ket{\psi_{k}^{(x\pm 1)}}$, i.e. $\tau_{\pm N}$ shifts the whole $\Lambda_x$ in $\Lambda_{x\pm 1}$ and, consequently, $\tau_{\pm N} \J_{\Lambda_x}=\J_{\Lambda_{x+1}}$. In this way we have:
\begin{equation*}
\begin{aligned}
&\Tredmat=\\
&=\left(\bigotimes_{x\in\mathcal{L}'}\J_{\Lambda_x} \right) \left(\bigotimes_{y\in\mathcal{L}'}\J_{\Lambda_{y+1}} \right)^\dag=\tau_{\pm 1}.
\end{aligned}
\end{equation*}
In other words, N steps of a shift in $\L$ are renormalised to a single step of a shift in $\L'$. A pictorial representation of the size-2 renormalisation is shown in fig.\eqref{fig:CG}.
\end{itemize}
\begin{figure}
\centering
(a)\resizebox{0.35\paperwidth}{!}{
\begin{tikzpicture}
  \node (A) at (-3,0)  {\resizebox{0.5\hsize}{!}{\input{LocUnit.tikz}}};
  \node (B) at (2,0) {\resizebox{0.4\hsize}{!}{\input{LocUnitprim.tikz}}};
  \draw[->,line width=2pt] ($(A)+(2.2,0)$)--($(B)-(1.8,0)$);
\end{tikzpicture}}
(b)\resizebox{0.35\paperwidth}{!}{\begin{tikzpicture}
  \node (A) at (-3,0)  {\resizebox{0.5\hsize}{!}{\input{Leftshtwo.tikz}}};
  \node (B) at (2.1,-0.1) {\resizebox{0.4\hsize}{!}{\input{Leftshmod.tikz}}};
  \draw[->,line width=2pt] ($(A)+(2.3,0.2)$)--($(B)-(1.9,-0.3)$);
\end{tikzpicture}}
        \caption{Trivial examples of size-2 renormalisation over one dimensional lattices. (a) Example of renormalisation of a local transformation. Here we have $\Lambda_x=\{2x,2x+1\}$, two steps of the unitary matrix $U$ over two cells are mapped in a single application of the new matrix $U'$. (b) Example of renormalisation of two steps of a right shift with $\Lambda_x=\{2x,2x+1\}$. After two steps, all the content of $\Lambda_x$ is moved in $\Lambda_{x+1}$. Thus the renormalised evolution is a shift $x\rightarrow x+1$ in the $\L'$ lattice.}
\label{fig:CG}
\end{figure}
Notice that, in principle, when we deal with more complex automata, we have no reason to believe that $\Tredmat$ is a unitary matrix. However, condition~\eqref{eq:cg} implies that $\Tredmat=U_\tS e^{i\phi}$ is a unitary matrix for every suitable lattice $\mathcal{L}$. If this is the case eq.~\eqref{eq:tred} guarantees that $\Tred$ is an 
automorphism.

The following Lemma imposes a necessary and sufficient condition on $U_\tT$ for eq.~\eqref{eq:cg} to hold:
\begin{lemma} \label{lem:comm}
Eq.~\eqref{eq:cg} holds with $\tS_\mathcal{L}=(\tT^N)_r$ being an automorphism \textit{iff} 
\begin{equation}\label{eq:commuwrap}
[\Tnmat,\Pi]=0.
\end{equation}
\end{lemma}
\begin{proof}
If \eqref{eq:cg} holds, we can compose it on the right with $\tJ$ and get
\begin{equation}\label{eq:renormp}
\tT^N_w\circ \tP=\tV\circ\tS_w\circ\tJ.
\end{equation}
Now, composing the r.h.s.~of the latter expression with $\tP=\tV\circ\tJ$ on the left gives
$
\tP\circ\tV\circ\tS_w\circ\tJ=\tV\circ\tS_w\circ\tJ=\tT^N_w\circ\tP,
$
while on the other hand, considering l.h.s., we obtain
$
\tP\circ\tT^N_w\circ\tP=\tT^N_w\circ\tP.
$

In terms of the unitary matrix $U_{\tT^N}$ this means:
$
\Pi U_{\tT^N}\Pi=\Pi U_{\tT^N}$,
i.e.  $U_{\tT^N}$ is block-diagonal over $\operatorname{Supp}(\Pi)$, 
and the commutation relation~\eqref{eq:commuwrap} trivially follows.
To prove the converse, 
we now remind that $\Tredmat =\J\Tnmat\J^\dag$. Remembering that $\J\J^\dag=I$ and 
$\J^\dag\J=\Pi_{\mathcal{L}}$, we then have
\begin{align*}
&\Tredmat \Tredmat^\dag= \J \Tnmat \J^\dag\J \Tnmat^\dag \J^\dag=\\
&= \J \Tnmat \Pi \Tnmat^\dag \J^\dag  = \J \Pi \Tnmat\Tnmat^\dag \J^\dag=\\
&=\J\J^\dag\J\J^\dag= I_{\mathcal{L'}},
\end{align*}
where we exploited the commutation of $\Tnmat$
with $\Pi$. The same holds for
$\Tredmat^\dag\Tredmat$ and thus $\Tredmat$ is
unitary over $\Aredloc $.  We can thus make the
identification $U_\tS=\Tredmat$, and obtain the
identity
$
 U_\tS\J=\J\Tnmat,
$
which yields
$
\tV\circ\tS_w=\tT^N_w\circ\tV,
$
with  $\tS_w$ an automorphism of the lattice algebra 
$\mathsf{A}(\mathcal{L'})$.
\end{proof}
Notice that eq.~\eqref{eq:commuwrap} is necessary and sufficient for condition~\eqref{eq:cg} for~\emph{some} automorphism $\tS_\mathcal L$, which by now is not necessarily the wrapped version of a fixed QCA $\tS$ on $\mathbb Z^s$. Thus, we do not have yet a necessary and sufficient condition for renormalisability.

It is easy to see that the commutation
relation~\eqref{eq:commuwrap} can be formulated as
follows
\begin{align}
  \begin{aligned}
    &\Tnmat \Pi
  \Tnmat^\dag=\Pi \iff 
\tT^N_w(\Pi)=\Pi,
  \end{aligned}
\label{eq:conschaingen}
\end{align}
for all the possible lattices.  We now show that all
maps $\tS_\mathcal L$ share the same local rule. This has
two important consequences. The first one is that eqs.~\eqref{eq:inducedN} 
and~\eqref{eq:commuwrap} become necessary and sufficient for renormalisability through 
$(\Pi_0,\Lambda_0)$. The second and most powerful one is that we can check the
condition in eq.~\eqref{eq:conschaingen} only on
the smallest lattice such that the neighborhood of
$\tT^N$ is regular.  We first state the
following two lemmas. The first one relates the
translations in $\mathcal{L}$ with the ones in the
coarse-grained lattice $\mathcal{L}'$:
\begin{lemma}
$(\tT^N)_r$ is invariant under shifts on $\mathcal{L'}$
\end{lemma}
\begin{proof}
A shift $\tau_x^w$ along the vector $x$ in $\mathcal{L'}$ 
corresponds to a shift $\tau_{Nx}^w$
along the vector $Nx$ in $\mathcal{L}$.
Since  
\begin{align}\label{eq:shiftisom}
\begin{aligned}
\tau_{Nx}^w\tV=\tV\tau_{x}^w,\\
\tau_{x}^w\tJ=\tJ\tau_{Nx}^w,\\
\end{aligned}
\end{align}
 and the evolution
$\tT^N$ commutes with $\tau_{Nx}^w$,
considering their composition $(\tT^N)_r$, we have:
\begin{align*}
[(\tT^N)_r,\tau_x]=0 \quad \forall x\in \mathcal{L}.
\end{align*}
\end{proof}
The other Lemma we have to prove is about the form of the unitary matrix $\Tnmat$.
\begin{lemma}
$\Tnmat$ can be written in the following form:
\begin{align*}
&\Tnmat=[(I_{\Lambda_{0}}\otimes V)(U\otimes I_{\Lambda_{\mathcal W}})],\\
&\mathcal W\coloneqq\mathcal{L}\setminus(\Lambda_0\cup\mathcal{N}^N_{\Lambda_0}),
\end{align*}
where $U$ is independent of the wrapping, and $\mathcal{N}^N_{\Lambda_0}$ denotes the neighbourhood of $\Lambda_0$ for $\tT^N$.
\end{lemma}
\begin{proof}
Since $\tT(\mathsf{A}(\Lambda_0))\subseteq\mathsf{A}(\mathcal{N}_{\Lambda_0})$, this means that $\tT^N(\mathsf{A}(\Lambda_0))\subseteq\mathsf{A}(\mathcal{N}^N_{\Lambda_0})$. 
In other words, the algebra $\tT^N(\mathsf{A}(\Lambda_0))$ is an homomorphic 
representation of the full matrix algebra $\mathsf{A}(\Lambda_0)$ in 
$\mathsf{A}(\mathcal N^N_{\Lambda_0})$, which means that there exists a unitary 
$U\in\mathsf{A}(\Lambda_0\cup\mathcal N^N_{\Lambda_0})$ such that, for every 
$X\in\mathsf A(\Lambda_0)$,
\begin{align*}
&\tT^N(X)=U^\dag(X\otimes I_{\mathcal{N}^N_{\Lambda_0}\setminus\Lambda_0})U\in\mathsf A(\Lambda_0\cup\mathcal{N}_{\Lambda_0}).
\end{align*}
Consider now the wrapped evolution $\tT^N_w$. This is implemented over the wrapped lattice as a unitary matrix (see eq.~\eqref{eq:evunit}), i.e.
\begin{align}
\tT^N_w(X)=\Tnmat^\dag(X\otimes I_{\mathcal{L}\setminus\mathcal{N}_X})\Tnmat.
\end{align}
Since $\tT$ and $\tT_w$ share the same local rule, we then have
\begin{align*}
\tT_w^N(X)&=\Tnmat^\dag(X\otimes I_{\mathcal{L}\setminus\Lambda_0})\Tnmat\\
&=
U^\dag(X\otimes I_{\mathcal{N}^N_{\Lambda_0}\setminus\Lambda_0})U\otimes I_{\mathcal{L}\setminus(\Lambda_0\cup\mathcal{N}^N_{\Lambda_0})}.
\end{align*}
One can easily prove that this implies
\begin{equation*}
\Tnmat(U\otimes I_{\mathcal{L}\setminus(\Lambda_0\cup\mathcal{N}^N_{\Lambda_0})})^\dag=(I_{\Lambda_0}\otimes V),
\end{equation*}
for some unitary $V\in\mathsf A(\mathcal L\setminus \Lambda_0)$, and finally
\begin{equation}\label{eq:decompun}
\Tnmat=(I_{\Lambda_0}\otimes V)(U\otimes I_{\mathcal{L}\setminus(\Lambda_0\cup\mathcal{N}^N_{\Lambda_0})}).
\end{equation}
Since the neighbourhood of $\Lambda_0$ after $N$ steps is $\Lambda_{\mathcal{N}}$, this gives the above decomposition.
\end{proof}
We now prove that the commutation condition is independent of the lattice $\mathcal L$.
\begin{lemma} If condition in \eqref{eq:conschaingen} is satisfied for a lattice $\mathcal{L}$ where $\tT^N$ is regular, then it is true for every other such lattice.
\end{lemma}
\begin{proof}
Let us divide the chain $\Pi_\mathcal{L}$ over $\mathcal{L}$ in
\begin{equation}
\begin{split}
&\Pi=\Pi_{0}\otimes\Pi_{{\mathcal{M}}}\otimes \Pi_{\mathcal W},\\
&{\mathcal M}\coloneqq\mathcal{N}\setminus\{0\},\\
&\mathcal W\coloneqq\mathcal{L}\setminus(\mathcal{N}\cup \{0\}),
\end{split}
\end{equation} 
with the convention:
\begin{align*}
\Pi_{{R}}=\left(\bigotimes_{x\in R}\Pi_{{x}}\right).
\end{align*}
The commutation relation~\eqref{eq:commuwrap} can then be expressed as
\begin{equation}
\begin{split}
&\Tnmat(\Pi_{0}\otimes\Pi_{\mathcal M}\otimes \Pi_{\mathcal W})\Tnmat^\dag=\Pi_{0}\otimes\Pi_{\mathcal M}\otimes \Pi_{\mathcal W}.
\end{split}
\label{eq:Tpi}
\end{equation}

We can represent eq.~\eqref{eq:Tpi} diagrammatically, with the decomposition~\eqref{eq:decompun} of $\Tnmat$, as follows
\begin{equation}
\input{TnPi.tikz} = \input{PiLact.tikz}\ .
\end{equation}
Exploiting the decompositon in eq.~\eqref{eq:decompun} and conjugating both members of eq.~\eqref{eq:Tpi} with $(I_{\Lambda_{0}}\otimes V)^\dag$ we get to:
\begin{equation}
\begin{split}
&(U \otimes I_{\Lambda_{\mathcal W}})(\Pi_{0}\otimes\Pi_{\mathcal M}\otimes \Pi_{\mathcal{W}})(U\otimes I_{\Lambda_{\mathcal W}})^\dag=\\
&=(I_{\Lambda_{0}}\otimes V)^\dag(\Pi_{0}\otimes\Pi_{\mathcal M}\otimes \Pi_{\mathcal{W}})(I_{\Lambda_{0}}\otimes V).
\end{split}
\end{equation}
Diagrammatically:
\begin{equation}\label{eq:diagfact}
\input{TnPiLeft.tikz}=\input{TnPiRight.tikz}\ .
\end{equation}
The two members of eq.~\eqref{eq:diagfact} are factorised in two different ways. This implies that both factorisations must hold, and then
the transformation $V$ on the right produces a factorised projection map, i.e.
\begin{equation}
\begin{split}
&(I_0\otimes V)^\dag(\Pi_{0}\otimes\Pi_{\mathcal M}\otimes \Pi_{{\mathcal W}})(I_0\otimes V )=\\
&=\Pi_{0}\otimes Z^\dag\Pi_{{\mathcal{M}}}Z\otimes \Pi_{{\mathcal W}},
\end{split}
\end{equation}
for some unitary $Z\in\mathsf A(\Lambda_\mathcal M)$. In particular, we have the following diagrammatic identity
\begin{equation}
\input{TnPish.tikz}=\input{TnPiZen.tikz}\ .
\end{equation}
The latter is independent of the lattice $\mathcal L$, and proves that the image of 
$\Pi$ under $\Tnmat \rvert_\Pi$ is a projection of the form 
$\tilde\Pi_{{\mathcal L'\setminus\{0\}}}\otimes\Pi_{0}$. Finally, by translation 
invariance we can conclude that for every $\mathcal L$ such that $\tT^N$ is regular over 
$\mathcal L$ one has
\begin{align*}
\Tnmat^\dag \Pi \Tnmat=\Pi,
\end{align*}
namely condition~\eqref{eq:conschaingen} holds.
\end{proof}
This immediately implies the following corollary
\begin{Corollary}
If condition in \eqref{eq:conschaingen} is satisfied for a lattice $\mathcal{L}$ where $\tT^N$ is regular, then
\begin{align}
\begin{split}
&(U\otimes I_\mathcal W)^\dag(\Pi_{0}\otimes Z^\dag\Pi_{{\mathcal{M}}}Z\otimes \Pi_{{\mathcal W}})(U\otimes I_\mathcal W)\\
&=\Pi_{0}\otimes \Pi_{{\mathcal{M}}}\otimes \Pi_{{\mathcal W}}.
\end{split}
\end{align}
\end{Corollary}
The above result allows us to show that condition~\eqref{eq:commuwrap} is necessary and sufficient for renormalisability through $(\Pi_0,\Lambda_0)$.
\begin{lemma}
Let $\tT^N$ be regular on the lattices $\mathcal L'$ and $\tilde{\mathcal L}'$, and let condition~\eqref{eq:commuwrap} be satisfied for both. Then the local rule of $(\tT^N)_r$ over 
$\mathcal{L'}$ and $\mathcal{\tilde{L}'}$ is the same.
\end{lemma} 
\begin{proof}
It is sufficient to notice that a local operator $A_x$ in $\Aplocc$ is obtained as $\Pi A_{\Lambda_x}\Pi$ for some $A_{\Lambda_x}$ on the macro-cell 
$\Lambda_x$. As a consequence, one has
\begin{align*}
&\Tnmat^\dag A_x \Tnmat=\Tnmat^\dag \Pi A_{\Lambda_x}\Pi\Tnmat\\
&=U^\dag(\Pi_xA_{\Lambda_x}\Pi_x\otimes Z^\dag\Pi_{\mathcal N_{\Lambda_x}\setminus\Lambda_x} Z)U \otimes\Pi_{\mathcal L\setminus(\Lambda_x\cup\mathcal N_{\Lambda_x})},
\end{align*}
where $U$ is independent of $\mathcal L'$. Finally, reminding eq.~\eqref{eq:tred} and~\eqref{eq:tredmat}, one has
\begin{align*}
(\tT^N)_r\rvert_{\tilde{\mathcal{L'}}}(A_x)=(\tT^N)_r\rvert_{\mathcal{L}'}(A_x),
\end{align*}
for every $\tilde{\mathcal L'}$, where we indicated with $(\tT^N)_r\rvert_{\mathcal{L}}$ the evolution $(\tT^N)_r$ restricted to the particular lattice $\mathcal{L}$.
\end{proof}
We have then proved the main result of this section
\begin{proposition}
  \label{prp:thecoarsegrainingisfinite}
  Let $\tT$ be a QCAs on $\mathbb{Z}^s$,
$N$ be a natural number
  and let $\Pi_0 \in \mathsf{A}(\Lambda_0)$ be a
  projection on the local algebra of the region
  $\Lambda_0 \coloneqq [0,N-1]^s$.
  The following conditions are equivalent:
  \begin{enumerate}

  \item there exists a QCA $\tS$
    on $\mathbb{Z}^s$ which is a size-$N$ renormalisation of
    $\tT$ 
through $(\Pi_0,\Lambda_0)$.
\item there exist a lattice
  $\mathcal{L} = \mathbb{Z}^s_{Nf}$ such that the
  regular neighborhood condition is satisfied for
  $\tT^N$ and $[\Tnmat,\Pi]=0$, where $\Tnmat$ is
  the unitary matrix defining
  the QCA on the lattice $\mathcal{L}$.

\end{enumerate}
\end{proposition}

This proves that the renormalisability condition of
eq.~\eqref{eq:cg}, reduces to a finite dimensional
equation in terms of the unitary matrix defining
the QCA on a wrapped lattice. In particular, the renormalisability condition can be checked by considering the
smallest wrapped lattice such that the regular neighbourhood condition is satisfied.
\begin{remark}
 The chain of projections $\Pi$ in eq.~\eqref{eq:commuwrap} identifies the degrees of freedom of the observable algebra of the lattice that we are preserving in the coarse-graining---i.e.~its support---while the discarded ones correspond to the kernel of $\Pi$. However, the choice of $\Pi$ does not fix univocally the operators $\J$.Indeed, the operators $\J$ and $\J'$ give rise to the same projection if :
\begin{align*}
\J'=\bigotimes_{x\in\mathcal{L'}} U_x \J \hspace{10pt} U_x \in SU(d).
\end{align*}
To such family of operators $\J$ we can naturally associate a family of $\tJ,\tV$ as in eq.~\eqref{eq:Jiso3}.
Each of those choices give rise to renormalised automata that differ by a local change of basis, namely, considering $\J$ and $\J'$ as above the respective unitaries that implement the induced maps are:
\begin{align*}
&\Tredmat^{\J} = \J  \Tnmat \J^\dag \\
&\Tredmat^{\J'} = 
\left (\bigotimes_{x\in\mathcal{L'}} U_x \right) \Tredmat \left (\bigotimes_{y\in\mathcal{L'}} U_y \right )^\dag.
\end{align*}
In other words, the choice of $\Pi$ fixes the preserved degrees of freedom in the coarse-graining, while the choice of the basis in which we represent the renormalised evolution is still arbitrary.
\end{remark}

\section{Renormalisation of
  one-dimensional quantum circuits}\label{sec:FDQC}
In this section we specify our analysis to QCAs on
a one dimensional lattice that can be realised as
FDQCs. In
particular, we will consider size-$2$ coarse-graining, 
i.e. two cells and two steps are coarse-grained 
into one cell and one step (the
projections $\Pi_{x}$ act on two cells).  Up
to grouping neighbouring cells together, we can
consider the case in which the QCA $\tT$ is
nearest-neighbour. From the index theory of QCAs
~\cite{Index} we have that $\tT$ is realisable as
a quantum circuit if and only if $\ind(\tT) =
1$. Then $\tT$ can be implemented as a Margolous
partitioning scheme (see Proposition
\eqref{prp:indextheorymainprop}) as follows:
\begin{equation}
\tT =\input{MargInd1.tikz}\ ,
\label{eq:marg}
\end{equation}

Let us now consider a size-$2$ renormalisation of
a QCA that can be implemented as in Equation
\eqref{eq:marg}.  Thanks to Proposition
\eqref{prp:thecoarsegrainingisfinite} we can
consider the case in which $\tT$ is defined on a
finite lattice. We start with the following lemma.

\begin{lemma}
\label{lmm:lemmathatintroducesP}
  Let $\tT$ be QCA that admits a realisation as a
  finite depth quantum circuit as in  Equation
\eqref{eq:marg}. Then $\tT$  admits a size-$2$ renormalisation if and only if
\begin{equation}
\resizebox{0.4\hsize}{!}{
\input{SepevConj.tikz}}=\resizebox{0.4\hsize}{!}{\input{RLChain.tikz}}
\label{eq:RenIndex}
\end{equation}
where we define
\begin{align}
  \label{eq:4}
  P_{\Lambda_x} = \input{Piseparated.tikz}\coloneqq\input{Pisep.tikz}.
\end{align}
\end{lemma}
\begin{proof}
We can write
\begin{equation}
\tT^{2}(\bigotimes_{x\in\mathcal{L'}}\Pi_{x})=\resizebox{0.5\hsize}{!}{\input{T2.tikz}}
\end{equation}
Since each step is translationally invariant we can shift the second evolution step in order to have
\begin{equation}
\tT^{2}(\bigotimes_{x\in\mathcal{L'}}\Pi_{x})=\resizebox{0.5\hsize}{!}{\input{T2shift.tikz}}.
\end{equation}
By introducing the projection $P_{\Lambda_x}$, we have
\begin{equation}
\begin{aligned}
\resizebox{0.4\hsize}{!}{\input{Sepev.tikz}}
\end{aligned}
=
\begin{aligned}
\resizebox{0.4\hsize}{!}{\input{Pichain.tikz}
}
\end{aligned}
\label{eq:evmod}
\end{equation}
and by applying $\input{M1.tikz}$ on both sides we obtain the thesis. 
\end{proof}

It is now convenient to define the operator
\begin{equation}
  \label{eq:definitionofG}
G\coloneqq \input{Ggen.tikz}\eqqcolon\input{Gfig.tikz},
\end{equation}
so that we have
\begin{equation}
\begin{aligned}
\resizebox{0.38\hsize}{!}{
\input{SepevConj.tikz}}
\end{aligned}=
\begin{aligned}
\resizebox{0.38\hsize}{!}{\input{Gev.tikz}}
\end{aligned}
.
\label{eq:Gev}
\end{equation}
It is also convenient to consider the
Schmidt decomposition of the projection $P_{\Lambda_x}$:
\begin{equation}
P_{\Lambda_x}= \input{Piseparated.tikz}
=\sum_{\mu_x}\lambda_{\mu_x}\otimes\rho_{\mu_x}.
\label{eq:SchmidtLR}
\end{equation}
In terms of $G$ and of the operators $\{\lambda_{\mu_x}, \rho_{\mu_x}\}$ we have the 
following necessary and sufficient condition for $\tT$ to admit a size-$2$ renormalisation.
\begin{proposition}
  Let us consider the Schmidt decomposition $P = \sum_{\mu}{\lambda}_\mu\otimes{\rho}_\mu$ . Then
  a one dimensional QCA $\tT$ that has a
  Margolus partitioning scheme as in Equation \eqref{eq:marg} admits a size-$2$ renormalisation iff for every ${\lambda}_\mu$ there exists
  $\tilde{\lambda}_\mu$ and for every
  ${\rho}_\mu$ there exists
$\tilde{\rho}_\mu$ such that
$\sum_{\mu}\tilde{\lambda}_\mu\otimes\tilde{\rho}_\mu$ is   a Schmidt Decomposition of $GPG^\dag$ 
and
\begin{align}
 G^\dag (\rho_{\mu}\otimes\lambda_{\nu}) G=\tilde\rho_{\mu}\otimes\tilde\lambda_{\nu},\quad \forall \mu,\nu,
\label{eq:cond}
\end{align}
\end{proposition}
\begin{proof}
Inverting the second layer of $G$'s on both sides of~\eqref{eq:RenIndex} we have the equivalent condition
\begin{equation}
\resizebox{0.38\hsize}{!}{
\input{LeftG.tikz}}=\resizebox{0.38\hsize}{!}{\input{Newchain.tikz}}\ ,
\end{equation}
where
\begin{align}
\input{NewDiff.tikz}\coloneqq G\left(\sum_{\mu}{\lambda}_\mu\otimes{\rho}_\mu\right)G^\dag.
\end{align} 
In formulas,
\begin{align}
  \begin{aligned}
  &\sum_{\boldsymbol\mu}\bigotimes_{x\in\mathcal{L'}} G^\dag(\rho_{\mu_{x-1}}\otimes\lambda_{\mu_x})G= \\
 & =\bigotimes_{y\in\mathcal{L'}}\sum_{\mu_y}G(\lambda_{\mu_y}\otimes\rho_{\mu_y})G^\dag \\
 &=\bigotimes_{y\in\mathcal{L'}}\sum_{\mu_y}\bar{\lambda}_{\mu_y}\otimes\bar{\rho}_{\mu_y},
  \end{aligned}
\label{eq:Gconv}
\end{align}
where the sum over $\boldsymbol\mu$ is shorthand for the sum over all the indices $\mu_x$,
for $x\in\mathcal L'$, and $\sum_\mu\bar{\lambda}_{\mu}\otimes\bar{\rho}_{\mu}$ is 
some Schmidt decomposition of $GPG^\dag$.
Inverting half of the $G$'s we obtain another equivalent condition:
\begin{equation}
\resizebox{0.38\hsize}{!}{
\input{LeftGconj.tikz}}=\resizebox{0.38\hsize}{!}{\input{Newchainconj.tikz}},\ 
\end{equation}
which in formulas reads
\begin{align}
\begin{aligned}
\sum_{\boldsymbol{\mu}} \bigotimes_{x \in \mathcal{L}'} \lambda_{\mu_{2x}} \otimes G^\dagger (\rho_{\mu_{2x}} \otimes \lambda_{\mu_{2x+1}}) G \otimes \rho_{\mu_{2x+1}} =\\
= \sum_{\boldsymbol{\nu}} \bigotimes_{x \in \mathcal{L}'}\bar{\lambda}_{\nu_{2x-1}} \otimes G \left( \bar{\rho}_{\nu_{2x-1}} \otimes \bar{\lambda}_{\nu_{2x}} \right) G^\dagger \otimes \bar{\rho}_{\nu_{2x}}.
\end{aligned}
\end{align}
Comparing the factorisations on the left and right hand side, and considering the cells 
$\Lambda_{x},\Lambda_{x+1}$, we have:
\begin{align}
\begin{aligned}
\sum_{\mu,\nu} \lambda_\mu \otimes G^\dag(\rho_\mu\otimes \lambda_\nu) G \otimes \rho_\nu =\\
= \sum_{\mu,\nu}{\lambda}_\mu\otimes {\tilde{\rho}}_\mu\otimes {\tilde{\lambda}}_\nu\otimes {\rho}_\nu,
\end{aligned}
\end{align}
for suitable $\tilde{\lambda}_\mu, \tilde{\rho}_\mu$, 
or diagrammatically
\begin{align}
\resizebox{0.38\hsize}{!}{
\input{LeftGbis.tikz}}=\resizebox{0.38\hsize}{!}{\input{Newchainbis.tikz}}\ .
\end{align}
Using the  linear independence of all factors in the terms of the Schmidt decomposition $\lambda_\nu, \rho_\mu$ we get to
\begin{equation}
G^\dag (\rho_{\nu}\otimes\lambda_{\mu})G={\tilde{\rho}}_{\nu}\otimes{\tilde{\lambda}}_{\mu},
\label{eq:tilde}
\end{equation}
for all possible combinations of $\mu,\nu$.
Finally, inserting this result in eq.~\eqref{eq:Gconv} we get to:
\begin{align}
\begin{aligned}
\sum_{\boldsymbol{\mu}}\bigotimes_{x\in\mathcal{L'}}{\bar{\lambda}}_\mu\otimes {\bar{\rho}}_\mu =\\
=\sum_{\boldsymbol{\mu}}\bigotimes_{x\in\mathcal{L'}}\tilde{\lambda}_\mu\otimes \tilde{\rho}_\mu.
\end{aligned}
\end{align}
Thus, invoking~\eqref{eq:tilde}, one can prove that 
$\sum_{\mu}{\tilde{\lambda}}_\mu\otimes {\tilde{\rho}}_\mu$ is another Schmidt 
decomposition of $G P G^\dag$. The proof of the converse statement is now straightforward, 
and this concludes the proof.
\end{proof}
We have the following immediate consequences.
 {\begin{Corollary}\label{cor:gsp}
If a one dimensional QCA $\tT$ with a Margolus partitioning scheme as in Equation \eqref{eq:marg} admits a size-$2$ renormalisation, then
 \begin{align}
 (G^\dag S)^2P(SG)^2=P
 \end{align}
\end{Corollary}}
 {\begin{Corollary}\label{cor:factorizedprojector}
  If $P$ is a factorised projector, i.e. $P = P_1 \otimes P_2 $ then there exist two projectors $\tilde{P}_1$ and
  $\tilde{P}_2$ such that
  $G^{\dag} ( P_1 \otimes P_2) G = \tilde{P}_1\otimes \tilde{P}_2$ .
\end{Corollary}}

It is easy to see that the renormalisability conditions~\eqref{eq:cond} are equivalent to
\begin{align}
  G^\dag
  (  {\rho}^{(n)}\otimes  {\lambda}^{(n)})
  G=
\tilde{\rho}^{(n)} \otimes   \tilde{\lambda}^{(n)}
\label{eq:mon}
\end{align}
where ${\rho}^{(n)},  {\lambda}^{(n)},
\tilde{\rho}^{(n)}$ and  $ \tilde{\lambda}^{(n)}$ denote the monomials of order $n$ in
$\rho_\mu, \lambda_\mu , \tilde{\rho}_\mu$ and 
$\tilde{\lambda_\mu} $ respectively, i.e.
\begin{align*}
  &\lambda^{(n)}=\lambda_{\nu_1}...\lambda_{\nu_n},
    \quad
   \rho^{(n)}= \rho_{\mu_1}...\rho_{\mu_n} \\
     &\tilde{\lambda}^{(n)}=\lambda_{\nu_1}...\lambda_{\nu_n},\quad
\tilde{\rho}^{(n)}=\rho_{\mu_1}...\rho_{\mu_n} 
\end{align*}
From this observation we can prove the following result.

\begin{lemma}
\label{lmm:supportalegbrasofP}
  Let us denote with
  $\mathsf{M} = \langle \{ \rho_\mu \} \rangle$,$\mathsf{N}=\langle \{ \lambda_\mu \} \rangle, \tilde{\mathsf{M}}=\langle \{ \tilde{\rho}_\mu \} \rangle$
  and $\tilde{\mathsf{N}} = \langle \{ \tilde{\lambda}_\mu \} \rangle$ the algebras generated
  by $\rho_\mu, \lambda_\mu , \tilde{\rho}_\mu$
  and $\tilde{\lambda_\mu} $. Then
  $\mathsf{M}$, $\mathsf{N}$,
  $\tilde{\mathsf{M}}$, $\tilde{\mathsf{N}}$
  are
$C^*$-algebras and 
  for any
  $ A \in \mathsf{N}$, $B \in \mathsf{M}$ there exist
$\tilde{A} \in \tilde{\mathsf{N}}$ and
$\tilde{B}\in   \tilde{\mathsf{M}}$
such that
\begin{equation}
\begin{split}
G^\dag (A\otimes B)G=\tilde{A}\otimes\tilde{B}
\end{split}
\label{eq:GAlg}
\end{equation}
\end{lemma}
\begin{proof}
  Since $\sum_{\mu} \lambda_\mu \otimes \rho_\mu = P = P^\dag =  \sum_{\mu} \lambda^\dag_\mu \otimes \rho^\dag_\mu$ and since $\Tr[rho^\dag_\mu\rho_\nu] = \delta_{\mu,\nu}$
 we have
  \begin{align}
\label{eq:13}
    \lambda^\dag_\mu=\sum_{\nu}
  c_\mu^{\nu}\lambda_\nu\quad 
  c_\mu^{\nu} := \Tr [\rho_\nu \rho_\mu].  
  \end{align}
  A similar identity holds for $\rho^\dag_{\mu}$,
$\tilde{\lambda}^\dag_{\mu}$ and $\tilde{\rho}^\dag_{\mu}$. This proves that $\mathsf{M}$, $\mathsf{N}$,
  $\tilde{\mathsf{M}}$, $\tilde{\mathsf{N}}$
  are
$C^*$-algebras.
  Since $P^2=P$, a similar computation gives 
  \begin{align}
    \label{eq:6}
    \lambda_\mu=\sum_{j,k}
  c_\mu^{j,k}\lambda_j\lambda_k \quad 
  \rho_\nu=\sum_{j,k}d_\nu^{j,k}\rho_j\rho_k.  
  \end{align}
  for suitable coefficients $c_\mu^{j,k}$ and
  $d_\nu^{j,k}$. By iteratively using Equation~\eqref{eq:6} we can express any homogeneous polynomial of degree $n$
  in the variables $\lambda_\mu$ as a homogeneous polynomial of degree $n+k$ for any $k$,  and the same holds for homogeneous polynomial in the variables $\rho_\mu$.
  Then, any arbitrary $A \in \mathsf{M}$ and $B \in \mathsf{N}$ can be written as homogeneous polynomials of the same degree $n$ for sufficiently large $n$.
  The thesis now follows from eq.~\eqref{eq:mon}).
\end{proof}

 {\begin{Corollary}
\label{cor:unitalalgebra}
If  $\mathsf{M} $ and $\mathsf{N}$
are both unital then also
$\tilde{\mathsf{M} }$ and $\tilde{\mathsf{N}}$ are unital and we have
that
  for any
  $ A \in \mathsf{N}$, $B \in \mathsf{M}$ there exist
$\tilde{A} \in \tilde{\mathsf{N}}$ and
$\tilde{B}\in   \tilde{\mathsf{M}}$
such that
\begin{align}
  \label{eq:14}
  \begin{aligned}
     G^\dag (A\otimes I)G=\tilde{A}\otimes I \\
G^\dag (I\otimes B)G=I\otimes \tilde{B} 
   \end{aligned}
  \end{align}
\end{Corollary}
\begin{proof}
  Let
  $\tilde{I}$  and $\tilde{I}'$ be the operators corresponding to $I \in \mathsf{M} $ and $I\in \mathsf{N}$ respectively  from Lemma~\ref{lmm:supportalegbrasofP}. Then we have
    $G^\dag (I \otimes B) G = \tilde{I} \otimes \tilde{B} $ and
  $G^\dag (A \otimes I )G = \tilde{A} \otimes \tilde{I}' $ for any $A \in \mathsf{M}$ and any $B \in \mathsf{N}$.
Since  $G$ is unitary $G^\dag( I \otimes I )G = I \otimes I = \tilde{I} \otimes \tilde{I}' $ from which we have the thesis.
\end{proof}}

\begin{Corollary}
  \label{cor:Gcommuteswithswap}
  Let $S: \ket{\psi}\ket{\phi} \mapsto \ket{\phi}\ket{\psi}$ be the swap gate.
If $[G,S] = 0$
then  $[G^2,P]=0$ and
\begin{align}
  \label{eq:5}
  \begin{aligned}
&\forall A \otimes B \in \mathsf{K}, \exists
\tilde{A} \otimes \tilde{B} \in \tilde{\mathsf{K}}  \mbox{ such that }\\
  &G^\dag (A\otimes B)G=\tilde{A}\otimes\tilde{B},     
  \end{aligned}
\end{align}
where we defined
\begin{align}
  \label{eq:7}
  \begin{aligned}
    \mathsf{K}:=
    (\mathsf{M}\otimes\mathsf{N} \cup
    \mathsf{N}\otimes\mathsf{M} )''
    \\
    \tilde{\mathsf{K}}:=
        (\tilde{\mathsf{M}}\otimes \tilde{\mathsf{N}} \cup
  \tilde{  \mathsf{N}} \otimes  \tilde{\mathsf{M}} )''
  \end{aligned}
\end{align}
\end{Corollary}
\begin{proof}
  If $G$ commutes with the swap, Equation
  $\eqref{eq:cond})$ implies that
  \begin{align}
    \label{eq:9}
    \lambda_{\mu}\otimes\rho_{\nu}=
    G (
    \tilde{\lambda}_{\mu}\otimes\tilde{\rho}_{\nu}
)
    G^{\dag}   .
  \end{align}
  Then we have:
  $G^2PG^{\dag 2} = G(G P G^{\dag} ) G^{\dag} = G(
  \sum_ \mu  \tilde{\lambda}_{\mu}\otimes\tilde{\rho}_{\mu})  G^{\dag}  =
  \sum_ \mu  \lambda_{\mu}\otimes \rho_{\mu} = P$.
  From Equation \eqref{eq:9} we have that
Equation \eqref{eq:GAlg} 
holds for every operator generated by $\{\lambda_\mu\otimes\rho_\nu\}
\cup \{\rho_\nu\otimes\lambda_\mu\}$.
\end{proof}

 {\begin{Corollary}
\label{cor:Gcommuteswithswapunital case}
If $G$ commutes with the swap and 
$\mathsf{M}$ and $\mathsf{N} $
are unital $C^*$-algebras then 
\begin{align}
  \label{eq:5}
  \begin{aligned}
&\forall A \in \mathsf{Q}, \exists
\tilde{A} \in \tilde{\mathsf{Q}}  \mbox{ such that }\\
  &G^\dag (A\otimes I)G=\tilde{A}\otimes I,     
  \end{aligned}
\end{align}
where $ \mathsf{Q}$ is the algebra generated by
$\mathsf{M} \cup
    \mathsf{N}$, i.e.  $\mathsf{Q}:=
    (\mathsf{M} \cup
    \mathsf{N} )''$.
\end{Corollary}
\begin{proof}
  The thesis easily follows combining Eq. \eqref{eq:14} and $[G,S] = 0$
\end{proof}}

 {\subsection{Renormalisation of a quantum circuit to a QCA with a different index}\label{sec:shift}}
In this section we will consider the case in
which a one dimensional QCA with index equal to
$1$ admits a size-$2$ renormalisation to a QCA
which is given by a left or right shift
followed by local unitaries.

Thus,
we
need to be able to identify a full matrix algebra
of $d\times d$ matrices in the macro-cell
$\Lambda_x$, that after two steps is totally moved
to the left macro-cell. The following proposition
specialises the renormalisability condition to this
case. We start with three preliminary lemmas.

\begin{lemma}\label{lem:projalg}
If a QCA $\tT$ is renormalisable through $(\Pi_0,\Lambda_0)$, then 
$\Pi\tT^2_w(\mathsf A_{\Lambda_x})\Pi$ is a C*-algebra, and the following inequality holds:
\begin{align}
\begin{split}
&\dim[\mathsf S(\Pi\tT^2_w(\mathsf A_{\Lambda_x})\Pi,\mathsf A_{\Lambda_{y}})]\\
&\leq\dim[
\Pi_0\mathsf S(\tT^2_w(\mathsf A_{\Lambda_x}),\mathsf A_{\Lambda_{y}})\Pi_0].
\end{split}
\end{align}
\end{lemma}
\begin{proof}
First of all we prove that, despite $\Pi\bullet\Pi$ does not define a homomorphism of C*-
algebras, $\Pi\tT^2_w(\mathsf A_{\Lambda_x})\Pi$ is a C*-algebra. Indeed, by the renormalisability criterion~\eqref{eq:commuwrap}, we have $\tT^2_w(\Pi)=\Pi$, and thus, for $A_x\in\Lambda_x$,
\begin{align*}
\Pi\tT^2_w(A_x)\Pi=\tT^2_w(\Pi)\tT^2(A_x)\tT^2_w(\Pi)=\tT^2(\Pi A_x\Pi).
\end{align*}
Then, if $T_0,T_1\in\Pi\tT^2_w(\mathsf A_{\Lambda_x})\Pi$, i.e.~$T_i=\Pi \tT^2_w(S_i)\Pi$ with $S_0,S_1\in\mathsf A_{\Lambda_x}$, we have
\begin{align*}
T_0T_1=\Pi \tT^2_w(S_0)\Pi \tT^2_w(S_1)\Pi=\Pi \tT^2_w(S_0\Pi S_1)\Pi,
\end{align*}
which can be written as
\begin{align*}
T_0T_1&=\Pi \left[\tT^2_w(S_0\Pi_0 S_1)\tT^2_w\left(\bigotimes_{z\neq x}\Pi_0\right)\right]\Pi\\
&=\Pi\tT^2_w(S_0\Pi_0 S_1)\Pi,
\end{align*}
i.e. $T_0T_1\in\Pi\tT^2_w(\mathsf A_{\Lambda_x})\Pi$. This implies that, if we define 
$\mathcal N'_x$ the set of all the $y\in\mathcal L'$ such that 
$\Lambda_y\in\mathcal N^2_{\Lambda_x}$, and $T\in\Pi\tT^2_w(\mathsf A_{\Lambda_x})\Pi$ 
picking $y_0\in\mathcal N^2_{\Lambda_x}$ we have
\begin{align*}
&T=\Pi\tT^2_w(S)\Pi,\\
&\tT^2_w(S)=\bigotimes_{z\not\in\mathcal N'_x}I\otimes\left(\sum_i A_i\otimes B_i\right),
\end{align*}
with $A_i\in\mathsf A_{\Lambda_{y_0}}$ and $B_i\in\bigotimes_{y\in\mathcal N'_x\setminus\{y_0\}}\mathsf A_y$ defining a Schmidt decomposition of $\tT^2_w(S)$. The colleciton of such $A_i$'s generates $\mathsf S(\tT^2_w(\mathsf A_{\Lambda_x}),\mathsf A_{\Lambda_{y_0}})$. Then,
\begin{align*}
\Pi\tT^2_w(S)\Pi=\bigotimes_{z\not\in\mathcal N'_x}\Pi_0\otimes\left(\sum_i\Pi_0 A_i\Pi_0\otimes \Pi_0^{\otimes k}B_i\Pi_0^{\otimes k}\right),
\end{align*}
where $k$ is the cardinality of $\mathcal N'_x\setminus\{y_0\}$. The operators
$\Pi_0^{\otimes k}B_i\Pi_0^{\otimes k}$ featuring in the last decomposition are not linearly independent, thus we cannot take $\Pi_0A_i\Pi_0$ as generators of 
$\mathsf S(\tT^2_w(\mathsf A_{\Lambda_x}),\mathsf A_{\Lambda_{y_0}})$. 
However, if we expand $\Pi_0^{\otimes k}B_i\Pi_0^{\otimes k}$ on a basis for 
$\mathsf A_{\mathcal N'_x\setminus\{y_0\}}$, say $\tilde B_j$, as
\begin{align*}
\Pi_0^{\otimes k}B_i\Pi_0^{\otimes k}=\sum_jc_{ij}\tilde B_j,
\end{align*}
we have
\begin{align*}
\Pi\tT^2_w(S)\Pi=\bigotimes_{z\not\in\mathcal N'_x}\Pi_0\otimes\left(\sum_{ij}c_{ij}\Pi_0 A_i\Pi_0\otimes \tilde B_j\right),
\end{align*}
and thus the generators of 
$\mathsf S(\tT^2_w(\mathsf A_{\Lambda_x}),\mathsf A_{\Lambda_{y_0}})$ are
\begin{align*}
\tilde A_j\coloneqq\sum_ic_{ij}\Pi_0A_i\Pi_0.
\end{align*}
This implies that the algebra generated by the operators $\Pi_0A_i\Pi_0$ contains the support algebra of our interest. In formula,
\begin{align*}
\mathsf S(\Pi\tT^2_w(\mathsf A_{\Lambda_x})\Pi,\mathsf A_{\Lambda_{y_0}})\subseteq\Pi_0\mathsf S(\tT^2_w(\mathsf A_{\Lambda_x}),\mathsf A_{\Lambda_{y_0}})\Pi_0,
\end{align*}
whence the inequality in the thesis.
\end{proof}

\begin{lemma}\label{lem:suppshift}
If a unit index QCA $\tT$ can be renormalised to a left (right) shift, then the support 
algebra of $\tT^2(\mathsf A_{\Lambda_x})$ in $\mathsf A_{\Lambda_{x-1}}$ must coincide with $\mathsf L_{2x-2}$ ($\mathsf R_{2x+2}$).
\end{lemma}
\begin{proof}
We provide the proof in the case of a left shift, as the case of a right shift is 
completely analogous. Since $\mathsf S(\tT^2(\mathsf A_{\Lambda_x}),\mathsf A_{\Lambda_{x-1}})\subseteq \mathsf L_{2x-2}$, it is contained in an algebra isomorphic to 
a cell algebra $\mathsf A_0$. Let us consider a wrapping $\tT_w$ of $\tT$ and suppose that 
$\tT$ can be renormalised to $\tau_1^{-1}$ through projection $(\Pi_0,\Lambda)$. Then by 
condition~\eqref{eq:cg}, and considering~\eqref{eq:shiftisom}, we have
\begin{align*}
\tT^2_w(\Pi \mathsf A_{\Lambda_x}\Pi)=\Pi \mathsf A_{\Lambda_{x-1}}\Pi\cong\mathsf A_0,
\end{align*}
that means
\begin{align*}
\tT^2_w(\Pi)\tT^2_w (\mathsf A_{\Lambda_x})\tT^2_w(\Pi)&=\Pi\tT_w^2(\mathsf A_{\Lambda_x})\Pi\\
&\equiv\mathsf S(\Pi\tT_w^2(\mathsf A_{\Lambda_x})\Pi,\mathsf A_{\Lambda_{x-1}}).
\end{align*}
Now, exploiting lemma~\ref{lem:projalg}, we have
\begin{align*}
d^2&=\dim[\mathsf S(\Pi\tT^2(\mathsf A_{\Lambda_x})\Pi,\mathsf A_{\Lambda_{x-1}})]\\
&\leq\dim[\mathsf S(\tT^2(\mathsf A_{\Lambda_x}),\mathsf A_{\Lambda_{x-1}})] \leq d^2.
\end{align*}
Thus, $\mathsf S(\tT^2(\mathsf A_{\Lambda_x}),\mathsf A_{\Lambda_{x-1}})$ must be 
isomorphic to a cell algebra $\mathsf A_0$, and it must then coincide with 
$\mathsf L_{2x-2}$.
\end{proof}
\begin{lemma}\label{lem:projrep}
Let $Q$ be a projection in the tensor product of two full matrix algebras 
$\mathsf A\otimes\mathsf B$, and let 
\begin{align*}
Q(A\otimes I)Q=V(A\otimes\ket\psi\bra\psi)V^\dag,\quad\forall A\in\mathsf A
\end{align*}
for some unitary $V$. Then
\begin{align*}
Q=I\otimes\ket\varphi\bra\varphi.
\end{align*}
\end{lemma}
\begin{proof}
By the properties of Kraus decompositions, the identity in the hypothesis implies that
\begin{align*}
Q(I\otimes \ket{\eta_n})=\varphi_n V(I\otimes\ket\psi),
\end{align*}
where $\{\ket{\eta_n}\}$ is an arbitrary orthonormal basis. Multiplying both sides by 
$I\otimes\bra{\eta_n}$ to the right and summing over $n$, we have then
\begin{align*}
Q=V(I\otimes\ket\psi\bra\varphi),
\end{align*}
where $\ket\varphi=\sum_n\varphi_n^*\ket{\eta_n}$, and finally, as $Q=Q^\dag Q$,
\begin{align*}
Q=I\otimes\ket\varphi\bra\varphi.
\end{align*}
\end{proof}

\begin{proposition}
The condition~\eqref{eq:cg} in the case of renormalisation of a unit index QCA to a left shift reduces to
\begin{align}\label{eq:diffind}
G^\dag=C(J\otimes W)S,\quad P_0=I\otimes\ket\varphi\bra\varphi,
\end{align}
where $J$ and $W$ are suitable unitaries and $C$ is a controlled unitary
\begin{align*}
C=I\otimes\ket\psi\bra\psi+X\otimes(I-\ket\psi\bra\psi), 
\end{align*}
and with $W^\dag\ket\psi=\ket\varphi$. In the case of a right shift, it reduces to
\begin{align}\label{eq:diffind}
G^\dag=C(J\otimes W)S,\quad P_0=\ket\varphi\bra\varphi\otimes I,
\end{align}
where $J$ and $W$ are suitable unitaries and $C$ is a controlled unitary
\begin{align*}
C=\ket\psi\bra\psi\otimes I+(I-\ket\psi\bra\psi)\otimes X, 
\end{align*}
and with $J^\dag\ket\psi=\ket\varphi$.
\end{proposition}
\begin{proof}
The support algebra of $\tT^2(\mathsf A_{\Lambda_x})$ in $\mathsf A_{\Lambda_{x-1}}$ is a 
subalgebra $\mathsf L'$ of $\tT(\mathsf L_{2x-1})$, and precisely 
$\mathsf L'=\mathsf S(\tT(\mathsf L_{2x-1}),\mathsf L_{2x-2})\cong \mathsf S(\mathsf L_{2x-1},\tT^{-1}(\mathsf L_{2x-2}))$. Similarly, the support algebra of 
$\tT^2(\mathsf A_{\Lambda_x})$ in $\mathsf A_{\Lambda_{x+1}}$ is 
$\mathsf R'\cong \mathsf S(\mathsf R_{2x+1},\tT^{-1}(\mathsf R_{2x+2}))$. 
By lemma~\ref{lem:suppshift}, we need to have 
\begin{align*}
&\mathsf L'\cong\tT^{-1}(\mathsf L_{2x-2})\cong\mathsf L_{2x-1}\\
&\ (\mathsf R'\cong\tT^{-1}(\mathsf R_{2x+2})\cong\mathsf R_{2x+1})
\end{align*}
In the following, we will denote by $\tilde{\mathsf L}$ the algebra $\tT^{-1}(\mathsf L)$ and by $\tilde{\mathsf R}$ the algebra $\tT^{-1}(\mathsf R)$. According to the factorisation
\begin{align*}
&\mathsf A_{2x-1}\otimes\mathsf A_{2x}\cong\tilde{\mathsf L}_{2x-1}\otimes\tilde{\mathsf R}_{2x-1},\\
&\ (\mathsf A_{2x+1}\otimes\mathsf A_{2x+2}\cong\tilde{\mathsf L}_{2x+1}\otimes\tilde{\mathsf R}_{2x+1}),
\end{align*}
and by definition of support algebra, we have
\begin{align*}
&\mathsf L_{2x-1}\subseteq \mathsf S(\mathsf L_{2x-1},\tilde{\mathsf L}_{2x-1})\otimes\mathsf S(\mathsf L_{2x-1},\tilde{\mathsf R}_{2x-1}),\\
&\ (\mathsf R_{2x+1}\subseteq \mathsf S(\mathsf R_{2x+1},\tilde{\mathsf L}_{2x+1})\otimes\mathsf S(\mathsf R_{2x+1},\tilde{\mathsf R}_{2x+1})).
\end{align*}
We now focus on the case where the renormalisation is a left shift. The case of a right 
shift is completely analogous. It is most convenient to refer to the algebras 
$\mathsf A_{2x-1}\otimes\mathsf A_{2x}$ after application of the first step $\tT$. Given 
then $A\in\tilde{\mathsf L}_{2x}\subseteq \mathsf A_{2x}\otimes\mathsf A_{2x+1}$, we have
\begin{align}\label{eq:ltildel}
\begin{split}
\tT(A)&=M_2^\dag[I\otimes \tU(A)]M_2\\
&= M_1\left(\sum_i L_i\otimes W_i\right)M_1^\dag.
\end{split}
\end{align}
for suitable $\tU$ homomorphism from $\tilde{\mathsf L}_{2x}$ to $\mathsf A_{2x}$ and 
operators $L_i$ and $W_i$.
Applying $\tT^2$ to $A$ then gives the following element of $\mathsf A_{2x-2}\otimes\mathsf A_{2x-1}\otimes\mathsf A_{2x}\otimes\mathsf A_{2x+1}$
\begin{align}\label{eq:t2}
\begin{split}
\tT^2(A)&=\sum_i M^\dag_2(I\otimes L_i)M_2\otimes M_2^\dag(W_i\otimes I)M_2.
\end{split}
\end{align}
From eq.~\eqref{eq:ltildel} we have
\begin{align}\label{eq:gschanges}
&G^\dag S[\tU(A)\otimes I]SG=\sum_iL_i\otimes W_i,
\end{align}
where the collection of operators $L_i$ for all $A$ must generate a full cell algebra. 
From eq.~\eqref{eq:t2}, considering that 
$M_2^\dag(I\otimes L_i)M_2\in\mathsf L_{2x-2}$ 
while $M_2^\dag(W_i\otimes I)M_2\in\mathsf R_{2x}$, by lemma~\ref{lem:projrep} it must 
be $\Pi_0 =M_2^\dag(\ket\psi\bra\psi\otimes I)M_2$, that can also be written as
\begin{align}\label{eq:keytop}
P_0=G^\dag S(I\otimes\ket\psi\bra\psi)SG. 
\end{align}
This implies that the restriction of $\Pi\tT^2(A)\Pi$ to $\Lambda_{x-1}\cup\Lambda_x$ is
\begin{align*}
\begin{split}
&(\Pi_0\otimes\Pi_0)\tT^2(A)(\Pi_0\otimes\Pi_0)\\
&\ =\sum_iw_i M_2^\dag[\ket\psi\bra\psi\otimes L_i]M_2\otimes\Pi_0,
\end{split}
\end{align*}
with $w_i=\bra\psi W_i\ket\psi$. We can then write
\begin{align*}
(\Pi_0\otimes\Pi_0)\tT^2(A)(\Pi_0\otimes\Pi_0)=\tZ(A)\otimes\Pi_0,
\end{align*}
with
\begin{align*}
\tZ(A)=\sum_iw_i M_2^\dag[\ket\psi\bra\psi\otimes L_i]M_2.
\end{align*}
Now, since $\tZ(A)$ must be of the form
\begin{align*}
\tZ(A)&=Y^\dag(\ket\eta\bra\eta\otimes A)Y,
\end{align*}
for a suitable unitary $Y$ and normalised $\ket\eta$, it must be $\sum_iw_iL_i=ZAZ^\dag$, i.e.
\begin{align*}
ZAZ^\dag=(I\otimes\bra\psi)\{G^\dag S[\tU(A)\otimes I]SG\}(I\otimes\ket\psi).
\end{align*}
Since this has to hold for every $A$, by the same argument as in the proof of 
lemma~\ref{lem:projrep}, we have
\begin{align*}
(I\otimes\ket{\psi}\bra{\psi})G^\dag S =(I\otimes\ket{\psi}\bra{\psi})J\otimes W,
\end{align*}
with $J\coloneqq ZU^\dag$ and $W$ a suitable unitary such that 
$W^\dag\ket\psi=\ket\varphi$. Moreover, we have
\begin{align}\label{eq:pishift}
SG(I\otimes\ket{\psi}\bra{\psi})G^\dag S =I\otimes\ket{\varphi}\bra{\varphi}.
\end{align}
Now, by eq.~\eqref{eq:keytop} and corollary~\ref{cor:gsp}, we have
\begin{align*}
(G^\dag S)^3(I\otimes\ket\psi\bra\psi)(SG)^3=G^\dag S(I\otimes\ket\psi\bra\psi)SG,
\end{align*}
and then
\begin{align*}
G^\dag S(I\otimes\ket\psi\bra\psi)SG=SG(I\otimes\ket\psi\bra\psi)G^\dag S,
\end{align*}
which finally leads to
\begin{align*}
P=I\otimes\ket\varphi\bra\varphi.
\end{align*}
Moreover,
\begin{align*}
SG(I\otimes Q)G^\dag S =I\otimes\tilde Q,
\end{align*}
where $Q\coloneqq I-\ket\psi\bra\psi$ and $\tilde Q\coloneqq I-\ket{\varphi}\bra{\varphi}$,
or, equivalently
\begin{align*}
I\otimes Q=G^\dag S(I\otimes\tilde Q)SG.
\end{align*}
Reminding eq.~\eqref{eq:gschanges}, we then have
\begin{align*}
\sum_iL_i\otimes QW_iQ=G^\dag S[\tU(A)\otimes \tilde Q]SG,
\end{align*}
which in turn implies that
\begin{align*}
\sum_iL_i\otimes QW_iQ= {Z'}A{Z'}^\dag\otimes Q,
\end{align*}
for a suitable unitary ${Z'}$. In other terms,
\begin{align*}
G^\dag S[\tU(A)\otimes \tilde Q]SG=Z'A{Z'}^\dag\otimes Q.
\end{align*}
If $\{\ket{\theta_m}\}_{m=0}^{d-1}$ is an orthonormal basis with 
$\ket{\theta_0}=\ket\psi$, for $m>0$ we then have
\begin{align*}
(I\otimes\bra{\theta_m})G^\dag S[\tU(A)\otimes \tilde Q]SG(I\otimes\ket{\theta_m})={Z'}A{Z'}^\dag.
\end{align*}
We can then follow the same argument as in the proof of lemma~\ref{lem:projrep}, and 
conclude that
\begin{align*}
(I\otimes \ket{\theta_m}\bra{\theta_m})G^\dag S=(I\otimes \ket{\theta_m}\bra{\theta_m})K\otimes W,
\end{align*}
with $K={Z'}U^\dag$, and finally
\begin{align*}
G^\dag S=C(J\otimes W),
\end{align*}
where, upon defining $X\coloneqq KJ^\dag$, $C$ is the controlled unitary
\begin{align*}
C=I\otimes\ket\psi\bra\psi+X\otimes Q.
\end{align*}
%
The composition of $\input{M2.tikz}$ and $\input{M1.tikz}$ then results in a swap up to a local unitary, followed by $C$. Diagrammatically, we have 
\begin{equation}\label{eq:gshift}
\input{Ggen.tikz}=\input{Uswap.tikz}\ .
\end{equation}
Considering eq.~\eqref{eq:pishift}, it must be
\begin{align*}
I\otimes W^\dag\ket\psi\bra\psi W=I\otimes\ket\varphi\bra\varphi,
\end{align*}
i.e.~one can choose an immaterial overall phase so that $W^\dag\ket\psi=\ket\varphi$.
\end{proof}

Considering eq.~\eqref{eq:RenIndex}, one has, up to an irrelevant phase factor, the  following four possibilities (see Appendix~\ref{app:shift} for the proof).
\begin{enumerate}
\item
$\ket\psi=\ket\varphi$, $C=I$, $Z=J^2$.
\item
$\braket{\psi}{\varphi}=0$, ${W^\dag}^2\ket\psi=\ket\psi$, $C=I$, $Z=J^2$.
\item
$\ket\psi=\ket\varphi$, $X\neq I$, $Z=J^2$.
\item
$\braket{\psi}{\varphi}=0$, ${W^\dag}^2\ket\psi=\ket\psi$, $X\neq I$, $Z=JXJ$.
\end{enumerate}

 {\section{Renormalisation of one dimensional Qubit QCA}\label{sec:Qubits}}
In this section we will apply the tools developed
so far to the renormalisation of one dimensional
qubit QCAs with index $1$, i.e. QCA on
$\mathbb{Z}$ such that each cell is a two
dimensional quantum system and such that they can
be realized as a finite depth quantum circuit.

Comparing Equation \eqref{eq:10} with Equation \eqref{eq:marg} we have
\begin{align}
\begin{aligned}
\input{M2.tikz}=\input{M2qbit.tikz}\ ,\ \\
\\
\input{M1.tikz}=\input{def2.tikz}\ .
\end{aligned}
\end{align}
and the operator $G$ of Equation \eqref{eq:definitionofG} becomes
\begin{align}
G=C_{\phi}(U\otimes U)C_\phi
\label{eq:Gqubit}
\end{align}

\begin{lemma}\label{lem:factalg}
  Let $\tT$ be a one dimensional qubit QCA with
  index $1$ which admits a size-$2$
  renormalisation. Let $P$ and $G$  be defined as in Lemma
  \ref{lmm:lemmathatintroducesP}.  Let us denote
  with $\mathcal{M}$ the algebra of $2\times 2$
  complex matrices.  Then one of the following
  three conditions holds:
  \begin{enumerate}
  \item \label{item:fullmatrix}
 For any $A,B \in \mathcal{M} $ there exist
   $\tilde{A},\tilde{B}  \in \mathcal{M}$,
(where $\tilde{A}$ depends only on $A$ and 
 $\tilde{B}$ depends only on $B$)  such that
  \begin{align}
    \label{eq:factorfullmatrix}
  G^\dag (A\otimes B) G=\tilde{A}\otimes\tilde{B}    
\end{align}
 \item \label{item:abelian}
There exist a  two dimensional unital abelian
algebra $\mathcal{A}$ and a two dimensional unital abelian algebra
$\tilde{\mathcal{A}}$ such that 
 for any $A,B \in \mathcal{A} $ there exist
   $\tilde{A},\tilde{B}  \in \tilde{\mathcal{A}}$,
(where $\tilde{A}$ depends only on $A$ and 
 $\tilde{B}$ depends only on $B$)  such that
  \begin{align}
    \label{eq:factorfullmatrix2}
  G^\dag (A\otimes B) G=\tilde{A}\otimes\tilde{B}    
\end{align}
\item \label{item:factorizedprojectorcase} There exist a rank one projector 
$\ketbra{\psi}{\psi}$ and
 a rank one projector 
$\ketbra{\phi}{\phi}$ such that $P = I \otimes \ketbra{\psi}{\psi}$ or
$P =\ketbra{\psi}{\psi}  \otimes I$  and
$ G^\dag (I \otimes \ketbra{\psi}{\psi} ) G= I \otimes \ketbra{\phi}{\phi}$ or 
$ G^\dag (I \otimes \ketbra{\psi}{\psi} ) G=\ketbra{\phi}{\phi} \otimes I$ 
  \end{enumerate}

\end{lemma}
\begin{proof}
First let us notice that Equation \eqref{eq:Gqubit} implies that $[G,S] =0 $
  
 If $\mathsf{X}$ is a $C^*$-subalgebra
  of $\mathcal{M}$, we have four possible cases:
  $ i)$
  $\mathsf{X} = \mathcal{M}$, $ii)$
  $\mathsf{X} = \mathcal{A}$ where $ \mathcal{A}$
  is 2-dimensional unital abelian algebra, $iii)$
  $\mathsf{X} = \mathsf{I} := \{z I, z\in\mathbb
  C\}$, or $iv)$
  $\mathsf{X}=\mathsf{O} :=
  \{z\ket\psi\bra\psi\mid z\in\mathbb C\})$. 

  Let now $\mathsf{M} $, $\mathsf{N}$ be the
  algebras introduced in Lemma
  \ref{lmm:supportalegbrasofP} and let us examine
  all the possible cases.
  
  \emph{ $\mathsf{M}$ or $\mathsf{N}$ is not unital}.
  Let us suppose that  $\mathsf{M}$ is not unital.
  Then we have $\mathsf{M}=\mathsf{O} :=
  \{z\ket\psi\bra\psi\mid z\in\mathbb C\})$ which implies that
  $P = \ketbra{\psi}{\psi} \otimes I$, because $P$ must have rank 2.
  From Corollary \ref{cor:factorizedprojector} we then have that either $G^\dag (\ketbra{\psi}{\psi} \otimes I ) G = \ketbra{\phi}{\phi} \otimes I $ or
  $G^\dag (\ketbra{\psi}{\psi} \otimes I ) G =
  I \otimes \ketbra{\phi}{\phi} $.
  Moreover, since $[G,S] = 0 $,
  $G^\dag (\ketbra{\psi}{\psi} \otimes I ) G = \ketbra{\phi}{\phi} \otimes I $ implies that
  $G^\dag ( I \otimes\ketbra{\psi}{\psi} ) G = I \otimes \ketbra{\phi}{\phi} $ and
  $G^\dag ( \ketbra{\psi}{\psi} \otimes I  ) G = I \otimes \ketbra{\phi}{\phi} $ implies that
  $G^\dag ( I \otimes\ketbra{\psi}{\psi} ) G = \ketbra{\phi}{\phi} \otimes I $. Then condition \ref{item:factorizedprojectorcase} is satisfied.
  The hypothesis the $\mathsf{N}$ is not unital would lead to the same conclusion.

  \emph{ $\mathsf{M}$ or $\mathsf{N}$ are
    trivial}.  Let us suppose that either
  $\mathsf{M} = \mathsf{I}$ or
  $\mathsf{N} = \mathsf{I}$ .  Then either
  $P = I \otimes \ketbra{\psi}{\psi}$ or and
  $P = \ketbra{\psi}{\psi} \otimes I$ and the
  situation is the same as in the previous case,
  i.e. condition
  \ref{item:factorizedprojectorcase} is satisfied.
  
    \emph{ $\mathsf{M} $ and  $\mathsf{N} $ are both unital and $\mathsf{M} = \mathcal{M}$ or  $\mathsf{N} = \mathcal{M}$
    }.
    In this case Corollary \ref{cor:Gcommuteswithswapunital case} applies with $\mathsf{Q} =  \mathcal{M}$
    and condition \ref{item:fullmatrix} is satisfied.

\emph{ $\mathsf{M} $ and  $\mathsf{N} $ are 2-dimensional, unital, abelian algebra and $[\mathsf{M}, \mathsf{N} ] \neq 0$
    }.
    In this case Corollary \ref{cor:Gcommuteswithswapunital case} applies with $\mathsf{Q} =  (\mathcal{M}\cup \mathcal{N})''$. Since
    $\mathsf{M} $ and  $\mathsf{N} $ are non commuting,$ (\mathcal{M}\cup \mathcal{N})'' = \mathcal{M}$
    and condition~\ref{item:abelian} is satisfied.
    
    \emph{$\mathsf{M} $ and  $\mathsf{N} $ are 2-dimensional, unital, abelian algebras and $[\mathsf{M}, \mathsf{N} ] = 0$}
    Since $[\mathsf{M}, \mathsf{N} ] = 0$ we must have
    $\mathsf{M} = \mathsf{N}  = \mathsf{A} $ for a 
2-dimensional, unital, abelian algebra $\mathsf{A} $.
      In this case Corollary \ref{cor:Gcommuteswithswapunital case} applies with $\mathsf{Q} =  (\mathcal{M}\cup \mathcal{N})''= \mathsf{A} $. 
    and condition~\ref{item:abelian}  is satisfied.
\end{proof}

\begin{lemma}\label{lmm:GfactorrGSVVF}
 Let $\tT$ be a one-dimensional qubit QCA with
  index $1$ which admits a size-$2$
  renormalisation.

  \begin{enumerate}
  \item [a)] \label{item:a_lem15} If item \ref{item:fullmatrix} of Lemma \ref{lem:factalg} holds then $G$ is factorised, i.e. $G = W\otimes W$
    for some unitary gate $W$.
  \item [b)] \label{item:b_lem15} If item \ref{item:abelian} of Lemma \ref{lem:factalg} holds then
    there exists an orthonormal basis $\{\ket{\mu_i}\}_{i=0,1}$ such that
    \begin{align}
      \label{eq:17}
      P = \sum_{i,j} c_{i,j} \ketbra{\mu_i}{\mu_i} \otimes \ketbra{\mu_j}{\mu_j}
    \end{align}
    and
 \begin{align}
    \label{eq:11}
 G^\dag = e^{i\gamma} (V\otimes V )F 
 \end{align}
where, in the $\{\ket{\mu_i}\}_{i=0,1}$ basis, $V$ and $ F$
have the matrix form 
  \begin{align}
    \label{eq:15}
    \begin{aligned}
    &V=\begin{pmatrix} a & be^{i\beta}\\ -be^{-i\beta} &  a \end{pmatrix} \\
    &F=\begin{pmatrix}e^{i\theta_{00}}&0&0&0\\0&1&0&0\\
0&0&1&0\\
0&0&0&e^{i\theta_{11}}\end{pmatrix}    
    \end{aligned}
  \end{align}
  where $a,b\in \mathbb{R}$, and $ a^2 + b^2 =1$.
\item[c)] \label{item:b_lem15}  if item \ref{item:factorizedprojectorcase} of lemma \ref{lem:factalg} holds then
  there exist a pair of  orthonormal basis $\{\ket{\mu_i}\}_{i=0,1}$ and $\{\ket{\tilde{\mu}_i}\}_{i=0,1}$ such that  $P = I \otimes \ketbra{\mu_0}{\mu_0}$ or
  $P = \ketbra{\mu_0}{\mu_0}\otimes I$ and 
 \begin{align}
\label{eq:18}
   \begin{aligned}
     G^\dag &(I \otimes \ketbra{\mu_0}{\mu_0} ) G= I
     \otimes \ketbra{\tilde{\mu}_0}{\tilde{\mu}_0} \\ &\implies G^\dag =
e^{i \gamma}     (V\otimes V )F,
   \end{aligned}
   \\
   \label{eq:19}
   \begin{aligned}
 G^\dag &(I \otimes \ketbra{\mu_0}{\mu_0} ) G= \ketbra{\tilde{\mu}_0}{\tilde{\mu}_0} 
     \otimes I\\
     &\implies G^\dag= e^{i \gamma} S(V\otimes V )F
   \end{aligned}
 \end{align}
where $S$ is the swap gate, and, in the  $\{\ket{\mu_i}\}_{i=0,1}$ basis,   $V$ and $F$ have the matrix form as in Equation \eqref{eq:15}.
  \end{enumerate}
\end{lemma}
\begin{proof}
  We will prove the thesis by examining the three
  cases.

  If item \ref{item:fullmatrix}  of Lemma \ref{lem:factalg} holds then
$G$ is a
  unitary operator that maps factorised states into
  factorised states, which implies
  \cite{brylinski2002universal} that either
  $G = W\otimes W'$ or $G = (W\otimes W') S$ where
  $W,W'$ are some unitary operators and $S$ is the
  swap gate.  Since $[G,S] = 0$, we must have
  $G = W\otimes W$ or $G = (W\otimes W) S$.
  If $G = (W\otimes W) S$,
  then $\tilde{A}$ in Equation \eqref{eq:factorfullmatrix} would depend on $B$. Since
  $\tilde{A}$ must only depend on $A$, the case
  $G = (W\otimes W) S$ is excluded.

  Let now consider the case in which item \ref{item:abelian} of Lemma \ref{lem:factalg} holds.
  Let $\{ \ket{\mu_0}, \ket{\mu_1} \}$ an orthonormal pair of vectors such that
$  \mathcal{A} =
\langle
\ketbra{\mu_0}{\mu_0}
\cup
\ketbra{\mu_1}{\mu_1} $,
and let 
$\{ \ket{\tilde{\mu}_0}, \ket{\tilde{\mu}_1} \}$ an orthonormal pair of vectors such that
$  \tilde{\mathcal{A}} =
\langle
\ketbra{\tilde{\mu}_0}{\tilde{\mu}_0}
\cup
\ketbra{\tilde{\mu}_1}{\tilde{\mu}_1} $.
Reminding that
$P = \sum_{\mu} \lambda_\mu \otimes \rho_\mu$
 and $\mathcal{A} = \langle \{\lambda_\mu \} 
\cup \{ \rho_\mu \}\rangle $, we must have
    $  P = \sum_{i,j} c_{i,j} \ketbra{\mu_i}{\mu_i} \otimes \ketbra{\mu_j}{\mu_j}$
Equation \eqref{eq:factorfullmatrix2} implies that
\begin{align}
  \label{eq:16}
  \begin{aligned}
  G^\dag &(\ket{\mu_i}\bra{\mu_i}\otimes\ket{\mu_j}\bra{\mu_j})G=\\
  &=\ket{\tilde\mu_i}\bra{\tilde\mu_i}\otimes\ket{\tilde\mu_j}\bra{\tilde\mu_j}   
  \end{aligned}
\end{align}
for all $i,j \in \{0,1\}$.
Then, from Equation \eqref{eq:16}
we have
\begin{align*}
\begin{aligned}
  &G^\dag=\sum_{i,j}e^{i\theta'(i,j)}
  \ket{\tilde{\mu}_i}\bra{\mu_i}\otimes
  \ket{\tilde{\mu}_j}\bra{\mu_j}
  =(V'\otimes V')F', \\
  &F' := \sum_{i,j}e^{i\theta'(i,j)}
  \ket{{\mu}_i}\bra{\mu_i}\otimes
  \ket{{\mu}_j}\bra{\mu_j}\\
  & V' := \ketbra{\tilde{\mu}_0}{\mu_0} + \ketbra{\tilde{\mu}_1}{\mu_1}.
\end{aligned}
\end{align*}
Since $ [G,S] =0 $
we have $[F,S]=0$, 
and thus $\theta'_{01}=\theta'_{10} = \theta''$.
In the $\{ \ket{\mu_0}, \ket{\mu_1} \}$, $V'$ and $F'$
have the following matrix form
\begin{align}\label{eq:F}
&F'=\begin{pmatrix}e^{i\theta'_{00}}&0&0&0\\
0&e^{i\theta''}&0&0\\
0&0&e^{i\theta''}&0\\
0&0&0&e^{i\theta'_{11}}\end{pmatrix}\\
&V'=\begin{pmatrix}
  a e^{i \alpha'}&b e^{i\beta'}\\
  -b e^{-i\beta'} &a e^{-i \alpha'}
\end{pmatrix}, \; a^2+b^2 = 1.
\end{align}
Upon defining
$\beta = \beta' + \alpha'$,
$\theta_{00} = \theta'_{00} + 2\alpha' - \theta''$,
$\theta_{11} = \theta'_{11} - 2\alpha' - \theta''$, and
$\gamma = \theta'_{01}$ we obtain Equation \eqref{eq:11} and \eqref{eq:15}.

Let us now assume that item
\ref{item:factorizedprojectorcase} of Lemma
\ref{lem:factalg} holds.
Let us define $\ket{\mu _0} := \ket{\psi}$,
$\ket{\tilde{\mu _0}} := \ket{\phi}$ and let
$\ket{\mu _1} $, $\ket{\tilde{\mu}_1} $ be a pair of normalized states  such that $\braket{\mu_1}{\mu_0} = \braket{\tilde{\mu}_0}{\tilde{\mu}_1} = 0$.
Let us suppose that
\begin{align}
  \label{eq:20}
  G^{\dag} (\ketbra{\mu_0}{\mu_0} \otimes I) G =
  \ketbra{\tilde{\mu}_0}{\tilde{\mu}_0} \otimes I.
\end{align}
Then one can easily prove that Equation \eqref{eq:16} holds and therefore $G^\dag = e^{i \gamma } (V \otimes V) F$.
On the other hand if
\begin{align}
  \label{eq:20}
  G^{\dag} (\ketbra{\mu_0}{\mu_0} \otimes I) G =
  I \otimes \ketbra{\tilde{\mu}_0}{\tilde{\mu}_0} 
\end{align}
then it is $SG$ which obeys Equation \eqref{eq:16} and we have $G^\dag = e^{i \gamma }S (V \otimes V) F$.
\end{proof}

\begin{proposition}
\label{prop:coarse-graining-qubitmainresult}
Let $\tT$ be a one dimensional qubit QCA with
index $1$, and let us express
operator $G$ in the basis of the eigenstates of $\sigma_z$  as follows:
  \begin{align}
    \label{eq:8}
    \begin{aligned}
&G = C_\phi (U \otimes U) C_\phi\\
      &  U:= 
      \begin{pmatrix}
        x & y \\
        -\overline{y} & \overline{x}
      \end{pmatrix}
     \\
      &C_\phi = \begin{pmatrix}
        1 & 0& 0& 0 \\
        0 & 1& 0& 0 \\
        0 & 0& 1& 0 \\
        0 & 0& 0& e^{i\phi} 
      \end{pmatrix}
     \end{aligned}
  \end{align}
  Then if $\tT$ admits a size-2 renormalisation iff
  one of
  the following three conditions is satisfied
  \begin{enumerate}
  \item $\phi=0$
  \item  $ x = 0\lor y =0 $
  \item$G = - (Z_\beta\otimes Z_\beta )G'
    (Z_\beta\otimes Z_\beta)^\dag $ where:
    \begin{align} 
      \label{eq:63}
      &\begin{aligned}
        \text{either } 
        G' &= - C_\pi( H \otimes H )C_\pi
        \\
      \text{or }
G' &= C_\pi (H \sigma_x \otimes H \sigma_x) C_\pi
      \end{aligned}\\     
      &    \begin{aligned}
    &H= \frac{1}{\sqrt{2}}
    \begin{pmatrix}
     1 & 1 \\
1 &         -1
\end{pmatrix} \\
    &Z_\beta=
    \begin{pmatrix}
      e^{i \beta} & 0 \\
0 &            e^{-i \beta}
\end{pmatrix}
\end{aligned}
    \end{align}
  \end{enumerate}
\end{proposition}
\begin{proof}
Let us  examine 
the three possible cases given by Lemma \ref{lmm:GfactorrGSVVF}.

Let us suppose that item $a)$ of Lemma \ref{lmm:GfactorrGSVVF} holds. 
Then  $G = W \otimes W$
and the Schmidt rank $r(G)$ of $G$ is $1$.

If either item $b)$ or  item $c)$ with the case of Equation \eqref{eq:18} holds, then
$G^\dag = e^{i \gamma } (V \otimes V) F $.  From
Equation \eqref{eq:15} we have that $r(F) \leq 2$.
Since the
Schmidt rank does not change for a local change of
basis, it must be that also
$r(G) \leq 2$.

Recalling
that $r(O)$
is the rank of the reshuffled matrix
$(\mathcal{R}(O))_{ij,kl}= O_{ik,jl}$ and using the identity 
$G = C_\phi (U \otimes U) C_\phi$,  a straightforward albeit lengthy calculation gives:
\begin{align}
  \label{eq:46}
  \begin{aligned}
  r(G) \leq 2 &\implies \text{det} \big( \mathcal{R} (G) \big)  = 0 \\ & \implies
  |x|^4|y|^4(e^{i\phi} -1)^4   = 0 \\
  & \implies
  \phi=0 \lor x=0 \lor y=0   
\end{aligned} 
\end{align}

Therefore, if either item
$a)$ or item $b)$ or  item $c)$ with the case of Equation \eqref{eq:18} holds, then item $1.$ or item $2.$ of Proposition \ref{prop:coarse-graining-qubitmainresult} must be satisfied.

It remains to analyze the case
in which item
Equation \eqref{eq:19} of Lemma \ref{lmm:GfactorrGSVVF} holds.
In this case we have
$r(SG) \leq 2$
and from a direct computation we have
\begin{align}
  \label{eq:47}
   \begin{aligned}
  r(SG) \leq 2 &\implies \text{det} \big( \mathcal{R} (SG )\big)  = 0 \\ & \implies
 16 |xy|^4 \sin^4\left(  \frac{\phi}{2}\right) = 1 \\
  & \implies
  \phi=\pi \land |x|=|y|= \frac{1}{\sqrt{2}}   
\end{aligned}  
\end{align}
We then have
  \begin{align}
\label{eq:52}
  &  \begin{aligned}
    G &= (Z_{\beta} \otimes Z_{\beta})    G'_\alpha (Z^\dag_{\beta} \otimes Z^\dag_{\beta}), \\
    G'_\alpha &:= C_{\pi} (H_\alpha \otimes H_\alpha)  C_{\pi}
    \end{aligned}
   \\
&    \begin{aligned}
    &H_\alpha= \frac{1}{\sqrt{2}}
    \begin{pmatrix}
      e^{i \alpha} & 1 \\
-1 &            e^{-i \alpha}
\end{pmatrix}
    &Z_\beta=
    \begin{pmatrix}
      e^{i \beta} & 0 \\
0 &            e^{-i \beta}
\end{pmatrix}
    \end{aligned}
  \end{align}

From  \eqref{eq:19} and reminding that $[G^2,P]=[G,S]=0$ we have
\begin{align}
  \label{eq:48}
  &\begin{aligned}
    S G^\dag &(\ketbra{\mu_i}{\mu_i} \otimes \ketbra{\mu_j}{\mu_j})GS =\\
    & = \ketbra{\tilde{\mu_i}}{\tilde{\mu_i}} \otimes \ketbra{\tilde{\mu_j}}{\tilde{\mu_j}}    
  \end{aligned}\\
  \label{eq:49}
  &\begin{aligned}
    S G^\dag &(\ketbra{\tilde{\mu_i}}{\tilde{\mu_i}} \otimes \ketbra{\tilde{\mu_j}}{\tilde{\mu_j}})GS = \\
    &= \ketbra{{\mu_i}}{{\mu_i}} \otimes \ketbra{{\mu_j}}{{\mu_j}}    
  \end{aligned}
\end{align}
Let us denote as $\mathsf{B}$ the algebra generated by $\{  \ketbra{{\mu_i}}{{\mu_i}} ,\ketbra{\tilde{\mu_j}}{\tilde{\mu_j}} \}_{i,j=0,1}$.
From Equation \eqref{eq:48} and \eqref{eq:49} we
have that, for any $b_1 \in \mathsf{B}$ there exist a
$\tilde{b}_1 \in \mathsf{B}$  and
 for any $b_2 \in \mathsf{B}$ there exist a
$\tilde{b}_2 \in \mathsf{B}$ such that
\begin{align}
  \label{eq:50}
   S G^\dag &(b_1\otimes b_2)GS = \tilde{b_1} \otimes \tilde{b_2}  
\end{align}
If the algebra generated by
$\{ \ketbra{{\mu_i}}{{\mu_i}} $ and the one
generated by
$\{ \ketbra{{\tilde{\mu}_i}}{{\tilde{\mu}_i}} $
are different (i.e.
$0< \braket{\mu_0}{\tilde{\mu}} <1$) then
$\mathbf{B} = \mathcal{M}$.  Then either
$SG = W\otimes W$.  However,
\begin{align}
  \label{eq:53}
  \begin{aligned}
  {\Tr}_1[S C_\pi (H_\alpha\otimes H_\alpha) C_\pi] =\\
  i \sin(\alpha)\begin{pmatrix}
    e^{i\alpha} & e^{i\beta} \\
    -e^{-i\beta} & -e^{-i\alpha} 
  \end{pmatrix}  
  \end{aligned}
\end{align}
which is not
proportional to a unitary.
Let us now consider the case in which the algebra
generated by $\{ \ketbra{{\mu_i}}{{\mu_i}} $ and
the one generated by
$\{ \ketbra{{\tilde{\mu}_i}}{{\tilde{\mu}_i}} $
are identical, i.e. either $i)$
$\ket{\mu_i} \propto \ket{\tilde{\mu_i}} $ or
$ii) \ket{\mu_i} \propto \ket{\tilde{\mu_{i+1}}} $
(the sum mod $2$ is understood).
Let us denote with $\lambda_i$ and $\dket{l_i}$
the eigenvalues and the corresponding eigenvectors of $SG$.
We have
\begin{align}
  \label{eq:54}
  \begin{aligned}
  &\lambda_1 = -1, \quad
  \lambda_2 = 1\\
  &\lambda_3 = e^{i \chi}\quad
  \lambda_4 =e^{-i \chi}   
  \end{aligned}
\end{align}
where $ e^{i \chi}$ and $e^{-i \chi}  $
are the eigenvalues of $H_\alpha^2$.
A straightforward calculation gives
\begin{align}
  \label{eq:58}
  \Tr[H_\alpha^2] = \cos(\chi) = \cos(2\alpha) -1
\end{align}
If case $i)$ holds then $SG$ must be diagonal in the
$\{\ket{\mu_i} \ket{\mu_j} \}$ basis.
If $\alpha \neq \pm\frac{\pi}{2}$ we have four distinct eigenvectors. Therefore $\dket{l_i}$ must be a product state for any $i$. But this is not possible, since $\dket{l_1} = \dket{\sigma_y}$.
Then we must have $\alpha = \frac{\pi}{2}$ or $\alpha = - \frac{\pi}{2}$.
Let us now consider the case in which 
$SG$ is antidiagonal in the $\{\ket{\mu_i} \ket{\mu_j} \}$ basis.
Then we must have that $e^{i \chi} = - e^{- i \chi}$
i.e. $\Tr [U^2] = \cos(2\alpha) -1 = 0$, which implies either $\alpha = 0$ or   $\alpha = \pi$.
We have then four possible cases: $\alpha = 0, \pm\frac{\pi}{2} , \pi$.
A straightforward calculation gives:
\begin{align}
  \label{eq:60}
  \begin{aligned}
    H_{0} &= H\sigma_x
& H_{\pi} &= -\sigma_zH \sigma_x\sigma_z\\
  H_{\frac{\pi}{2}} &= +i   Z_{-\frac{\pi}{4}} H Z^\dag_{-\frac{\pi}{4}} &
  H_{\frac{\pi}{2}} &= -i   Z_{\frac{\pi}{4}} H Z^\dag_{\frac{\pi}{4}}  
  \end{aligned}
\end{align}
which is item $3.$ of the present proposition.

Let us now prove the reverse implication, namely
that if one of the three condition of Proposition
\ref{prop:coarse-graining-qubitmainresult} is
satisfied, then the QCA is renormalizable.

If $\phi=0$ then $\tT$ is a QCA made of local
unitaries and it clearly admits a size-$2$
renormalisation (see Fig. \eqref{fig:CG}).

Let us now suppose that either $x =0 $ or $y=0$.
Consider a wrapping of $\tT$ on $\mathbb{Z}_n$
where $n$ is an even number greater than
$8$~\footnote{Actually, it is shown in the
  Appendix that it is sufficient to consider a
  lattice $\mathbb{Z}_n$ with n even and greater
  than $6$.}. One can see that the regular
neighborhood condition is satisfied for $\tT^2$
and that $\tT^2$ is diagonal in the computational
basis, i.e.
$\tT^2 = \sum_{i_1=0}^1\dots \sum_{i_n=0}^1
\lambda (i_1, \dots, i_n) \bigotimes_{x=1}^n
\ketbra{i_x}{i_x} $ for suitable phases
$ \lambda (i_1, \dots, i_n)$. Then
$[U_{\tT^2},\Pi] = 0$ for $\Pi := \otimes_x \Pi_x$
and
$\Pi_x := \ketbra{a}{a}\otimes \ketbra{b}{b} +
\ketbra{a'}{a'}\otimes \ketbra{b'}{b'}$ for any
choices of $a,b,a',b' \in \{0,1\}$ such that
$\Pi_x$ has rank $2$.

Let us now assume that item $3.$ holds.
We can focus on $G'$ since a local change of basis is irrelevant.
We have
\begin{align}
  \label{eq:62}
  \begin{aligned}
    C_\pi H \otimes H  C_\pi  = S C_\pi^{(+)}\\
    C_\pi H \sigma_x \otimes H \sigma_x C_\pi = S
    \sigma_x^{+}\otimes \sigma_x^{+} C_\pi^{(+)}
  \end{aligned}
\end{align}
where $C_\pi^{(+)}$ is the control-Z gate in the
$\ket{+} , \ket{-}$ basis and $\sigma_x^{+}$ is
the pauli $\sigma_x$ in the $\ket{+} , \ket{-}$
basis, i.e. the $\sigma_z$.  It is easy to see that:
if $P = I \otimes \ketbra{\pm}{\pm}$ then
$ C_\pi H \otimes H  C_\pi$ is renormalized to a left shift and 
$ C_\pi H \otimes H  C_\pi$ is renormalized to a left  shift followed by a conjugation with $\sigma_z^{+}$.
If $P =  \ketbra{\pm}{\pm}\otimes I$
$ C_\pi H \otimes H  C_\pi$ is renormalized to a right shift and 
$ C_\pi H \otimes H  C_\pi$ is renormalized to a right  shift followed by a conjugation with $\sigma_z^{+}$.
\end{proof}

We now want to classify all the renormalisable QCAs, along with 
their respective renormalised evolution. The explicit calculations are straightforward 
but tedious and are carried out in Appendix~\eqref{sec:cgd}. There are four main classes 
of such QCA which are described below and are summarized in the following table (the case of QCA which are local rotation ($\phi=0$) are omitted from the table).

\begin{widetext}

\begin{table*}[h!]
\centering

\setlength{\tabcolsep}{4pt}
\renewcommand{\arraystretch}{1.2}
\renewcommand{\tabularxcolumn}[1]{m{#1}}

\newcommand{\diapad}{6pt}
\newcommand{\dia}[1]{\ensuremath{\vcenter{%
  \vskip\diapad
  \hbox{\input{#1.tikz}}%
  \vskip\diapad}}}
\newcommand{\sys}[2]{\ensuremath{\displaystyle\begin{cases}#1\\[2pt]#2\end{cases}}}

\begin{tabularx}{\textwidth}{
|>{\centering\arraybackslash}c
|>{\centering\arraybackslash}X
|>{\centering\arraybackslash}X
|>{\centering\arraybackslash}X
|>{\centering\arraybackslash}X|
}
\hline\hline

\multicolumn{5}{|c|}{\textbf{Renormalised evolutions}}
\\

\hline\hline

$\tT$
& $P=Q_1$
& $P=Q_2$
& $P=I\otimes\ketbra{c}{c}$
& $P=\ketbra{c}{c}\otimes I$
\\

\hline\hline


\dia{MargIndRen}
&
{\begin{tabular}{@{}l@{}}
\sys{\phi'=2\phi}{\theta'=\phi-4\theta}\\[8pt]
\sys{\phi'=2\phi}{\theta'=4\theta-3\phi}
\end{tabular}}
&
\sys{\phi'=-2\phi}{\theta'=\phi}
&
\multicolumn{2}{c|}{\sys{\phi'=0}{\theta'=\pm2(\theta+\delta_{c,0}\phi)}}
\\

\hline

\dia{MargIndRen2}
&
\sys{\phi'=2\phi}{\theta'=-\phi}
&
\sys{\phi'=-2\phi}{\theta'=\phi}
&
\multicolumn{2}{c|}{\sys{\phi'=0}{\theta'=\mp(2c-1)\phi}}
\\

\hline

$\tT$
&$P=I\otimes\ketbra{+}{+}$&$P=I\otimes\ketbra{-}{-}$
& $P=\ketbra{+}{+}\otimes I$ & $P=\ketbra{-}{-}\otimes I$\\
\hline\hline
\dia{MargIndRen3}
&
$\tau_{-}$
&
$\operatorname{Ad}_{\sigma^+_z}\circ\tau_{-}$
&
$\tau_{+}$
&
$\operatorname{Ad}_{\sigma^+_z}\circ\tau_+$
\\

\hline\hline

\end{tabularx}

\caption{Renormalisation flow of the non-trivial size-2 renormalisable qubit QCAs with $\phi\neq0$. Each row is read as $(\mathcal T,P)\mapsto\mathcal S$: the first column shows the microscopic evolution $\mathcal T$, while the remaining columns give the QCA $\mathcal S$ obtained after one size-2 renormalisation step with the indicated coarse-graining projection $P$. Here $C_\phi=\operatorname{diag}(1,1,1,e^{i\phi})$, $R_z(\theta)=e^{i\theta\sigma_z}$, $\vec n\cdot\vec\sigma=a\sigma_x+b\sigma_y$ with $a^2+b^2=1$, and $H$ is the Hadamard gate. In the first two rows, $\phi'$ and $\theta'$ are respectively the controlled-phase and local $z$-rotation parameters of $\mathcal T'$. The projectors are
$Q_1:=|00\rangle\langle00|+|11\rangle\langle11|$ and
$Q_2:=|01\rangle\langle01|+|10\rangle\langle10|$, with $c\in{0,1}$ and $\delta_{c,0}$ the Kronecker delta. The two outputs in the $R_z$ row for $P=Q_1$ correspond, in order, to the choice of $J=\ket{00}\bra{0}+\ket{11}\bra{1}$  and $J=\ket{00}\bra{1}+\ket{11}\bra{0}$;  In the last row, $|\pm\rangle=(|0\rangle\pm|1\rangle)/\sqrt2$ are the eigenstates of $\sigma_x$ with eigenvalues $\pm1$; $\tau_-$ and $\tau_+$ denote the left and right full shifts, respectively; and $\operatorname{Ad}_{\sigma_z^+}(A)=\sigma_z^+A\sigma_z^+$, where $\sigma_z^+=|+\rangle\langle+|-|-\rangle\langle-|$ is Pauli $z$ in the ${|+\rangle,|-\rangle}$ basis.}

\label{tab:res}
\end{table*}
\end{widetext}

\begin{enumerate}
\item \textbf{QCA that equal to the identity after two steps}: This class of QCA corresponds to $\tT^2=\mathcal{I}$ and is renormalised in a local unitary. Since two steps of the evolution correspond to the identity, any choice of the projection $P$ constitutes a legitimate choice for the coarse-graining.
\item \textbf{QCA that apply a local unitary transformation}: This class of QCA corresponds to an evolution made of local unitaries. Imposing the commutation of two steps with the projection, it is easy to see that the only suitable projections are the abelian one over a factorised basis of $U^2$ (see Figure \eqref{fig:CG}).
\item \textbf{QCA that after two steps are diagonal on a factorised basis}: This class of QCA can be renormalised through a projection diagonal in the factorised computational basis. There are two main subclasses:
\begin{enumerate}
\item QCA which are diagonal in the computational basis\\
\item QCA which are anti-diagonal in the computational basis
\end{enumerate}
\item \textbf{QCA made of two layers of controlled $\pi$-phase gates and a layer of Hadamard gates (or Hadamard gates preceded by a $\sigma_z$) } : This class of QCA are renormalized to a shift. Therefore
 these QCA exhibit a change of index after renormalisation, showing that the topology is not robust under the renormalisation scheme.
\end{enumerate}
The first two cases essentially boil down to the example in Fig.\eqref{fig:CG}. Table~\eqref{tab:res} summarises the results for the cases where $\phi\neq 0$. 
We can observe that all the size-2 renormalised QCAs do exhibit information spreading, since the image of their evolution is supported on a single cell. In particular, either they are factorised, or their Margolus partitioning scheme involves commuting layers, which implies that all the odd layers can be grouped into one, and the same for the even layers,
thus collapsing in a transformation with a neighbourhood contained in that of a single step, independently of how many steps are performed (see Fig.\eqref{fig:collapse}). It is interesting to notice that, despite the change of index, also the fourth class of renormalisable QCA falls under this category: despite being shifted, the image of the renormalised evolution is still supported on a single cell.

By direct inspection of the solutions in Table~\eqref{tab:res}, one can straightforwardly 
see that the only fixed point with $\phi\neq 0$ is given by case $3.(a)$. This is given by 
the evolution:
\begin{align}
\begin{cases}
\phi=\frac{2n\pi}{3},\\
\theta=\frac{2n\pi}{3},
\end{cases}
\end{align}
which is mapped into itself by means of the projection $P=\ketbra{0}{0}\otimes\ketbra{1}{1}+\ketbra{1}{1}\otimes\ketbra{0}{0}.$
This case corresponds to the local evolution given by:
\begin{align}
\begin{aligned}
    &\tT_0({A}_0)=X^{\dag}(I\otimes e^{-i\frac{2n\pi}{3}\sigma_z}A_0 e^{i\frac{2n\pi}{3}\sigma_z}\otimes I)X,\\
    &X=(C_{\frac{2n\pi}{3}}\otimes I_3)(I_1\otimes C_{\frac{2n\pi}{3}})\,.
\end{aligned}
\end{align}

\begin{figure}
\centering
\begin{subfigure}{0.9\textwidth}
\begin{tikzpicture}
\node (A) at (-1,0) {\resizebox{0.20\textwidth}{!}{\input{T2noproj.tikz}}};
\node (B) at ($(A)+(4,0)$) {\resizebox{0.20\textwidth}{!}{\input{MargInd1sqr.tikz}}};
\draw[->] (A)--(B);
\end{tikzpicture}
\caption{}
\label{subf:1}
\end{subfigure}
\centering
\begin{subfigure}{0.9\textwidth}
\begin{tikzpicture}
\node (A) at (-1,0) {\resizebox{0.20\textwidth}{!}{\input{T2noproj.tikz}}};
\node (B) at ($(A)+(4,0)$) {\resizebox{0.20\textwidth}{!}{\input{Leftsh.tikz}}};
\draw[->] (A)--(B);
\end{tikzpicture}
\caption{}
\label{subf:2}
\end{subfigure}
\caption{The two examples of no-spreading of information. Either two steps of the evolutions that can be renormalised collapse in a single one: no information can spread beyond the first step (fig. \ref{subf:1}) or the renormalised evolution amounts to a shift followed at most by a local unitary (fig.\ref{subf:2}) .}
\label{fig:collapse}
\end{figure}

\section{Conclusion}
In conclusion, we analysed the problem of renormalising cellular automata on $\mathbb Z^s$ 
via a Kadanoff-like coarse-graining procedure. An analysis similar in its conceptual structure is at the core of the recently proposed \emph{operator algebraic} approach to Wilson's renormalisation group~\cite{PhysRevLett.127.230601,Osborne:2023aa,Luijk:2024aa,Stottmeister:2023aa}.
We provide necessary and sufficient conditions for renormalisability that are amenable to further analysis. We then specialised to the case where $s=1$ and the cells are qubits, with a coarse-graining where two cells and two time steps are mapped into a single one. We completely solved this case. Our results show that for our special case $s=1$ and qubit cells, most QCAs are not renormalisable, with the exception of QCAs where information does not spread,
either because the dynamics is local or because it is transported
rigidly through a full shift. We conjecture that this very restrictive result is due to the extremely limited choice of degrees of freedom that one can ``erase'' in the coarse-graining process in this simple case. In order to prove that more interesting cases exist, the present analysis should be developed either for coarse-grainings where more than two cells and more than two steps are renormalised into a single one, or move beyond the qubit-cell case. One can easily prove that the second choice includes the first one as a special case. Further developments can then be achieved provided a generalisation of the classification of Margolus partitioning schemes. 

Another interesting development would be to carry a similar analysis in the case of 
\emph{Fermionic} cellular automata. Also this case requires a generalisation of the classification of Margolus partitioning schemes, which is then a required intermediate step.

Finally, our necessary and sufficient condition might be exploited to study the 
renormalisation problem for QCAs on higher-dimensional graphs, such as $\mathbb Z^2$ or 
$\mathbb Z^3$. We expect that in these cases more interesting situations can arise, with 
fixed points that might represent phase transitions for the statistical mechanics of QCAs,
opening interesting questions about the universality classes of QCAs and their relation to 
other models, both in statistical mechanics and in quantum field theory. It is worth noting that our necessary and sufficient condition for a QCA to admit a coarse-graining coincides with the one identified in \cite{ArrighiQW} for a Quantum Walk to be Lorentz covariant. Additionally, we have focused on the problem of exact coarse-graining. This approach can serve as a foundation for developing more flexible notions of coarse-graining, where some of the current framework's constraints are relaxed. For instance, in \cite{Rotundo}, an approximate coarse-graining method is proposed for quantum walks, in which a given evolution is replaced by the unitary dynamics that best approximates it—after tracing out part of the system’s degrees of freedom.

\acknowledgments
PP acknowledges financial support from European Union---Next Generation EU through the MUR project Progetti di Ricerca d'Interesse Nazionale (PRIN) QCAPP No. 2022LCEA9Y.
AB acknowledges financial support from European Union---Next Generation EU through the MUR project Progetti di Ricerca d'Interesse Nazionale (PRIN) DISTRUCT No. P2022T2JZ9.
LT acknowledges financial support from European Union---Next Generation EU through the National Research Centre for HPC, Big Data
and Quantum Computing, PNRR MUR Project CN0000013-ICSC.

\bibliographystyle{quantum}
\bibliography{bibliography}

\appendix

\section{The quasi-local algebra}\label{sec:qloc}
Here we define rigorously the quasi-local algebra. 
In the following $\Omega$ and $\Omega_i$ denote finite subsets of $\mathbb Z^s$.
Consider now the  algebras of operators acting on a finite set of sites $\Omega\subset\mathbb{Z}^s$
\begin{equation}
\mathsf{A}(\Omega)=\bigotimes_{x\in\Omega}\mathsf{A}_{x},
\end{equation}
where $\mathsf{A}_x$ is the C*-algebra of operators acting on cell $x$.
For $\Omega_1\subset\Omega_2$ we can consider $\mathsf{A}(\Omega_1)$ as a subalgebra of $\mathsf{A}(\Omega_2)$ by tensoring with the identity over the region $\Omega_2\setminus\Omega_1$ the operators in $\mathsf{A}(\Omega_1)$, that is:
\begin{equation}\label{eq:embed}
O\otimes I_{\Omega_2\setminus\Omega_1}\in\mathsf{A}(\Omega_2)\hspace{5pt}\forall O\in\mathsf{A}(\Omega_1).
\end{equation}
In this way the product $O_1O_2$ is well defined over $\Omega_1\cup \Omega_2$ for $O_1\in\mathsf{A}(\Omega_1)$ and $O_2\in\mathsf{A}(\Omega_2)$. We can now define the equivalence relation $\thicksim$ by saying that $O_1\thicksim O_2$ with $O_i\in\mathsf{A}(\Omega_i)$ if
\begin{equation}
O_1\otimes I_{{(\Omega_1\cup\Omega_2)\setminus\Omega_1}}=O_2\otimes I_{{(\Omega_1\cup\Omega_2)\setminus\Omega_2}} .
\end{equation}
In other words two operators over $\Omega_1$ and $\Omega_2$ are equivalent under $\thicksim$ if they are are the same operator when seen as elements of $\Omega_1\cup\Omega_2$. The local algebra is then the algebra of the equivalence classes under $\thicksim$ of all operators algebras over finite set of sites, or
\begin{definition}[Local Algebra]\label{def:loc}
The Local Algebra $\mathsf{A}^{\mathrm{loc}}$ is defined as
\begin{equation}
\mathsf{A}^{\mathrm{loc}}=\bigcup_{\Omega}\mathsf{A}(\Omega)\hspace{3pt} \big/ \thicksim
\end{equation}
\end{definition}
Moreover, since $\mathsf{A}(\Omega)$ for every finite set $\Omega\subseteq \mathbb Z^s$ is 
a normed space with the uniform operator norm $\|\cdot\|_\infty$, and 
$\|O\otimes I\|_\infty=\|O\|_\infty$, the norm of $O\in\mathsf A_{\Omega_1}$ when embedded 
in $\mathsf{A}(\Omega)$ as in eq.~\eqref{eq:embed} is the same as that of $O$ seen as an 
element of $\mathsf A_{\Omega_1}$. We can then make also $\mathsf{A}^\mathrm{loc}$ into a 
normed space by defining the norm of an equivalence class as the norm of any of its 
representatives which is indeed independent of the chosen representative.
 
We can then define an equivalence relation for Cauchy sequences, setting $\{\mathsf A_{n}\}_{n=0}^\infty \simeq \{\mathsf B_{n}\}_{n=0}^\infty$ if 
\begin{equation}
\forall \epsilon>0 \hspace{5pt} \exists n_0 \hspace{3pt} \mbox{s.t.}\hspace{3pt} \lVert \mathsf A_{k}-\mathsf B_{k}\rVert_{\infty} < \epsilon \hspace{10pt} \forall k \geq n_0.
\end{equation}
We can then complete the local algebra by the standard procedure of defining the algebra of equivalence classes of Cauchy sequences. As a result of this procedure, we obtain
the quasi-local algebra $\mathsf A(\mathbb Z^s)$. The detailed definition follows.
\begin{definition}[quasi-local algebra]
Let $\mathcal{C}(\mathsf{A}^\mathrm{loc})$ denote the space of Cauchy sequences over 
$\mathsf{A}^\mathrm{loc}$.
The quasi-local algebra $\mathsf{A}(\mathbb{Z}^s)$ is defined as
\begin{equation}
\mathsf{A}(\mathbb{Z}^s)=\mathcal{C}(\mathsf{A}^\mathrm{loc}) \big / \simeq
\end{equation}

\end{definition}
Intuitively speaking, this algebra contains all the operators that can be arbitrarily well approximated by local operators in $\mathsf{A}^\mathrm{loc}$. We remark that $\mathsf{A}^\mathrm{loc}\subset\mathsf A(\mathbb Z^s)$.

\section{Uniqueness of the Margolus partitioning scheme}

Let $M'_1$ and $M'_2$ be a pair of unitary operators such that 
eq.~\eqref{eq:marg} is satisfied. Then we must have
\begin{align}
\resizebox{0.8\hsize}{!}{
\input{MargInd1.tikz}=\input{MargInd12.tikz}},
\end{align}
which implies that
\begin{align}
\resizebox{0.8\hsize}{!}{
\input{M1id.tikz}=\input{M2id.tikz}},
\end{align}
The two different factorisations imply that both diagrams must be equal to a factorised unitary of the form
\begin{align*}
\resizebox{0.4\hsize}{!}{
\input{M1id.tikz}}=\resizebox{0.4\hsize}{!}{\input{M1unit.tikz}},
\end{align*}
and this shows that
\begin{align}
M'_1=(U\otimes V)M_1, \quad M'_2=M_2(V^\dag\otimes U^\dag).
\label{eq:unimarg}
\end{align}

\section{Renormalisation of a circuit to a shift}\label{app:shift}
 {
Considering eq.~\eqref{eq:RenIndex}, with $G^\dag$ of the form dictated by 
eq.~\eqref{eq:gshift} and $P=I\otimes\ket\varphi\bra\varphi$, and reminding that 
$W\ket\varphi=\ket\psi$, we know that it must be
\begin{align*}
&(|\bra\psi W\ket\psi|^2)I\otimes\ket\psi\bra\psi+I\otimes QW\ket\psi\bra\psi W^\dag Q\\
&+\bra\psi W^\dag\ket\psi X\otimes QW\ket\psi\bra\psi\\
&+\bra\psi W\ket\psi X^\dag\otimes \ket\psi\bra\psi W^\dag Q\\
&=I\otimes\ket\varphi\bra\varphi,
\end{align*}
where we remind that $Q=I-\ket\varphi\bra\varphi$. If we multiply by $\bra{x_i}\otimes I$ 
to the left and by $\ket{x_i}\otimes I$ to the right, with $X\ket{x_i}=\omega_i\ket{x_i}$, 
we obtain
\begin{align*}
&(|\bra\psi W\ket\psi|^2)\ket\psi\bra\psi+QW\ket\psi\bra\psi W^\dag Q\\
&+\omega_i\bra\psi W^\dag\ket\psi QW\ket\psi\bra\psi\\
&+\omega_i^*\bra\psi W\ket\psi \otimes \ket\psi\bra\psi W^\dag Q=\ket\varphi\bra\varphi.
\end{align*}
Taking the difference between the identities with indices $i$ and $j$, we obtain
\begin{align*}
&(\omega_i-\omega_j)\bra\psi W^\dag\ket\psi QW\ket\psi\bra\psi\\
&+(\omega_i^*-\omega_j^*)\bra\psi W\ket\psi \otimes \ket\psi\bra\psi W^\dag Q=0,
\end{align*}
which has three solutions:
\begin{enumerate}
\item
$\omega_i=\omega_j$ for all $i,j$, i.e.~$X=e^{i\theta}I$. In this case, either $W\ket\psi=\ket\psi$, i.e.~$\ket\varphi=\ket\psi$, or $\braket{\varphi}{\psi}=0$.
\item
$\bra\psi W^\dag\ket\psi=0$. In this case $X$ is free, and $\braket\psi\varphi=0$.
\item
$QW\ket\psi=0$. Also in this case $X$ is free, while $\ket\psi=\ket\varphi$.
\end{enumerate}
}






\section{Renormalised evolution}\label{sec:cgd}
Considering eq.~\eqref{eq:inducedN}) it is easy to see that the renormalised evolution of the single cell algebra $\mathsf{B}_0$ over the periodic lattice $\mathcal{L'}$ is:
\begin{equation}
\mathcal{S}(B)\coloneqq (\tT^2)_{r}(B)=\tJ\circ \tT^{2}_{w} \circ \tV (B),\quad\forall B\in\mathsf B_0.
\label{eq:Sev}
\end{equation}
Since the evolution $\mathcal{S}$ consists in a single step of a nearest neighbours automaton, the minimum regular lattice $\mathcal{L'}$ is made of four cells, i.e. $\mathcal{L'}=\{-1,0,1,2\}$, while for the evolution $\tT^2_{\mathcal{L}}$ we need $\lvert\mathcal{L}\rvert =8$. However, the following Lemma tells us that the renormalised local rule can be computed over a lattice with $\lvert\mathcal{L}\rvert =6.$
\begin{lemma}
The matrix implementing the renormalised local rule is 
\begin{equation}
V_{\mathcal{S}}=\left(\bigotimes_{x\in\mathcal{L'}}\tilde{J}^\dag_{\Lambda_x}\right)\tilde{U}^2_{\tT}\left(\bigotimes_{x\in\mathcal{L'}}\tilde{J}_{\Lambda_x}\right), \label{eq:Vtilderen}
\end{equation}
with:
\begin{equation*}
\tilde{J}_{\Lambda_x}\coloneqq M_1J_{\Lambda_x}.
\end{equation*}
To compute this rule, we can use a lattice $\mathcal{L}$ with $\lvert\mathcal{L}\rvert =6.$
\end{lemma}
\begin{proof}
 If the evolution $\tT$ is implemented by the unitary $U_\tT$, eq.~\eqref{eq:Sev} reads:
\begin{align}
&\mathcal{S}(B)=V_{\mathcal{S}}B_0 V_{\mathcal{S}}^\dag \\
&V_\mathcal{S}=\left(\bigotimes_{x\in\mathcal{L'}}J^\dag_{\Lambda_x}\right)U^2_{\tT}\left(\bigotimes_{x\in\mathcal{L'}}J_{\Lambda_x}\right)\label{eq:Vren}
\end{align}
Then repeating the same manipulations on the evolution as in eq.~\eqref{eq:evmod}, we can define the operator $\tilde{J}_{\Lambda_x}$ as:
\begin{equation}
\begin{split}
&\tilde{J}^\dag_{\Lambda_x}\coloneqq J^\dag_{\Lambda_x}M_1^\dag\eqqcolon\input{Rdnew.tikz} \\
&\tilde{J}_{\Lambda_x}\coloneqq M_1J_{\Lambda_x}\eqqcolon\input{Rnew.tikz} \\
&\tilde{J}_{\Lambda_x}\tilde{J}_{\Lambda_x}^\dag=\input{RRdagnew.tikz}_{\Lambda_x}=\input{Piseparated.tikz}_{\Lambda_x}=P_{\Lambda_x} 	\\
&\tilde{J}_{\Lambda_x}^\dag\tilde{J}_{\Lambda_x}=\input{RdagR.tikz}=I_x\in\mathsf{A}_r(\mathcal{L'})
\label{eq:Jtildeop}
\end{split}
\end{equation}
With those definitions we get to the diagrammatic equation:
\begin{equation}
\mathcal{S}(B)=\input{Renev.tikz}
\label{eq:sbzero}
\end{equation}
where the numbers label the cells in $\mathcal{L}$, the letters the cells of the 
coarse-grained lattice $\mathcal{L'}$, and the half $G$-transformations are meant to act 
on the $(-1,0)$ and $(5,6)$ cell, where possibly $-1\equiv7$. This allows us to express eq.~\eqref{eq:Vren} using the evolution with $G$ operators as in eq.~\eqref{eq:Gev}. Since our constraints are expressed in terms of G, this is amenable for a direct application of our results. Calling $\tilde{U}^2_\tT$ the operator implementing the evolution ~\eqref{eq:RenIndex},i.e.
\begin{equation} 
\tilde{U}^{\dag 2}_\tT\left(\bigotimes_{x\in\mathcal{L'}}P_x\right)(\tilde{U}^2_\tT)=\resizebox{0.35\hsize}{!}{\input{Gev.tikz}},
\end{equation}
eq.~\eqref{eq:Vren} reads
\begin{equation}
V_{\mathcal{S}}=\left(\bigotimes_{x\in\mathcal{L'}}\tilde{J}^\dag_{\Lambda_x}\right)\tilde{U}^2_{\tT}\left(\bigotimes_{x\in\mathcal{L'}}\tilde{J}_{\Lambda_x}\right). 
\end{equation}
Exploiting eq.~\eqref{eq:tilde} we have
\begin{equation}
\mathcal{S}(B)=\input{Renevconj.tikz}
\end{equation}
Since there are no interaction between the cells $5-6$ nor $-1-0$, this equation reduces to a local expression for the evolution $\mathcal{S}({B})$ on a lattice $\mathcal{L}=\{0,1,2,3,4,5\}$.
\end{proof} 
Thus , in the following, we will compute the evolution using eq.~\eqref{eq:Vtilderen} over only three projections. 
\subsection{Renormalised Qubit automata}
We now explicitly calculate the renormalised evolution for the above case of study of qubit cellular automata. 
Let us start from the case where we can choose $P$ to have factorised eigenvectors, i.e.
\begin{equation}
\begin{split}
&P_{\Lambda_x}=\ket{\psi_0}\bra{\psi_0}+\ket{\psi_1}\bra{\psi_1}\\
&\ket{\psi_k}=\ket{l(k)}\ket{j(k)}\\
&P_{\Lambda_x}: \mathsf{A}(\Lambda_x)\to\mathsf{A}(\Lambda_x)
\end{split}
\label{eq:cgproj}
\end{equation}
where $l,j$ are functions $\{0,1\}\to\{0,1\}$, and $\{\ket{0},\ket{1}\}$ are a basis of the single cell algebra isomorphic to $\mathbb{C}^2$. In the same way as in eq.~\eqref{eq:Jop} the choice of a basis for the space $(\mathsf{A}_r)_x\simeq \mathbb{C}^2$ allows us to define the operator $J_{\Lambda_x}$ corresponding to $\Pi_{x}$, we can define the operator $\tilde{J}_{\Lambda_x}$ in eq.~\eqref{eq:Jtildeop} corresponding to $P_{\Lambda_x}$ as well, i.e.
\begin{equation}
\tilde{J}_{\Lambda_x}=\ket{0}_r\bra{\psi_0}+\ket{1}_r\bra{\psi_1},
\label{eq:tildeJ}
\end{equation}
that for the choice of the projection in eq.~\eqref{eq:cgproj} becomes
\begin{equation}
\tilde{J}_{\Lambda_x}=\ket{0}_r\bra{l(0)}\bra{j(0)}+\ket{1}_r\bra{l(1)}\bra{j(1)}.
\end{equation}
Here the subscript $r$ simply indicates that
\begin{equation}
\tilde{J}_{\Lambda_x}: \mathsf{A}(\Lambda_x)\to(\mathsf{A}_r)_x.
\end{equation}
Let us consider for the remainder a minimal coarse-grained lattice $\mathcal L'$, i.e.~with four cells. eq.~\eqref{eq:Vtilderen} states that the evolution of the renormalised CA has eigenvalues equal to those of two layers of $G$, corresponding to eigenvectors of the form 
\begin{equation}
\ket{\Psi_\mathbf k}\coloneqq\ket{\psi_{k_0}}\ket{\psi_{k_1}}\ket{\psi_{k_2}}\ket{\psi_{k_3}}.
\end{equation}
The corresponding eigenvector is then given by
\begin{equation}
\ket{\mathbf k}\coloneqq\ket{k_0}_r\ket{k_1}_r\ket{k_2}_r\ket{k_3}_r.
\end{equation}

Let us begin computing the case corresponding to $G$ factorised, i.e. $G=U\otimes U$. In this case, in order to compute the renormalised evolution it is enough to specify the action over the single $\ket{\psi_k}$. In the basis identified by $P$, we have 
\begin{equation}
U^2=\begin{pmatrix} e^{i\lambda}& 0\\ 
0 & e^{-i\lambda}\end{pmatrix}\ ,
\end{equation}
whose action on $\ket{\psi_{k}}=\ket{l(k)}\ket{j(k)}$ is
\begin{equation}
(U^2\otimes U^2)\ket{\psi_k}=e^{i[(-1)^{l(k)}+(-1)^{j(k)}]\lambda}\ket{\psi_k}.
\end{equation} 
Thus the renormalised evolution will be a local unitary with eigenstates given by the basis chosen in $\eqref{eq:cgproj}$, whose diagonal form is
\begin{equation}
V=\begin{pmatrix} e^{i\delta_{l(0),j(0)}(-1)^{l(0)}2\lambda} & 0\\
0& e^{i\delta_{l(1),j(1)}(-1)^{l(1)}2\lambda} \end{pmatrix}.
\end{equation}
Notice that, since the renormalised evolution have phases proportional to an integer multiple of $\lambda$, if $\lambda$ is an irrational number, the iteration of the renormalisation will reach a fixed point (the identity CA) only if one chooses 
$l(k)=j(k)\oplus1$ for both $k=0,1$.

Let us now consider the case of $G$ of the form
\begin{align}
\begin{aligned}
G&=C_\phi (R_z(\theta)\otimes R_z(\theta)) C_\phi=\\
&=\begin{pmatrix} e^{2i\theta} & 0 &0 & 0\\
0& 1 & 0 & 0\\
0& 0 & 1 & 0\\
0& 0 & 0 & e^{-2i\theta+2i\phi}\end{pmatrix}\,.
\end{aligned}
\end{align}
Every $G$ in the first layer of $\tilde{U}^2_{\tT}$ will act on the second qubit of a cell (say $\alpha$) and the first one of the next cell ($\alpha\oplus1$).

It is clear that these operators represent commuting gates between neighbouring cells in the lattice $\mathcal L'$, thus we can artificially split them in two layers, e.g.~$(\tilde G_{12}\otimes\tilde G_{30})(\tilde G_{01}\otimes \tilde G_{23})$ where
\begin{align}
&\tilde G\ket{k_\alpha}\ket{k_{\alpha\oplus 1}}=\omega(k_\alpha,k_{\alpha\oplus1})\ket{k_\alpha}\ket{k_{\alpha\oplus 1}},\\
&\omega(k_\alpha,k_{\alpha\oplus1})\coloneqq e^{i\{[ (-1)^{j(k_\alpha)}+ (-1)^{l(k_{\alpha\oplus1})}]\theta +l(k_{\alpha\oplus1})j(k_\alpha)2\phi\}}.
\end{align}
After the application of these gates, the vector $\ket{\mathbf k}$ acquires a phase
\begin{align*}
\omega_1(\mathbf k)=\prod_{\alpha=0}^3\omega(k_\alpha,k_{\alpha\oplus1}).
\end{align*}
Now, the second layer of $G$ operators will act on the two qubits in the same cell, and can be interpreted as a layer of local gates acting on the single cells $\alpha$ as
\begin{align}
&\tilde R\ket{k_\alpha}=\nu(k_\alpha)\ket{k_\alpha},\\
&\nu(k_\alpha)\coloneqq e^{i\{[ (-1)^{j(k_\alpha)}+ (-1)^{l(k_{\alpha})}]\theta +l(k_{\alpha})j(k_\alpha)2\phi\}}.
\end{align}
After the application of these gates, the vector $\ket{\mathbf k}$ acquires a phase
\begin{align*}
\omega_2(\mathbf k)=\prod_{\alpha=0}^3\nu(k_\alpha).
\end{align*}
The final phase $\omega_1(\mathbf k)\omega_2(\mathbf k)$ can have different values depending on the functions $l$ and $j$. Let us consider the various cases, using the following equalities holding for $l\in\{0,1\}$
\begin{align*}
(-1)^l=-(2l-1),\quad l\oplus1=1-l.
\end{align*}
\begin{enumerate}
\item
$l(k)=j(k)=k$. In this case, neglecting an overall phase factor $e^{i16\theta}$, we have
\begin{align}
\omega_1(\mathbf k)\omega_2(\mathbf k)=\prod_{\alpha=0}^3e^{i[k_\alpha(2\phi-8\theta)+k_\alpha k_{\alpha\oplus 1}2\phi}]
\end{align}
This evolution corresponds to a QCA obtained with two layers of controled-phase gates $C_{\phi'}$ and one layer of local unitary gates given by $\exp(-i\sigma_z\theta')$, with 
\begin{align}
\begin{cases}
\phi'=2\phi,\\
\theta'=\phi-4\theta.
\end{cases}
\end{align}
\item $l(k)=j(k)=k\oplus 1$. In this case, neglecting an overall phase $e^{i16(\phi-\theta)}$, we have
\begin{align}
\omega_1(\mathbf k)\omega_2(\mathbf k)=\prod_{\alpha=0}^3e^{i[k_\alpha(8\theta-6\phi)+k_\alpha k_{\alpha\oplus 1}2\phi}].
\end{align}
This evolution corresponds to a QCA obtained with two layers of controled-phase gates $C_{\phi'}$ and one layer of local unitary gates given by $\exp(-i\sigma_z\theta')$, with 
\begin{align}
\begin{cases}
\phi'=2\phi,\\
\theta'=4\theta-3\phi.
\end{cases}
\end{align}
\item $l(k)=j(k)\oplus 1$. In this case, we have
\begin{align}
\omega_1(\mathbf k)\omega_2(\mathbf k)=\prod_{\alpha=0}^3e^{i[k_\alpha2\phi-k_\alpha k_{\alpha\oplus 1}2\phi}].
\end{align}
This evolution corresponds to a QCA obtained with two layers of controled-phase gates $C_{\phi'}$ and one layer of local unitary gates given by $\exp(-i\sigma_z\theta')$, with 
\begin{align}
\begin{cases}
\phi'=-2\phi,\\
\theta'=\phi.
\end{cases}
\end{align}
\item $l(k)=c_1$ or $j(k)=c_1$ constant. In this case, the other function will be either $k$ or $k\oplus1$, that we can express as $1/2+{s_2}(1/2-k)$ with $s_2=\pm1$, and we have, up to an irrelevant phase factor,
\begin{align}
\omega_1(\mathbf k)\omega_2(\mathbf k)=\prod_{\alpha=0}^3e^{i[s_2k_\alpha4(\theta+c_1\phi)]}.
\end{align}
This evolution corresponds to a QCA obtained one layer of local unitary gates given by $\exp(-i\sigma_z\theta')$, with 
\begin{align}
\theta'=2s_2(\theta+c_1\phi).
\end{align}
\end{enumerate}
Finally, consider the case
\begin{align*}
G=C_\phi (R_z(\theta)\vc{n}\cdot\vc{\sigma} R_z(\epsilon))\otimes(R_z(\theta)\vc{n}\cdot\vc{\sigma} R_z(\epsilon)) C_\phi.
\end{align*}
Remembering that $C_\phi$ is diagonal in the basis of eigenvectors of $\sigma_z$, we can always choose the $\sigma_y$ and $\sigma_x$ in such a way that $(R_z(\theta)\vc{n}\cdot\vc{\sigma} R_z(\epsilon))=\sigma_y$ and:
\begin{align*}
G=C_\phi \sigma_y\otimes\sigma_y C_\phi.
\end{align*}
Moreover, we can insert $\sigma_y\sigma_y=I$ on the left and get:
\begin{align}
G=\sigma_y\otimes\sigma_y(\sigma_y\otimes\sigma_y C_\phi \sigma_y\otimes\sigma_y)C_\phi.
\end{align}
We then have:
\begin{align*}
&\sigma_y\otimes\sigma_y C_\phi \sigma_y\otimes\sigma_y=\begin{pmatrix}e^{i\phi}&0&0&0\\0&1&0&0\\0&0&1&0\\0&0&0&1\end{pmatrix},\\
&\tilde{C}_\phi\coloneqq (\sigma_y\otimes\sigma_y C_\phi \sigma_y\otimes\sigma_y)C_\phi=\begin{pmatrix}e^{i\phi}&0&0&0\\0&1&0&0\\0&0&1&0\\0&0&0&e^{i\phi}\end{pmatrix}.
\end{align*}
Thus, we can write $G$ as:
\begin{align}
\begin{aligned}
G&=\sigma_y\otimes\sigma_y \tilde{C}_\phi=\\
&=\begin{pmatrix}0&0&0&e^{i(\phi+\pi)}\\0&0&1&0\\0&1&0&0\\e^{i(\phi+\pi)}&0&0&0\end{pmatrix}.
\end{aligned}
\end{align}
Following the steps of the previous calculation and considering the gates $\tilde{G}$ in $\mathcal{L'}$, we can write the action of the single $\tilde{G}$ on $\ket{k_\alpha}\ket{k_{\alpha\oplus 1}}$ as:
\begin{align}
&\tilde G\ket{k_\alpha}\ket{k_{\alpha\oplus 1}}=\omega(k_\alpha,k_{\alpha\oplus1})\ket{k_\alpha\oplus 1}\ket{k_{\alpha\oplus 1}\oplus 1},\\
&\omega(k_\alpha,k_{\alpha\oplus1})\coloneqq e^{i\{\phi+\pi\}\delta_{j(k_\alpha),l(k_{\alpha\oplus 1})}}.
\end{align}
Conversely to the previous case, we can see as $\tilde G$ changes the vector $\ket{k_\alpha}\ket{k_{\alpha\oplus 1}}$ to $\ket{k_{\alpha}\oplus 1}\ket{k_{\alpha\oplus 1}\oplus 1}$. Thus, we have to evaluate the action of the second layer of $\tilde G$ over the new vectors $\ket{k_{\alpha}\oplus 1}\ket{k_{\alpha\oplus 1}\oplus 1}$. Again, this second layer will act on the two qubits in the same cells, and thus we can interpret it as a layer of local gates. This gives:
\begin{align}
&\tilde R\ket{k_\alpha\oplus 1}=\nu(k_\alpha)\ket{k_\alpha},\\
&\nu(k_\alpha)\coloneqq e^{i\{\phi+\pi\}\delta_{j(k_\alpha)\oplus 1,l(k_{\alpha})\oplus 1}}= e^{i\{\phi+\pi\}\delta_{j(k_\alpha),l(k_{\alpha})}}.
\end{align}

We can then define the phases:
\begin{align}
&\omega_1(\mathbf{k})\coloneqq \prod_{\alpha=0}^3\omega(k_\alpha,k_{\alpha\oplus 1})\\
&\omega_2(\mathbf{k})\coloneqq \prod_{\alpha=0}^3\nu(k_\alpha).
\end{align}
This gives the final phase:
\begin{align}\label{eq:omegadelta}
\omega_1(\mathbf{k})\omega_2(\mathbf{k})=\prod_{\alpha=0}^3 e^{i\{\phi+\pi\}\delta_{j(k_\alpha),l(k_{\alpha})}}e^{i\{\phi+\pi\}\delta_{j(k_\alpha),l(k_{\alpha\oplus 1})}}.
\end{align}
Notice that the cases $j(k)=l(k)=k$ and $j(k)=l(k)=k\oplus 1$ give the same results. Moreover, the term $\omega_2(\mathbf k)$ contributes only if $j(k)=c$ (or equivalently if $l(k)=c$). Indeed, it is easy to see that we have:
\begin{align}
&\omega_2(\mathbf k)=1 && \mbox{if } j(k)\neq l(k)\\
&\omega_2(\mathbf k)=e^{i4(\pi+\phi)} && \mbox{if} j(k)=l(k).
\end{align}
We can now consider the different cases exploiting the equality:
\begin{align}
\delta_{ij}=1+2ij-i-j.
\end{align}
\begin{enumerate}
\item
$l(k)=j(k)$. We can consider without loss of generality the case $l(k)=k$. In this case, neglecting an overall phase factor $e^{i8\phi}=e^{i8\{\phi+\pi\}}$, we have
\begin{align}
\omega_1(\mathbf k)\omega_2(\mathbf k)=\omega_1(\mathbf{k})=\prod_{\alpha=0}^3e^{-2i\phi k_\alpha}e^{2i\phi k_\alpha k_{\alpha\oplus1}}
\end{align}
This evolution corresponds to a QCA obtained with two layers of controled-phase gates $C_{\phi'}$ and one layer of local unitary gates given by $\exp(-i\sigma_z\theta')$, with 
\begin{align}
\begin{cases}
\phi'=2\phi,\\
\theta'=-\phi.
\end{cases}
\end{align}
\item $l(k)=j(k)\oplus 1$. In this case,neglecting an overall phase factor $e^{i4\phi}=e^{i4\{\phi+\pi\}}$, we have
\begin{align}
\omega_1(\mathbf k)\omega_2(\mathbf k)= \omega_1(\mathbf k)=\prod_{\alpha=0}^3e^{2i\phi k_\alpha}e^{-2i\phi k_\alpha k_{\alpha\oplus1}}.
\end{align}
This evolution corresponds to a QCA obtained with two layers of controled-phase gates $C_{\phi'}$ and one layer of local unitary gates given by $\exp(-i\sigma_z\theta')$, with 
\begin{align}
\begin{cases}
\phi'=-2\phi,\\
\theta'=\phi.
\end{cases}
\end{align}
\item $l(k)=c_1$ or $j(k)=c_1$ constant. Let us assume without loss of generality that $l(k)=c_1$. It is easy to see that eq.~\eqref{eq:omegadelta} reduces to :
\begin{align*}
\omega_1(\mathbf{k})\omega_2(\mathbf{k})=\prod_{\alpha=0}^3e^{2i\{\phi+\pi\}\delta_{c_1,l(k)}}=\\
=\prod_{\alpha=0}^3e^{2i\{\phi+\pi\}\{2c_1l(k_\alpha)-c_1-l(k_\alpha)+1\}}.
\end{align*}
Again, we can express the two possible choices of $l(k)$ as $\frac{1}{2}+s_2(\frac{1}{2}-k)$ for $s_2=\pm 1$. Up to an irrelevant phase factor, this reduces to:
\begin{align*}
\omega_1(\mathbf{k})\omega_2(\mathbf{k})=\prod_{\alpha=0}^3e^{-2is_2(2c_1-1)k_\alpha\{\phi+\pi\}}.
\end{align*}
This evolution corresponds to a QCA obtained one layer of local unitary gates given by $\exp(-i\sigma_z\theta')$, with 
\begin{align}
\theta'=-s_2(2c_1-1)\phi.
\end{align}
\end{enumerate}

\end{document}

%% file: tikzstyle.tex
\usetikzlibrary {arrows.meta}
\usetikzlibrary{%
    decorations.pathreplacing,%
    decorations.pathmorphing,%
    arrows,
    arrows.meta,
    positioning,
    calc,
    shapes,
    shadows,
    shapes.geometric,
   math
    }

\definecolor{myblue1}{RGB}{0,125,209}
\definecolor{myblue2}{RGB}{0,155,239}
\definecolor{mycolor}{RGB}{100,155,239}
\definecolor{cifcolor}{RGB}{0,255,128}
\definecolor{Ucolor}{RGB}{102,102,255}
\definecolor{phasecolor}{RGB}{200,99,180}
\definecolor{phasecolorbis}{RGB}{230,89,140}
\definecolor{phasecolortris}{RGB}{250,79,120}
\definecolor{tassellationcol}{RGB}{240,150,20}
\definecolor{important}{RGB}{200,240,230}
\newcommand\bitsize{0.6 cm}
\newcommand\projspace{0.12cm}
\def\LR{\hspace{\projspace}$L$\hspace{\projspace} \nodepart{two}\hspace{\projspace}$R$\hspace{\projspace}}

\def\LRbar{\hspace{\projspace}$\bar{L}$\hspace{\projspace} \nodepart{two}\hspace{\projspace}$\bar{R}$\hspace{\projspace}}
\def\LtilR{\hspace{\projspace}$\tilde{L}$\hspace{\projspace} \nodepart{two}\hspace{\projspace}${R}$\hspace{\projspace}}
\def\LRtilde{\hspace{\projspace}${L}$\hspace{\projspace} \nodepart{two}\hspace{\projspace}$\tilde{R}$\hspace{\projspace}}
\newcommand\bitdistance{1 cm}
\newcommand\transfborder{0.1cm}
\newcommand\transfsize{\bitsize+\bitdistance-\transfborder}
\tikzmath{
\transfdistance=2*\bitdistance +2*\bitsize;
}
\tikzmath{
\transfoffset= 0.25*\bitdistance - \transfborder*0.5;
\halfsize=\bitsize*0.5;
\halfdistance=\bitdistance*0.5;
}
\newcommand\hdist{\bitsize+\bitdistance+\halfsize}
\newcommand\nhdist{-\bitsize-\bitdistance-\halfsize}
\newcommand\halfleft{-\bitsize-\transfoffset + \transfsize*0.5-\transfborder}
\tikzset{
  local/.style 2 args={
    rectangle, 
    rounded corners=1pt, 
    minimum height={#1},
    minimum width={#2},
    align=center, 
    line width=1pt,
    draw=white,
    font=\color{black}\sffamily, 
    fill=myblue1
  },
  local/.default={0.6 cm}{0.6 cm},
  M1/.style 2 args={
 	rectangle,
 	rounded corners=2pt ,
 	align=center,
 	minimum height={#1},
   	 minimum width={#2},
 	line width =1pt,
 	draw=black,
    font=\color{black}\sffamily, 
    fill=myblue1,
  },
  M1/.default={\bitsize}{\transfsize},
 M2/.style 2 args={
 	rectangle,
 	rounded corners=2pt ,
 	align=center,
 	minimum height={#1},
   	 minimum width={#2},
 	line width =1pt,
 	draw=black,
    font=\color{black}\sffamily, 
    fill=myblue2
  },
  M2/.default={\bitsize}{\transfsize},
 G/.style 2 args={
 	rectangle,
 	rounded corners=2pt ,
 	align=center,
 	minimum height={#1},
   	 minimum width={#2},
 	line width =1pt,
 	draw=black,
    font=\color{black}\sffamily, 
    fill=mycolor
  },
  G/.default={\bitsize}{\transfsize},
  bit/.style 2 args= {  
    rectangle, 
    rounded corners=5pt, 
    minimum height={#1},
    minimum width={#2},
    align=center, 
    line width=1pt,
    draw=black,
    font=\color{black}\sffamily, 
    fill=White
   },
 bit/.default={\bitsize}{\bitsize},
 bittil/.style 2 args= {  
    rectangle, 
    rounded corners=10pt,
    drop shadow, 
    minimum height={#1},
    minimum width={#1},
    align=center, 
    line width=1pt,
    draw=black,
    font=\color{black}\sffamily, 
    fill={#2}
   },
 tassel/.style 2 args= {  
    rectangle, 
    rounded corners=10pt,
    minimum height={#1},
    minimum width={#1},
    align=center, 
    line width=1pt,
    opacity={#2},
    draw=black,
    font=\color{black}\sffamily, 
    fill=tassellationcol
   },
  U/.style= {  
    rectangle, 
    rounded corners=1pt, 
    minimum height={\bitsize},
    minimum width={\bitsize},
    align=center, 
    line width=1pt,
    draw=black,
    fill=Ucolor,
    font=\color{black}\sffamily, 
   },
  LRbit/.style = {  
    rectangle, 
    rounded corners=1pt, 
    minimum height={\bitsize},
    minimum width={\transfsize},
    align=center, 
    line width=1pt,
    draw=black,
    fill=gray,
    font=\color{black}\sffamily, 
   },
   pi/.style 2 args={
 	rectangle,
 	rounded corners=2pt ,
 	align=center,
 	minimum height={#1},
   	 minimum width={#2},
 	line width =1pt,
 	draw=black,
	opacity=1,
    font=\color{black}\sffamily, 
    fill=yellow,
  },
  pi/.default={\bitsize}{\transfsize},
LR/.style ={
 	rectangle split,
	minimum height=0.6 cm,
   	minimum width =7.5cm,
	rectangle split horizontal,
	rectangle split parts=2,
 	rounded corners=2pt ,
 	align=center,
 	line width =1pt,
 	draw=black,
	opacity=1,
    font=\color{black}\sffamily, 
    fill=yellow,
 	text centered,
  },
  swap/.pic ={
	\draw (-\bitdistance,-1)--(-\bitdistance,-0.75)--(\bitdistance,0.75)--(\bitdistance,1);
	\draw (\bitdistance,-1)--(\bitdistance,-0.75)--(-\bitdistance,0.75)--(-\bitdistance,1);  
  },
Cif/.style 2 args={
 	rectangle,
 	rounded corners=2pt ,
 	align=center,
 	minimum height={#1},
   	 minimum width={#2},
 	line width =1pt,
 	draw=black,
    font=\color{black}\sffamily, 
    fill=cifcolor
  },
Cif/.default={\bitsize}{\transfsize},
phase/.style 2 args={
 	rectangle,
 	rounded corners=2pt ,
 	align=center,
 	minimum height={\bitsize},
   	 minimum width={#1},
 	line width =1pt,
 	draw=black,
    font=\color{black}\sffamily, 
    fill={#2}
  },
Op/.style 2 args={
	circle,
	draw=black,
	line width=1pt,
	font=\color{black}\sffamily ,
	minimum height={#1},
   	 minimum width={#2},
   	 fill=White,
 },
 Op/.default={\bitsize*0.25}{\bitsize*0.25},
 J/.pic={
  \draw (-\bitdistance,\nhdist)--(0,\nhdist*0.5)--(\bitdistance,\nhdist),},
 }

%% file: sample.tikzstyles
\usetikzlibrary {arrows.meta}
\usetikzlibrary{%
    decorations.pathreplacing,%
    decorations.pathmorphing,%
    arrows,
    arrows.meta,
    positioning,
    calc,
    shapes,
    shadows,
    shapes.geometric
    }
    
\tikzset{
  bit/.style={
  rectangle,
  minimum height=0.4cm,
  minimum width=0.4cm,
  draw=black,
  line width=1pt
  }
}
\tikzstyle{new edge style 0}=[-, fill=white, draw={rgb,255: red,249; green,20; blue,0}]
\tikzstyle{new edge style 1}=[-, fill=black]
\tikzstyle{new edge style 2}=[-, fill={rgb,255: red,255; green,247; blue,0}, draw={rgb,255: red,162; green,164; blue,60}]
\tikzstyle{new edge style 3}=[-, fill={rgb,255: red,131; green,131; blue,131}]
\tikzstyle{new edge style 4}=[-, fill={rgb,255: red,0; green,209; blue,255}, draw={rgb,255: red,49; green,57; blue,208}]

%% file: figures/Leftsh.tikz.tex
\begin{tikzpicture}[baseline=(A0.base)]
\foreach \a in {0,...,6}{
\node (A\a ) at (\a cm,1.1) {};
\draw[black] ($(A\a)-(0,0.7)$)--($(A\a)-(0,0.2)$);
\draw[black] ($(A\a)-(0,0.2)$)--($(A\a)+(1,0.3)$);
\draw[black] ($(A\a)+(1,0.3)$)--($(A\a)+(1,0.7)$);
}
\node at ($(A0)-(0.8,0)$) {...};
\node at ($(A6)+(1.5,0)$) {...};
\end{tikzpicture}

%% file: figures/QCAexample.tikz.tex
\begin{tikzpicture}[baseline=(O.base)]
\node (O) at (0,\hdist) {};
\foreach \t in  {0,...,8}{
\draw[black]  (\bitdistance*\t+\bitsize*\t,6*\hdist) -- (\bitdistance*\t+\bitsize*\t, \bitsize*0.5+2pt) ;
\node[U] at  (\bitdistance*\t+\bitsize*\t,7*\hdist) {$\mathcal U$} ;
}
\foreach \t in {0,...,2}
\node[phase={\transfsize+\bitsize}{phasecolor}] at (\transfdistance*0.5+\transfdistance*\t+\bitsize*\t+\bitdistance*\t,0) {$\phi$};
\foreach \t in{0,...,3}{
\node[phase={\transfsize}{phasecolorbis}] at (\bitsize-\transfsize*0.5+\transfdistance*\t,\hdist) {$\phi'$};
}
\foreach \t in{0,1}{
\node[phase={2*\transfsize+\bitdistance}{phasecolortris}] at (\transfdistance*0.5+\bitdistance+\transfdistance*\t+2*\bitsize*\t+2*\bitdistance*\t,4*\hdist) {$\tilde\phi$};
}
\draw[fill=phasecolorbis,rounded corners=2pt,draw=white,opacity=1] (-\bitsize+\transfdistance*4+\transfsize-\transfoffset,0.0001*\hdist)--(-\bitsize+\transfdistance*4,0.0001*\hdist)--(-\bitsize+\transfdistance*4,2*\hdist)--(-\bitsize+\transfdistance*4+\transfsize-\transfoffset,2*\hdist);
\draw[black,rounded corners=2pt, line width=1pt,opacity=1]  (-\bitsize+\transfdistance*4+\transfsize-\transfoffset,0.0001*\hdist)--(-\bitsize+\transfdistance*4,0.0001*\hdist)--(-\bitsize+\transfdistance*4,2*\hdist)--(-\bitsize+\transfdistance*4+\transfsize-\transfoffset,2*\hdist);
\draw[fill=phasecolortris,rounded corners=2pt,draw=white,opacity=1] (-\bitsize+\transfdistance*4+\transfsize-\transfoffset,3*\hdist)--(-\bitsize+\transfdistance*4,3*\hdist)--(-\bitsize+\transfdistance*4,5*\hdist)--(-\bitsize+\transfdistance*4+\transfsize-\transfoffset,5*\hdist);
\draw[black,rounded corners=2pt, line width=1pt,opacity=1]  (-\bitsize+\transfdistance*4+\transfsize-\transfoffset,3*\hdist)--(-\bitsize+\transfdistance*4,3*\hdist)--(-\bitsize+\transfdistance*4,5*\hdist)--(-\bitsize+\transfdistance*4+\transfsize-\transfoffset,5*\hdist);
\end{tikzpicture}

%% file: figures/Qubitev.tikz
\begin{tikzpicture}[baseline=(O.base)]
\node (O) at (0,3.5*\hdist) {};
\foreach \t in  {0,...,5}{
\draw[black]  (\bitdistance*\t+\bitsize*\t,9*\hdist) -- (\bitdistance*\t+\bitsize*\t, \bitsize*0.5+2pt) ;
\node[U] at (\bitdistance*\t+\bitsize*\t,7*\hdist) {$\mathcal U$};
}
\foreach \t in {0,...,2}{
\node[Cif] at (\bitsize+\transfoffset+\transfdistance*\t,\hdist) {$\mathcal C_\phi$};
}
\foreach \t in {0,...,1}{
\node[Cif] at (\bitsize+\transfoffset+\transfsize*0.5+\transfdistance*\t,4*\hdist) {$\mathcal C_\phi$};
}

\draw[fill=cifcolor,rounded corners=2pt,draw=white,opacity=1] (\bitsize+\transfoffset+\transfsize*0.5+\transfdistance*2,3*\hdist)--(\bitsize+\transfoffset+\transfdistance*2,3*\hdist)--(\bitsize+\transfoffset+\transfdistance*2,5*\hdist)--(\bitsize+\transfoffset+\transfsize*0.5+\transfdistance*2,5*\hdist);
\draw[black,rounded corners=2pt, line width=1pt,opacity=1] (\bitsize+\transfoffset+\transfsize*0.5+\transfdistance*2,3*\hdist)--(\bitsize+\transfoffset+\transfdistance*2,3*\hdist)--(\bitsize+\transfoffset+\transfdistance*2,5*\hdist)--(\bitsize+\transfoffset+\transfsize*0.5+\transfdistance*2,5*\hdist);

\draw[fill=cifcolor,rounded corners=2pt,draw=white,opacity=1] (\halfleft,3*\hdist)--(-\bitsize-\transfoffset,3*\hdist)--(-\bitsize-\transfoffset,5*\hdist)--(\halfleft,5*\hdist);
\draw[black,rounded corners=2pt, line width=1pt,opacity=1] (-\bitsize-\transfoffset,3*\hdist)--(\halfleft,3*\hdist)--(\halfleft,5*\hdist)--(-\bitsize-\transfoffset,5*\hdist);

\end{tikzpicture}

%% file: figures/Lattice.tikz.tex
\begin{tikzpicture}
\foreach \t in {0,...,5}{
\foreach \s in {0,...,5}{
\node[bittil={1.5*\bitsize}{myblue1}] at (1.5*\t*\bitsize+1.5*\t*\bitdistance,1.5*\s*\bitsize+1.5*\s*\bitdistance) {\tiny $(\t,\s)$};}
}
\end{tikzpicture}

%% file: figures/LatticePart.tikz.tex
\begin{tikzpicture}
\foreach \t in {0,...,5}{
\foreach \s in {0,...,5}{
\node[bittil={1.5*\bitsize}{myblue1}] at (1.5*\t*\bitsize+2*\t*\bitdistance,1.5*\s*\bitsize+2*\s*\bitdistance) {\tiny $(\t,\s)$};}
}
\foreach \t in {0,...,2}{
\foreach \s in {0,...,2}{
\node[tassel={3*\bitsize+\bitdistance}{0.5}] at (\bitdistance+\bitsize+4.5*\t*\bitsize+3*\t*\bitdistance,1.1*\bitdistance+\bitsize+4.5*\s*\bitsize+3*\s*\bitdistance) {};
\node at  (\bitdistance+\bitsize+4.5*\t*\bitsize+3*\t*\bitdistance,1.1*\bitdistance+\bitsize+4.5*\s*\bitsize+3*\s*\bitdistance) {\tiny $\Lambda_{\t,\s}$};
}}
\end{tikzpicture}

%% file: figures/Latticeprimbis.tikz
\begin{tikzpicture}
\foreach \t in {0,1}{
\foreach \s in {0,1}{
\node[bittil={1.5*\bitsize}{myblue1}] at (-5cm+1.5*\t*\bitsize+2*\t*\bitdistance,1.5*\s*\bitsize+2*\s*\bitdistance) {};}
}
\node[tassel={3*\bitsize+\bitdistance}{0.5}] (T) at (-5cm +\bitdistance+\bitsize,1.1*\bitdistance+\bitsize) {};
\node at  (T) {\tiny $\Lambda_{x,y}$};
\node[bittil={1.5*\bitsize}{myblue2}] (Tprim) at ($(T)+(4*\bitdistance+4*\bitsize,0)$) {$(x,y)$};
\draw[->,line width=1pt] ($(T)+(1.5cm+1.5*\bitsize+0.5*\bitdistance,0)$)--($(Tprim)-(0.6cm + \bitsize,0)$);
\end{tikzpicture}

%% file: figures/J.tikz.tex
\begin{tikzpicture}[baseline=(O.base)]
\node (O) at (0,\hdist) {};
\draw[fill=myblue1]  (-3,0) circle [radius=2cm];
\draw[fill=tassellationcol] (5,0) circle [radius=4cm];
\draw[fill=myblue2] (5,0) circle [radius=2cm];
\node at (5,0) {$\mathsf{A}\lvert_\Pi(\mathcal{L})$};
\node at (-3,0) {$\mathsf{A}(\mathcal{L'})$};
\node at (5,3) { $\mathsf{A}(\mathcal{L})$};
\draw[->, line width=1pt] (-3,2) to [bend left] node [anchor=south] {$\tV$} (5,2);
\draw[->, line width=1pt] (5,-2) to [bend left] node [anchor=north] {$\tJ$} (-3,-2);
\end{tikzpicture}

%% file: figures/LocUnit.tikz
\begin{tikzpicture}
\foreach \t in  {0,...,5}{
\draw[black]  (\bitdistance*\t+1.5*\bitsize*\t,6*\hdist) -- (\bitdistance*\t+1.5*\bitsize*\t, \bitsize*0.5+2pt) ;
\node at (\bitdistance*\t+1.5*\bitsize*\t,-0.1) {\t}; 
}
\foreach \t in {0,...,5}{
\node[U] at (1.5*\bitsize*\t+\bitdistance*\t,0.0001*\hdist) {$U$};
\node[U] at (1.5*\bitsize*\t+\bitdistance*\t,3*\hdist) {$U$};
}
\foreach \t in {0,...,2}{
\node[tassel={1.25*\bitsize+\bitdistance}{0.5}] at (0.75*\bitsize+0.5*\bitdistance+\bitsize*2.5*\t+\bitdistance*2.25*\t, 1.5*\hdist) { };
}
\end{tikzpicture}

%% file: figures/LocUnitprim.tikz.tex
\begin{tikzpicture}
\foreach \t in  {0,...,2}{
\draw[black]  (\bitdistance*\t+1.5*\bitsize*\t,3*\hdist) -- (\bitdistance*\t+1.5*\bitsize*\t, \bitsize*0.5+2pt) ;
\node at (\bitdistance*\t+1.5*\bitsize*\t,-0.1) {\t}; 
}
\foreach \t in {0,...,2}{
\node[U] at (1.5*\bitsize*\t+\bitdistance*\t,0.0001*\hdist) {$U'$};
}

\end{tikzpicture}

%% file: figures/Leftshtwo.tikz.tex
\begin{tikzpicture}[baseline=(A0.base)]
\foreach \a in {0,...,7}{
\node (A\a ) at (\bitsize*\a+\bitdistance*\a ,0.001*\hdist) {};
\draw[black, line width=1pt] ($(A\a)-(\bitsize*0.5+\bitdistance*0.5,\hdist)$)--($(A\a)-(\bitsize*0.5+\bitdistance*0.5,0.25*\hdist)$);
\node at (\bitdistance*\a+\bitsize*\a-\bitsize*0.5-\bitdistance*0.5,-1.4) {\huge \a}; 
\draw[black, line width=1pt] ($(A\a)-(\bitsize*0.5+\bitdistance*0.5,0.25*\hdist)$)--($(A\a)+(\bitsize*0.5+\bitdistance*0.5,0)$);
\draw[black, line width=1pt] ($(A\a)+(\bitsize*0.5+\bitdistance*0.5,0)$)--($(A\a)+(\bitsize*0.5+\bitdistance*0.5,0.6)$);
\node (B\a ) at ($(A\a)+(0,2*\hdist)$) {};
\draw[black, line width=1pt] ($(B\a)-(\bitsize*0.5+\bitdistance*0.5,\hdist)$)--($(B\a)-(\bitsize*0.5+\bitdistance*0.5,0.25*\hdist)$);
\draw[black, line width=1pt] ($(B\a)-(\bitsize*0.5+\bitdistance*0.5,0.25*\hdist)$)--($(B\a)+(\bitsize*0.5+\bitdistance*0.5,0)$);
\draw[black, line width=1pt] ($(B\a)+(\bitsize*0.5+\bitdistance*0.5,0)$)--($(B\a)+(\bitsize*0.5+\bitdistance*0.5,0.005*\hdist)$);
}
\foreach \t in {0,...,3}{
\node[tassel={1.2*\bitsize+0.9*\bitdistance}{0.5}] at (\bitsize*0.15+\bitdistance*0.15+2*\bitsize*\t+2*\bitdistance*\t,0.45*\hdist) {};
}
\node at ($(A0)-(2,\nhdist)$) {\huge ...};
\node at ($(A7)+(1.8,\hdist)$) {\huge ...};
\end{tikzpicture}

%% file: figures/Leftshmod.tikz.tex
\begin{tikzpicture}[baseline=(A0.base)]
\foreach \a in {0,...,3}{
\node (A\a ) at (\bitsize*\a+\bitdistance*\a ,0.001*\hdist) {};
\draw[black, line width=1pt] ($(A\a)-(\bitsize*0.5+\bitdistance*0.5,\hdist)$)--($(A\a)-(\bitsize*0.5+\bitdistance*0.5,0.25*\hdist)$);
\node at (\bitdistance*\a+\bitsize*\a-\bitsize*0.5-\bitdistance*0.5,-1) {\a}; 
\draw[black, line width=1pt] ($(A\a)-(\bitsize*0.5+\bitdistance*0.5,0.25*\hdist)$)--($(A\a)+(\bitsize*0.5+\bitdistance*0.5,0)$);
\draw[black, line width=1pt] ($(A\a)+(\bitsize*0.5+\bitdistance*0.5,0)$)--($(A\a)+(\bitsize*0.5+\bitdistance*0.5,0.6)$);
}
\node at ($(A0)-(2,0.1)$) {\huge ...};
\node at ($(A3)+(1.8,-0.1)$) {\huge ...};
\end{tikzpicture}

%% file: figures/TnPi.tikz.tex
\begin{tikzpicture}[baseline=(O.base)]
\node (O) at (0,0.25*\hdist) {};
\foreach \t in  {0,...,2}{
\draw[black]  (\bitdistance*\t+\bitsize*\t,6*\hdist) -- (\bitdistance*\t+\bitsize*\t, \bitsize*0.5+2pt) ;
}
\node[bit]  at (0,0)   {\tiny $\Pi_0$};
\node[bit] at (\bitdistance+\bitsize,0) {\tiny $\Pi_{{\mathcal{M}}}$};
\node[bit] at (2*\bitdistance+2*\bitsize,0) {\tiny $\Pi_{\mathcal W}$};
\node[M1] at (\bitsize+\transfoffset+\transfsize*0.5,4*\hdist) {$V$};
\node[M2] at (\bitsize+\transfoffset,\hdist) {$U$};
\end{tikzpicture}

%% file: figures/PiLact.tikz.tex
\begin{tikzpicture}[baseline=(O.base)]
\node (O) at (0,0.25*\hdist) {};
\foreach \t in  {0,...,2}{
\draw[black]  (\bitdistance*\t+\bitsize*\t,6*\hdist) -- (\bitdistance*\t+\bitsize*\t, \bitsize*0.5+2pt) ;
}
\node[bit]  at (0,0)   {\tiny$\Pi_0$};
\node[bit] at (\bitdistance+\bitsize,0) {\tiny $\Pi_{{\mathcal{M}}}$};
\node[bit] at (2*\bitdistance+2*\bitsize,0) {\tiny $\Pi_{\mathcal W}$};
\end{tikzpicture}

%% file: figures/TnPiLeft.tikz.tex
\begin{tikzpicture}[baseline=(O.base)]
\node (O) at (0,0.25*\hdist) {};
\foreach \t in  {0,...,2}{
\draw[black]  (\bitdistance*\t+\bitsize*\t,3*\hdist) -- (\bitdistance*\t+\bitsize*\t, \bitsize*0.5+2pt) ;
}
\node[bit]  at (0,0)   {\tiny $\Pi_0$};
\node[bit] at (\bitdistance+\bitsize,0) {\tiny $\Pi_{{\mathcal{M}}}$};
\node[bit] at (2*\bitdistance+2*\bitsize,0) {\tiny $\Pi_{\mathcal W}$};
\node[M2] at (\bitsize+\transfoffset,\hdist) {$U$};
\end{tikzpicture}

%% file: figures/TnPiRight.tikz.tex
\begin{tikzpicture}[baseline=(O.base)]
\node (O) at (0,0.25*\hdist) {};
\foreach \t in  {0,...,2}{
\draw[black]  (\bitdistance*\t+\bitsize*\t,3*\hdist) -- (\bitdistance*\t+\bitsize*\t, \bitsize*0.5+2pt) ;
}
\node[bit]  at (0,0)   {\tiny $\Pi_0$};
\node[bit] at (\bitdistance+\bitsize,0) {\tiny $\Pi_{{\mathcal{M}}}$};
\node[bit] at (2*\bitdistance+2*\bitsize,0) {\tiny $\Pi_{\mathcal W}$};
\node[M1] at (\bitsize+\transfoffset+\transfsize*0.5,\hdist) {$V^\dag$};
\end{tikzpicture}

%% file: figures/TnPish.tikz.tex
\begin{tikzpicture}[baseline=(O.base)]
\node (O) at (0,0.25*\hdist) {};
\foreach \t in  {0,...,1}{
\draw[black]  (\bitdistance*\t+\bitsize*\t,3*\hdist) -- (\bitdistance*\t+\bitsize*\t, \bitsize*0.5+2pt) ;
}
\node[bit]  at (0,0)   {\tiny $\Pi_0$};
\node[bit] at (\bitdistance+\bitsize,0) {\tiny $\Pi_{\mathcal{M}}$};
\node[M2] at (\bitsize+\transfoffset,\hdist) {$U$};
\end{tikzpicture}

%% file: figures/TnPiZen.tikz.tex
\begin{tikzpicture}[baseline=(O.base)]
\node (O) at (0,0.25*\hdist) {};
\foreach \t in  {0,...,1}{
\draw[black]  (\bitdistance*\t+\bitsize*\t,3*\hdist) -- (\bitdistance*\t+\bitsize*\t, \bitsize*0.5+2pt) ;
}
\node[bit]  at (0,0)   {\tiny$\Pi_0$};
\node[bit] at (\bitdistance+\bitsize,0) {\tiny $\Pi_{\mathcal{M}}$};
\node[U] at (\bitsize+\bitdistance,\hdist) {$Z^\dag$};
\end{tikzpicture}

%% file: figures/MargInd1.tikz.tex
\begin{tikzpicture}[baseline=(O.base)]
\node (O) at (0,\hdist) {};
\foreach \t in  {0,...,5}{
\draw[black]  (\bitdistance*\t+\bitsize*\t,6*\hdist) -- (\bitdistance*\t+\bitsize*\t, \bitsize*0.5+2pt) ;
}
\foreach \t in {0,...,2}{
\node[M1] at (\bitsize+\transfoffset+\transfdistance*\t,\hdist) {$M_1$};
}
\foreach \t in {0,...,1}{
\node[M2] at (\bitsize+\transfoffset+\transfsize*0.5+\transfdistance*\t,4*\hdist) {$M_2$};
}

\draw[fill=myblue2,rounded corners=2pt,draw=white,opacity=1] (\bitsize+\transfoffset+\transfsize*0.5+\transfdistance*2,3*\hdist)--(\bitsize+\transfoffset+\transfdistance*2,3*\hdist)--(\bitsize+\transfoffset+\transfdistance*2,5*\hdist)--(\bitsize+\transfoffset+\transfsize*0.5+\transfdistance*2,5*\hdist);
\draw[black,rounded corners=2pt, line width=1pt,opacity=1] (\bitsize+\transfoffset+\transfsize*0.5+\transfdistance*2,3*\hdist)--(\bitsize+\transfoffset+\transfdistance*2,3*\hdist)--(\bitsize+\transfoffset+\transfdistance*2,5*\hdist)--(\bitsize+\transfoffset+\transfsize*0.5+\transfdistance*2,5*\hdist);

\draw[fill=myblue2,rounded corners=2pt,draw=white,opacity=1] (\halfleft,3*\hdist)--(-\bitsize-\transfoffset,3*\hdist)--(-\bitsize-\transfoffset,5*\hdist)--(\halfleft,5*\hdist);
\draw[black,rounded corners=2pt, line width=1pt,opacity=1] (-\bitsize-\transfoffset,3*\hdist)--(\halfleft,3*\hdist)--(\halfleft,5*\hdist)--(-\bitsize-\transfoffset,5*\hdist);

\end{tikzpicture}

%% file: figures/SepevConj.tikz
\begin{tikzpicture}[baseline=(O.base)]
\node (O) at (0,2*\hdist) {};
\foreach \t in  {0,...,5}{
\draw[black]  (\bitdistance*\t+\bitsize*\t,13*\hdist) -- (\bitdistance*\t+\bitsize*\t, \bitsize*0.5+2pt) ;
}
\foreach \t in {0,...,2}{
\node[LR] at (\bitsize+\transfoffset+\transfdistance*\t,0) {\LR};
}
\foreach \t in {0,...,1}{
\node[M2] at (\bitsize+\transfoffset+\transfsize*0.5+\transfdistance*\t,\hdist) {$M_2$};
}
\foreach \t in {0,...,1}{
\node[M1] at (\bitsize+\transfoffset+\transfdistance*\t+\transfsize*0.5,4*\hdist) {$M_1$};
}
\foreach \t in {0,...,2}{
\node[M2] at (\bitsize+\transfoffset+\transfdistance*\t,7*\hdist) {$M_2$};
}
\foreach \t in {0,...,2}{
\node[M1] at (\bitsize+\transfoffset+\transfdistance*\t,10*\hdist) {$M_1$};
}

\draw[fill=myblue2,rounded corners=2pt,draw=white,opacity=1] (\bitsize+\transfoffset+\transfsize*0.5+\transfdistance*2,\bitdistance+\halfsize)--(\bitsize+\transfoffset+\transfdistance*2,\bitdistance+\halfsize)--(\bitsize+\transfoffset+\transfdistance*2,2*\hdist)--(\bitsize+\transfoffset+\transfsize*0.5+\transfdistance*2,2*\hdist);
\draw[black,rounded corners=2pt, line width=1pt,opacity=1] (\bitsize+\transfoffset+\transfsize*0.5+\transfdistance*2,\bitdistance+\halfsize)--(\bitsize+\transfoffset+\transfdistance*2,\bitdistance+\halfsize)--(\bitsize+\transfoffset+\transfdistance*2,2*\hdist)--(\bitsize+\transfoffset+\transfsize*0.5+\transfdistance*2,2*\hdist);
\draw[fill=myblue1,rounded corners=2pt,draw=white,opacity=1] (\bitsize+\transfoffset+\transfsize*0.5+\transfdistance*2,3*\hdist)--(\bitsize+\transfoffset+\transfdistance*2,3*\hdist)--(\bitsize+\transfoffset+\transfdistance*2,5*\hdist)--(\bitsize+\transfoffset+\transfsize*0.5+\transfdistance*2,5*\hdist);
\draw[black,rounded corners=2pt, line width=1pt,opacity=1] (\bitsize+\transfoffset+\transfsize*0.5+\transfdistance*2,3*\hdist)--(\bitsize+\transfoffset+\transfdistance*2,3*\hdist)--(\bitsize+\transfoffset+\transfdistance*2,5*\hdist)--(\bitsize+\transfoffset+\transfsize*0.5+\transfdistance*2,5*\hdist);

\draw[fill=myblue2,rounded corners=2pt,draw=white,opacity=1] (\halfleft,\bitdistance+\halfsize)--(-\bitsize-\transfoffset,\bitdistance+\halfsize)--(-\bitsize-\transfoffset,2*\hdist)--(\halfleft,2*\hdist);
\draw[black,rounded corners=2pt, line width=1pt,opacity=1] (-\bitsize-\transfoffset,\bitdistance+\halfsize)--(\halfleft,\bitdistance+\halfsize)--(\halfleft,2*\hdist)--(-\bitsize-\transfoffset,2*\hdist);
\draw[fill=myblue1,rounded corners=2pt,draw=white,opacity=1] (\halfleft,3*\hdist)--(-\bitsize-\transfoffset,3*\hdist)--(-\bitsize-\transfoffset,5*\hdist)--(\halfleft,5*\hdist);
\draw[black,rounded corners=2pt, line width=1pt,opacity=1] (-\bitsize-\transfoffset,3*\hdist)--(\halfleft,3*\hdist)--(\halfleft,5*\hdist)--(-\bitsize-\transfoffset,5*\hdist);
\end{tikzpicture}

%% file: figures/RLchain.tikz
\begin{tikzpicture}[baseline=(O.base)]
\node (O) at (0,2*\hdist) {};
\foreach \t in  {0,...,5}{
\draw[black]  (\bitdistance*\t+\bitsize*\t,11*\hdist) -- (\bitdistance*\t+\bitsize*\t, \bitsize*0.5+2pt) ;
}
\foreach \t in {0,...,2}{
\node[LR] at (\bitsize+\transfoffset+\transfdistance*\t,0) {\LR};
}

\end{tikzpicture}

%% file: figures/Piseparated.tikz
\begin{tikzpicture}
\foreach \t in  {1,...,2}{
\draw[black]  (\bitdistance*\t+\bitsize*\t,1) -- (\bitdistance*\t+\bitsize*\t, \bitsize*0.5+2pt) ;
}
\node[LR] at (\bitsize+\transfoffset+\transfdistance*0.5,0) {\LR};
\end{tikzpicture}

%% file: figures/Pisep.tikz
\begin{tikzpicture}[baseline=(O.base)]
\node (O) at (0,0.25*\hdist) {};
\foreach \t in  {1,...,2}{
\draw[black]  (\bitdistance*\t+\bitsize*\t,3*\hdist) -- (\bitdistance*\t+\bitsize*\t, \bitsize*0.5+2pt) ;
}
\node[M1] at (\bitsize+\transfoffset+\transfdistance*0.5,\hdist) {$M_1$};
\node[pi] at  (\bitsize+\transfoffset+\transfdistance*0.5,0) {$\Pi$};
\end{tikzpicture}

%% file: figures/T2.tikz
\begin{tikzpicture}
	\begin{pgfonlayer}{nodelayer}
		\node [style=none] (125) at (-3.75, 0.05) {...};
		\node [style=none] (127) at (-2.9, 0.55) {};
		\node [style=none] (128) at (-1.15, 0.55) {};
		\node [style=none] (129) at (-1.15, -0.7) {};
		\node [style=none] (130) at (-2.9, -0.7) {};
		\node [style=none] (135) at (-2.025, -0.05) {$\Pi$};
		\node [style=none] (260) at (-0.875, 0.55) {};
		\node [style=none] (261) at (0.875, 0.55) {};
		\node [style=none] (262) at (0.875, -0.7) {};
		\node [style=none] (263) at (-0.875, -0.7) {};
		\node [style=none] (264) at (0.025, -0.05) {$\Pi$};
		\node [style=none] (314) at (1.125, 0.55) {};
		\node [style=none] (315) at (2.875, 0.55) {};
		\node [style=none] (316) at (2.875, -0.7) {};
		\node [style=none] (317) at (1.125, -0.7) {};
		\node [style=none] (318) at (2, -0.05) {$\Pi$};
		\node [style=none] (386) at (-2.7, 0.775) {};
		\node [style=none] (387) at (-1.25, 0.775) {};
		\node [style=none] (388) at (-1.45, 1.575) {};
		\node [style=none] (389) at (-0.425, 1.575) {};
		\node [style=none] (390) at (-1.475, 2.225) {};
		\node [style=none] (391) at (-0.45, 2.225) {};
		\node [style=none] (392) at (-2.7, 1.275) {};
		\node [style=none] (393) at (-1.25, 1.275) {};
		\node [scale=0.5] (394) at (-2, 1.025) {$M_1$};
		\node [style=none] (395) at (-2.325, 3.35) {};
		\node [style=none] (396) at (-1.5, 3.325) {};
		\node [style=none] (397) at (-0.45, 3.325) {};
		\node [style=none] (398) at (-2.375, 1.275) {};
		\node [style=none] (399) at (-2.375, 0.775) {};
		\node [style=none] (400) at (-1.425, 0.775) {};
		\node [style=none] (401) at (-1.45, 1.575) {};
		\node [style=none] (402) at (-1.45, 1.275) {};
		\node [style=none] (403) at (-1.7, 2.525) {};
		\node [style=none] (404) at (-0.25, 2.525) {};
		\node [style=none] (405) at (-1.7, 3.025) {};
		\node [style=none] (406) at (-0.25, 3.025) {};
		\node [scale=0.5] (407) at (-1, 2.775) {$M_2$};
		\node [style=none] (408) at (-1.5, 3.025) {};
		\node [style=none] (409) at (-1.475, 2.525) {};
		\node [style=none] (410) at (-0.45, 2.525) {};
		\node [style=none] (411) at (-0.45, 3.025) {};
		\node [style=none] (412) at (-0.725, 0.775) {};
		\node [style=none] (413) at (0.725, 0.775) {};
		\node [style=none] (414) at (-0.725, 1.275) {};
		\node [style=none] (415) at (0.725, 1.275) {};
		\node [scale=0.5] (416) at (-0.025, 1.025) {$M_1$};
		\node [style=none] (417) at (-0.425, 1.275) {};
		\node [style=none] (418) at (-0.425, 0.775) {};
		\node [style=none] (419) at (0.55, 0.775) {};
		\node [style=none] (420) at (0.525, 1.275) {};
		\node [style=none] (421) at (0.525, 1.575) {};
		\node [style=none] (422) at (1.55, 1.575) {};
		\node [style=none] (423) at (0.5, 2.225) {};
		\node [style=none] (424) at (1.525, 2.225) {};
		\node [style=none] (425) at (0.475, 3.325) {};
		\node [style=none] (426) at (1.525, 3.325) {};
		\node [style=none] (427) at (0.525, 1.575) {};
		\node [style=none] (428) at (0.275, 2.525) {};
		\node [style=none] (429) at (1.725, 2.525) {};
		\node [style=none] (430) at (0.275, 3.025) {};
		\node [style=none] (431) at (1.725, 3.025) {};
		\node [scale=0.5] (432) at (0.975, 2.775) {$M_2$};
		\node [style=none] (433) at (0.475, 3.025) {};
		\node [style=none] (434) at (0.5, 2.525) {};
		\node [style=none] (435) at (1.525, 2.525) {};
		\node [style=none] (436) at (1.525, 3.025) {};
		\node [style=none] (437) at (1.55, 1.575) {};
		\node [style=none] (438) at (1.25, 0.775) {};
		\node [style=none] (439) at (2.7, 0.775) {};
		\node [style=none] (440) at (1.25, 1.275) {};
		\node [style=none] (441) at (2.7, 1.275) {};
		\node [scale=0.5] (442) at (1.95, 1.025) {$M_1$};
		\node [style=none] (443) at (1.55, 1.275) {};
		\node [style=none] (444) at (1.55, 0.775) {};
		\node [style=none] (445) at (2.525, 0.775) {};
		\node [style=none] (446) at (2.5, 1.275) {};
		\node [style=none] (447) at (2.5, 1.575) {};
		\node [style=none] (448) at (2.5, 1.575) {};
		\node [style=none] (449) at (2.5, 3.325) {};
		\node [style=none] (450) at (-2.625, 3.375) {};
		\node [style=none] (451) at (-1.175, 3.375) {};
		\node [style=none] (452) at (-1.375, 4.175) {};
		\node [style=none] (453) at (-0.35, 4.175) {};
		\node [style=none] (454) at (-1.4, 4.8) {};
		\node [style=none] (455) at (-0.375, 4.8) {};
		\node [style=none] (456) at (-2.625, 3.875) {};
		\node [style=none] (457) at (-1.175, 3.875) {};
		\node [scale=0.5] (458) at (-1.925, 3.625) {$M_1$};
		\node [style=none] (460) at (-1.425, 5.925) {};
		\node [style=none] (461) at (-0.375, 5.925) {};
		\node [style=none] (462) at (-2.325, 3.875) {};
		\node [style=none] (464) at (-1.35, 3.375) {};
		\node [style=none] (465) at (-1.375, 4.175) {};
		\node [style=none] (466) at (-1.375, 3.875) {};
		\node [style=none] (467) at (-1.625, 5.05) {};
		\node [style=none] (468) at (-0.175, 5.05) {};
		\node [style=none] (469) at (-1.625, 5.55) {};
		\node [style=none] (470) at (-0.175, 5.55) {};
		\node [scale=0.5] (471) at (-0.925, 5.3) {$M_2$};
		\node [style=none] (472) at (-1.425, 5.55) {};
		\node [style=none] (473) at (-1.4, 5.05) {};
		\node [style=none] (474) at (-0.375, 5.05) {};
		\node [style=none] (475) at (-0.375, 5.55) {};
		\node [style=none] (476) at (-0.65, 3.375) {};
		\node [style=none] (477) at (0.8, 3.375) {};
		\node [style=none] (478) at (-0.65, 3.875) {};
		\node [style=none] (479) at (0.8, 3.875) {};
		\node [scale=0.5] (480) at (0.05, 3.625) {$M_1$};
		\node [style=none] (481) at (-0.35, 3.875) {};
		\node [style=none] (482) at (-0.35, 3.375) {};
		\node [style=none] (483) at (0.625, 3.375) {};
		\node [style=none] (484) at (0.6, 3.875) {};
		\node [style=none] (485) at (0.6, 4.175) {};
		\node [style=none] (486) at (1.625, 4.175) {};
		\node [style=none] (487) at (0.575, 4.8) {};
		\node [style=none] (488) at (1.6, 4.8) {};
		\node [style=none] (489) at (0.55, 5.925) {};
		\node [style=none] (490) at (1.6, 5.925) {};
		\node [style=none] (491) at (0.6, 4.175) {};
		\node [style=none] (492) at (0.35, 5.05) {};
		\node [style=none] (493) at (1.8, 5.05) {};
		\node [style=none] (494) at (0.35, 5.55) {};
		\node [style=none] (495) at (1.8, 5.55) {};
		\node [scale=0.5] (496) at (1.05, 5.3) {$M_2$};
		\node [style=none] (497) at (0.55, 5.55) {};
		\node [style=none] (498) at (0.575, 5.05) {};
		\node [style=none] (499) at (1.6, 5.05) {};
		\node [style=none] (500) at (1.6, 5.55) {};
		\node [style=none] (501) at (1.625, 4.175) {};
		\node [style=none] (502) at (1.325, 3.375) {};
		\node [style=none] (503) at (2.775, 3.375) {};
		\node [style=none] (504) at (1.325, 3.875) {};
		\node [style=none] (505) at (2.775, 3.875) {};
		\node [scale=0.5] (506) at (2.025, 3.625) {$M_1$};
		\node [style=none] (507) at (1.625, 3.875) {};
		\node [style=none] (508) at (1.625, 3.375) {};
		\node [style=none] (509) at (2.475, 3.375) {};
		\node [style=none] (510) at (2.475, 3.825) {};
		\node [style=none] (513) at (2.45, 5.925) {};
		\node [style=none] (514) at (-2.725, -0.975) {};
		\node [style=none] (515) at (-1.275, -0.975) {};
		\node [style=none] (516) at (-1.475, -1.775) {};
		\node [style=none] (517) at (-0.45, -1.775) {};
		\node [style=none] (518) at (-1.5, -2.4) {};
		\node [style=none] (519) at (-0.475, -2.4) {};
		\node [style=none] (520) at (-2.725, -1.475) {};
		\node [style=none] (521) at (-1.275, -1.475) {};
		\node [scale=0.5] (522) at (-2.025, -1.225) {$M_1^\dag$};
		\node [style=none] (523) at (-2.35, -3.5) {};
		\node [style=none] (524) at (-1.525, -3.475) {};
		\node [style=none] (525) at (-0.475, -3.475) {};
		\node [style=none] (526) at (-2.325, -1.475) {};
		\node [style=none] (527) at (-2.425, -0.975) {};
		\node [style=none] (528) at (-1.525, -0.975) {};
		\node [style=none] (529) at (-1.475, -1.775) {};
		\node [style=none] (530) at (-1.475, -1.475) {};
		\node [style=none] (531) at (-1.725, -2.7) {};
		\node [style=none] (532) at (-0.275, -2.7) {};
		\node [style=none] (533) at (-1.725, -3.2) {};
		\node [style=none] (534) at (-0.275, -3.2) {};
		\node [scale=0.5] (535) at (-1.025, -2.95) {$M_2^\dag$};
		\node [style=none] (536) at (-1.525, -3.2) {};
		\node [style=none] (537) at (-1.5, -2.7) {};
		\node [style=none] (538) at (-0.475, -2.7) {};
		\node [style=none] (539) at (-0.475, -3.2) {};
		\node [style=none] (540) at (-0.75, -0.975) {};
		\node [style=none] (541) at (0.7, -0.975) {};
		\node [style=none] (542) at (-0.75, -1.475) {};
		\node [style=none] (543) at (0.7, -1.475) {};
		\node [scale=0.5] (544) at (-0.05, -1.225) {$M_1^\dag$};
		\node [style=none] (545) at (-0.45, -1.475) {};
		\node [style=none] (546) at (-0.45, -0.975) {};
		\node [style=none] (547) at (0.525, -0.975) {};
		\node [style=none] (548) at (0.5, -1.475) {};
		\node [style=none] (549) at (0.5, -1.775) {};
		\node [style=none] (550) at (1.525, -1.775) {};
		\node [style=none] (551) at (0.475, -2.4) {};
		\node [style=none] (552) at (1.5, -2.4) {};
		\node [style=none] (553) at (0.45, -3.475) {};
		\node [style=none] (554) at (1.5, -3.475) {};
		\node [style=none] (555) at (0.5, -1.775) {};
		\node [style=none] (556) at (0.25, -2.7) {};
		\node [style=none] (557) at (1.7, -2.7) {};
		\node [style=none] (558) at (0.25, -3.2) {};
		\node [style=none] (559) at (1.7, -3.2) {};
		\node [scale=0.5] (560) at (0.95, -2.95) {$M_2^\dag$};
		\node [style=none] (561) at (0.45, -3.2) {};
		\node [style=none] (562) at (0.475, -2.7) {};
		\node [style=none] (563) at (1.5, -2.7) {};
		\node [style=none] (564) at (1.5, -3.2) {};
		\node [style=none] (565) at (1.525, -1.775) {};
		\node [style=none] (566) at (1.225, -0.975) {};
		\node [style=none] (567) at (2.675, -0.975) {};
		\node [style=none] (568) at (1.225, -1.475) {};
		\node [style=none] (569) at (2.675, -1.475) {};
		\node [scale=0.5] (570) at (1.925, -1.225) {$M_1^\dag$};
		\node [style=none] (571) at (1.525, -1.475) {};
		\node [style=none] (572) at (1.525, -0.975) {};
		\node [style=none] (573) at (2.5, -0.975) {};
		\node [style=none] (574) at (2.45, -1.475) {};
		\node [style=none] (575) at (2.475, -1.775) {};
		\node [style=none] (576) at (2.475, -1.775) {};
		\node [style=none] (577) at (2.475, -3.475) {};
		\node [style=none] (578) at (-2.65, -3.525) {};
		\node [style=none] (579) at (-1.2, -3.525) {};
		\node [style=none] (580) at (-1.4, -4.325) {};
		\node [style=none] (581) at (-0.375, -4.325) {};
		\node [style=none] (582) at (-1.425, -4.975) {};
		\node [style=none] (583) at (-0.4, -4.975) {};
		\node [style=none] (584) at (-2.65, -4.025) {};
		\node [style=none] (585) at (-1.2, -4.025) {};
		\node [scale=0.5] (586) at (-1.95, -3.775) {$M_1^\dag$};
		\node [style=none] (587) at (-2.35, -6.175) {};
		\node [style=none] (588) at (-1.45, -6.225) {};
		\node [style=none] (589) at (-0.4, -6.225) {};
		\node [style=none] (590) at (-2.35, -4.025) {};
		\node [style=none] (591) at (-1.375, -3.525) {};
		\node [style=none] (592) at (-1.4, -4.325) {};
		\node [style=none] (593) at (-1.4, -4.025) {};
		\node [style=none] (594) at (-1.65, -5.275) {};
		\node [style=none] (595) at (-0.2, -5.275) {};
		\node [style=none] (596) at (-1.65, -5.775) {};
		\node [style=none] (597) at (-0.2, -5.775) {};
		\node [scale=0.5] (598) at (-0.95, -5.525) {$M_2^\dag$};
		\node [style=none] (599) at (-1.45, -5.775) {};
		\node [style=none] (600) at (-1.425, -5.275) {};
		\node [style=none] (601) at (-0.4, -5.275) {};
		\node [style=none] (602) at (-0.4, -5.775) {};
		\node [style=none] (603) at (-0.675, -3.525) {};
		\node [style=none] (604) at (0.775, -3.525) {};
		\node [style=none] (605) at (-0.675, -4.025) {};
		\node [style=none] (606) at (0.775, -4.025) {};
		\node [scale=0.5] (607) at (0.025, -3.775) {$M_1^\dag$};
		\node [style=none] (608) at (-0.375, -4.025) {};
		\node [style=none] (609) at (-0.375, -3.525) {};
		\node [style=none] (610) at (0.6, -3.525) {};
		\node [style=none] (611) at (0.575, -4.025) {};
		\node [style=none] (612) at (0.575, -4.325) {};
		\node [style=none] (613) at (1.6, -4.325) {};
		\node [style=none] (614) at (0.55, -4.975) {};
		\node [style=none] (615) at (1.575, -4.975) {};
		\node [style=none] (616) at (0.525, -6.225) {};
		\node [style=none] (617) at (1.575, -6.225) {};
		\node [style=none] (618) at (0.575, -4.325) {};
		\node [style=none] (619) at (0.325, -5.275) {};
		\node [style=none] (620) at (1.775, -5.275) {};
		\node [style=none] (621) at (0.325, -5.775) {};
		\node [style=none] (622) at (1.775, -5.775) {};
		\node [scale=0.5] (623) at (1.025, -5.525) {$M_2^\dag$};
		\node [style=none] (624) at (0.525, -5.775) {};
		\node [style=none] (625) at (0.55, -5.275) {};
		\node [style=none] (626) at (1.575, -5.275) {};
		\node [style=none] (627) at (1.575, -5.775) {};
		\node [style=none] (628) at (1.6, -4.325) {};
		\node [style=none] (629) at (1.3, -3.525) {};
		\node [style=none] (630) at (2.75, -3.525) {};
		\node [style=none] (631) at (1.3, -4.025) {};
		\node [style=none] (632) at (2.75, -4.025) {};
		\node [scale=0.5] (633) at (2, -3.775) {$M_1^\dag$};
		\node [style=none] (634) at (1.6, -4.025) {};
		\node [style=none] (635) at (1.6, -3.525) {};
		\node [style=none] (636) at (2.575, -3.525) {};
		\node [style=none] (637) at (2.475, -4.025) {};
		\node [style=none] (639) at (2.45, -6.225) {};
		\node [style=none] (640) at (-2.375, 0.55) {};
		\node [style=none] (641) at (-1.425, 0.55) {};
		\node [style=none] (642) at (-1.525, -0.7) {};
		\node [style=none] (643) at (-2.425, -0.7) {};
		\node [style=none] (644) at (-0.45, 0.775) {};
		\node [style=none] (645) at (0.5, 0.775) {};
		\node [style=none] (646) at (-0.5, -0.975) {};
		\node [style=none] (647) at (0.4, -0.975) {};
		\node [style=none] (648) at (-0.45, 0.55) {};
		\node [style=none] (649) at (0.5, 0.55) {};
		\node [style=none] (650) at (0.4, -0.7) {};
		\node [style=none] (651) at (-0.5, -0.7) {};
		\node [style=none] (652) at (1.575, 0.775) {};
		\node [style=none] (653) at (2.525, 0.775) {};
		\node [style=none] (654) at (1.525, -0.975) {};
		\node [style=none] (655) at (2.425, -0.975) {};
		\node [style=none] (656) at (1.575, 0.55) {};
		\node [style=none] (657) at (2.525, 0.55) {};
		\node [style=none] (658) at (2.425, -0.7) {};
		\node [style=none] (659) at (1.525, -0.7) {};
		\node [style=none] (660) at (2.525, 1.575) {};
		\node [style=none] (662) at (2.5, 2.225) {};
		\node [style=none] (664) at (2.525, 1.575) {};
		\node [style=none] (665) at (2.275, 2.525) {};
		\node [style=none] (667) at (2.275, 3.025) {};
		\node [scale=0.5] (669) at (2.975, 2.775) {$M_2$};
		\node [style=none] (670) at (2.475, 3.025) {};
		\node [style=none] (671) at (2.5, 2.525) {};
		\node [style=none] (672) at (3, 3.025) {};
		\node [style=none] (673) at (3, 2.525) {};
		\node [style=none] (674) at (3.075, 1.95) {};
		\node [style=none] (675) at (-2.35, 1.575) {};
		\node [style=none] (676) at (-2.35, 1.575) {};
		\node [style=none] (677) at (-2.325, 3.375) {};
		\node [style=none] (678) at (-2.375, 1.575) {};
		\node [style=none] (679) at (-2.35, 2.225) {};
		\node [style=none] (680) at (-2.375, 1.575) {};
		\node [style=none] (681) at (-2.125, 2.525) {};
		\node [style=none] (682) at (-2.125, 3.025) {};
		\node [scale=0.5] (683) at (-2.825, 2.775) {$M_2$};
		\node [style=none] (684) at (-2.325, 3.025) {};
		\node [style=none] (685) at (-2.35, 2.525) {};
		\node [style=none] (686) at (-2.85, 3.025) {};
		\node [style=none] (687) at (-2.85, 2.525) {};
		\node [style=none] (688) at (-2.925, 1.95) {};
		\node [style=none] (691) at (2.475, 5.875) {};
		\node [style=none] (693) at (2.5, 4.725) {};
		\node [style=none] (694) at (2.475, 4.2) {};
		\node [style=none] (695) at (2.275, 5.025) {};
		\node [style=none] (696) at (2.275, 5.525) {};
		\node [scale=0.5] (697) at (2.975, 5.275) {$M_2$};
		\node [style=none] (698) at (2.475, 5.525) {};
		\node [style=none] (699) at (2.5, 5.025) {};
		\node [style=none] (700) at (3, 5.525) {};
		\node [style=none] (701) at (3, 5.025) {};
		\node [style=none] (702) at (3.05, 4.45) {};
		\node [style=none] (704) at (-2.3, 4.075) {};
		\node [style=none] (705) at (-2.3, 4.075) {};
		\node [style=none] (706) at (-2.275, 5.875) {};
		\node [style=none] (707) at (-2.325, 4.075) {};
		\node [style=none] (708) at (-2.3, 4.725) {};
		\node [style=none] (709) at (-2.325, 4.2) {};
		\node [style=none] (710) at (-2.075, 5.025) {};
		\node [style=none] (711) at (-2.075, 5.525) {};
		\node [scale=0.5] (712) at (-2.775, 5.275) {$M_2$};
		\node [style=none] (713) at (-2.275, 5.525) {};
		\node [style=none] (714) at (-2.3, 5.025) {};
		\node [style=none] (715) at (-2.8, 5.525) {};
		\node [style=none] (716) at (-2.8, 5.025) {};
		\node [style=none] (717) at (-2.875, 4.45) {};
		\node [style=none] (718) at (-2.4, -1.625) {};
		\node [style=none] (719) at (-2.4, -1.625) {};
		\node [style=none] (720) at (-2.375, -3.55) {};
		\node [style=none] (721) at (-2.425, -1.475) {};
		\node [style=none] (722) at (-2.4, -2.4) {};
		\node [style=none] (723) at (-2.425, -1.75) {};
		\node [style=none] (724) at (-2.175, -2.725) {};
		\node [style=none] (725) at (-2.175, -3.225) {};
		\node [scale=0.5] (726) at (-2.875, -2.975) {$M_2^\dag$};
		\node [style=none] (727) at (-2.375, -3.225) {};
		\node [style=none] (728) at (-2.4, -2.725) {};
		\node [style=none] (729) at (-2.9, -3.225) {};
		\node [style=none] (730) at (-2.9, -2.725) {};
		\node [style=none] (731) at (-2.975, -2.05) {};
		\node [style=none] (732) at (-2.35, -4.375) {};
		\node [style=none] (733) at (-2.35, -4.375) {};
		\node [style=none] (734) at (-2.325, -6.225) {};
		\node [style=none] (735) at (-2.375, -3.975) {};
		\node [style=none] (736) at (-2.35, -5.025) {};
		\node [style=none] (737) at (-2.375, -4.5) {};
		\node [style=none] (738) at (-2.125, -5.275) {};
		\node [style=none] (739) at (-2.125, -5.775) {};
		\node [scale=0.5] (740) at (-2.825, -5.525) {$M_2^\dag$};
		\node [style=none] (741) at (-2.325, -5.775) {};
		\node [style=none] (742) at (-2.35, -5.275) {};
		\node [style=none] (743) at (-2.85, -5.775) {};
		\node [style=none] (744) at (-2.85, -5.275) {};
		\node [style=none] (745) at (-2.925, -4.75) {};
		\node [style=none] (750) at (2.4, -3.55) {};
		\node [style=none] (752) at (2.425, -2.4) {};
		\node [style=none] (753) at (2.475, -1.75) {};
		\node [style=none] (754) at (2.2, -2.725) {};
		\node [style=none] (755) at (2.2, -3.225) {};
		\node [scale=0.5] (756) at (2.9, -2.975) {$M_2^\dag$};
		\node [style=none] (757) at (2.4, -3.225) {};
		\node [style=none] (758) at (2.425, -2.725) {};
		\node [style=none] (759) at (2.925, -3.225) {};
		\node [style=none] (760) at (2.925, -2.725) {};
		\node [style=none] (761) at (3, -2.05) {};
		\node [style=none] (762) at (4.075, 0.05) {...};
		\node [style=none] (763) at (2.45, -4.025) {};
		\node [style=none] (764) at (2.425, -4.95) {};
		\node [style=none] (765) at (2.45, -4.3) {};
		\node [style=none] (766) at (2.2, -5.275) {};
		\node [style=none] (767) at (2.2, -5.775) {};
		\node [scale=0.5] (768) at (2.9, -5.525) {$M_2^\dag$};
		\node [style=none] (769) at (2.45, -5.775) {};
		\node [style=none] (770) at (2.425, -5.275) {};
		\node [style=none] (771) at (2.925, -5.775) {};
		\node [style=none] (772) at (2.925, -5.275) {};
		\node [style=none] (773) at (3.025, -4.6) {};
	\end{pgfonlayer}
	\begin{pgfonlayer}{edgelayer}
		\draw (130.center) to (129.center);
		\draw (129.center) to (128.center);
		\draw (128.center) to (127.center);
		\draw (127.center) to (130.center);
		\draw (263.center) to (262.center);
		\draw (262.center) to (261.center);
		\draw (261.center) to (260.center);
		\draw (260.center) to (263.center);
		\draw (317.center) to (316.center);
		\draw (316.center) to (315.center);
		\draw (315.center) to (314.center);
		\draw (314.center) to (317.center);
		\draw (386.center) to (387.center);
		\draw (390.center) to (389.center);
		\draw (388.center) to (391.center);
		\draw (386.center) to (392.center);
		\draw (392.center) to (393.center);
		\draw (393.center) to (387.center);
		\draw (401.center) to (402.center);
		\draw (403.center) to (404.center);
		\draw (403.center) to (405.center);
		\draw (405.center) to (406.center);
		\draw (406.center) to (404.center);
		\draw (390.center) to (409.center);
		\draw (391.center) to (410.center);
		\draw (408.center) to (396.center);
		\draw (411.center) to (397.center);
		\draw (412.center) to (413.center);
		\draw (412.center) to (414.center);
		\draw (414.center) to (415.center);
		\draw (415.center) to (413.center);
		\draw (417.center) to (389.center);
		\draw (423.center) to (422.center);
		\draw (421.center) to (424.center);
		\draw (428.center) to (429.center);
		\draw (428.center) to (430.center);
		\draw (430.center) to (431.center);
		\draw (431.center) to (429.center);
		\draw (423.center) to (434.center);
		\draw (424.center) to (435.center);
		\draw (433.center) to (425.center);
		\draw (436.center) to (426.center);
		\draw (420.center) to (427.center);
		\draw (438.center) to (439.center);
		\draw (438.center) to (440.center);
		\draw (440.center) to (441.center);
		\draw (441.center) to (439.center);
		\draw (443.center) to (437.center);
		\draw (446.center) to (448.center);
		\draw (450.center) to (451.center);
		\draw (454.center) to (453.center);
		\draw (452.center) to (455.center);
		\draw (450.center) to (456.center);
		\draw (456.center) to (457.center);
		\draw (457.center) to (451.center);
		\draw (465.center) to (466.center);
		\draw (467.center) to (468.center);
		\draw (467.center) to (469.center);
		\draw (469.center) to (470.center);
		\draw (470.center) to (468.center);
		\draw (454.center) to (473.center);
		\draw (455.center) to (474.center);
		\draw (472.center) to (460.center);
		\draw (475.center) to (461.center);
		\draw (476.center) to (477.center);
		\draw (476.center) to (478.center);
		\draw (478.center) to (479.center);
		\draw (479.center) to (477.center);
		\draw (481.center) to (453.center);
		\draw (487.center) to (486.center);
		\draw (485.center) to (488.center);
		\draw (492.center) to (493.center);
		\draw (492.center) to (494.center);
		\draw (494.center) to (495.center);
		\draw (495.center) to (493.center);
		\draw (487.center) to (498.center);
		\draw (488.center) to (499.center);
		\draw (497.center) to (489.center);
		\draw (500.center) to (490.center);
		\draw (484.center) to (491.center);
		\draw (502.center) to (503.center);
		\draw (502.center) to (504.center);
		\draw (504.center) to (505.center);
		\draw (505.center) to (503.center);
		\draw (507.center) to (501.center);
		\draw (514.center) to (515.center);
		\draw (518.center) to (517.center);
		\draw (516.center) to (519.center);
		\draw (514.center) to (520.center);
		\draw (520.center) to (521.center);
		\draw (521.center) to (515.center);
		\draw (529.center) to (530.center);
		\draw (531.center) to (532.center);
		\draw (531.center) to (533.center);
		\draw (533.center) to (534.center);
		\draw (534.center) to (532.center);
		\draw (518.center) to (537.center);
		\draw (519.center) to (538.center);
		\draw (536.center) to (524.center);
		\draw (539.center) to (525.center);
		\draw (540.center) to (541.center);
		\draw (540.center) to (542.center);
		\draw (542.center) to (543.center);
		\draw (543.center) to (541.center);
		\draw (545.center) to (517.center);
		\draw (551.center) to (550.center);
		\draw (549.center) to (552.center);
		\draw (556.center) to (557.center);
		\draw (556.center) to (558.center);
		\draw (558.center) to (559.center);
		\draw (559.center) to (557.center);
		\draw (551.center) to (562.center);
		\draw (552.center) to (563.center);
		\draw (561.center) to (553.center);
		\draw (564.center) to (554.center);
		\draw (548.center) to (555.center);
		\draw (566.center) to (567.center);
		\draw (566.center) to (568.center);
		\draw (568.center) to (569.center);
		\draw (569.center) to (567.center);
		\draw (571.center) to (565.center);
		\draw (574.center) to (576.center);
		\draw (578.center) to (579.center);
		\draw (582.center) to (581.center);
		\draw (580.center) to (583.center);
		\draw (578.center) to (584.center);
		\draw (584.center) to (585.center);
		\draw (585.center) to (579.center);
		\draw (592.center) to (593.center);
		\draw (594.center) to (595.center);
		\draw (594.center) to (596.center);
		\draw (596.center) to (597.center);
		\draw (597.center) to (595.center);
		\draw (582.center) to (600.center);
		\draw (583.center) to (601.center);
		\draw (599.center) to (588.center);
		\draw (602.center) to (589.center);
		\draw (603.center) to (604.center);
		\draw (603.center) to (605.center);
		\draw (605.center) to (606.center);
		\draw (606.center) to (604.center);
		\draw (608.center) to (581.center);
		\draw (614.center) to (613.center);
		\draw (612.center) to (615.center);
		\draw (619.center) to (620.center);
		\draw (619.center) to (621.center);
		\draw (621.center) to (622.center);
		\draw (622.center) to (620.center);
		\draw (614.center) to (625.center);
		\draw (615.center) to (626.center);
		\draw (624.center) to (616.center);
		\draw (627.center) to (617.center);
		\draw (611.center) to (618.center);
		\draw (629.center) to (630.center);
		\draw (629.center) to (631.center);
		\draw (631.center) to (632.center);
		\draw (632.center) to (630.center);
		\draw (634.center) to (628.center);
		\draw (640.center) to (399.center);
		\draw (400.center) to (641.center);
		\draw (642.center) to (528.center);
		\draw (643.center) to (527.center);
		\draw (648.center) to (644.center);
		\draw (645.center) to (649.center);
		\draw (650.center) to (647.center);
		\draw (651.center) to (646.center);
		\draw (656.center) to (652.center);
		\draw (653.center) to (657.center);
		\draw (658.center) to (655.center);
		\draw (659.center) to (654.center);
		\draw (665.center) to (667.center);
		\draw (662.center) to (671.center);
		\draw (665.center) to (673.center);
		\draw (667.center) to (672.center);
		\draw (674.center) to (662.center);
		\draw (674.center) to (664.center);
		\draw (670.center) to (509.center);
		\draw (681.center) to (682.center);
		\draw (679.center) to (685.center);
		\draw (681.center) to (687.center);
		\draw (682.center) to (686.center);
		\draw (688.center) to (679.center);
		\draw (688.center) to (680.center);
		\draw (684.center) to (677.center);
		\draw (398.center) to (680.center);
		\draw (695.center) to (696.center);
		\draw (693.center) to (699.center);
		\draw (695.center) to (701.center);
		\draw (696.center) to (700.center);
		\draw (702.center) to (693.center);
		\draw (702.center) to (694.center);
		\draw (698.center) to (691.center);
		\draw (710.center) to (711.center);
		\draw (708.center) to (714.center);
		\draw (710.center) to (716.center);
		\draw (711.center) to (715.center);
		\draw (717.center) to (708.center);
		\draw (717.center) to (709.center);
		\draw (713.center) to (706.center);
		\draw (707.center) to (709.center);
		\draw (462.center) to (709.center);
		\draw (510.center) to (694.center);
		\draw (724.center) to (725.center);
		\draw (722.center) to (728.center);
		\draw (724.center) to (730.center);
		\draw (725.center) to (729.center);
		\draw (731.center) to (722.center);
		\draw (731.center) to (723.center);
		\draw (727.center) to (720.center);
		\draw (721.center) to (723.center);
		\draw (738.center) to (739.center);
		\draw (736.center) to (742.center);
		\draw (738.center) to (744.center);
		\draw (739.center) to (743.center);
		\draw (745.center) to (736.center);
		\draw (745.center) to (737.center);
		\draw (741.center) to (734.center);
		\draw (735.center) to (737.center);
		\draw (754.center) to (755.center);
		\draw (752.center) to (758.center);
		\draw (754.center) to (760.center);
		\draw (755.center) to (759.center);
		\draw (761.center) to (752.center);
		\draw (761.center) to (753.center);
		\draw (757.center) to (750.center);
		\draw (766.center) to (767.center);
		\draw (764.center) to (770.center);
		\draw (766.center) to (772.center);
		\draw (767.center) to (771.center);
		\draw (773.center) to (764.center);
		\draw (773.center) to (765.center);
		\draw (639.center) to (769.center);
		\draw (765.center) to (763.center);
	\end{pgfonlayer}
\end{tikzpicture}

%% file: figures/T2shift.tikz
\begin{tikzpicture}[baseline=(O.base)]
\node (O) at (0,3*\hdist) {};
\foreach \t in  {0,...,5}{
\draw[black]  (\bitdistance*\t+\bitsize*\t,10*\hdist) -- (\bitdistance*\t+\bitsize*\t, \bitsize*0.5+2pt) ;
\node[bit]  at (\bitdistance*\t+\bitsize*\t,0)   {};
}
\foreach \t in {0,...,2}{
\node[M1] at (\bitsize+\transfoffset+\transfdistance*\t,\hdist) {$M_1$};
\node[pi] at (\bitsize+\transfoffset+\transfdistance*\t,0) {$\Pi$};
}
\foreach \t in {0,...,1}{
\node[M2] at (\bitsize+\transfoffset+\transfsize*0.5+\transfdistance*\t,4*\hdist) {$M_2$};
}
\foreach \t in {0,...,1}{
\node[M1] at (\bitsize+\transfoffset+\transfdistance*\t+\transfsize*0.5,7*\hdist) {$M_1$};
}
\foreach \t in {0,...,2}{
\node[M2] at (\bitsize+\transfoffset+\transfdistance*\t,10*\hdist) {$M_2$};
}

\draw[fill=myblue2,rounded corners=2pt,draw=white,opacity=1] (\bitsize+\transfoffset+\transfsize*0.5+\transfdistance*2,3*\hdist)--(\bitsize+\transfoffset+\transfdistance*2,3*\hdist)--(\bitsize+\transfoffset+\transfdistance*2,5*\hdist)--(\bitsize+\transfoffset+\transfsize*0.5+\transfdistance*2,5*\hdist);
\draw[black,rounded corners=2pt, line width=1pt,opacity=1] (\bitsize+\transfoffset+\transfsize*0.5+\transfdistance*2,3*\hdist)--(\bitsize+\transfoffset+\transfdistance*2,3*\hdist)--(\bitsize+\transfoffset+\transfdistance*2,5*\hdist)--(\bitsize+\transfoffset+\transfsize*0.5+\transfdistance*2,5*\hdist);
\draw[fill=myblue1,rounded corners=2pt,draw=white,opacity=1] (\bitsize+\transfoffset+\transfsize*0.5+\transfdistance*2,6*\hdist)--(\bitsize+\transfoffset+\transfdistance*2,6*\hdist)--(\bitsize+\transfoffset+\transfdistance*2,8*\hdist)--(\bitsize+\transfoffset+\transfsize*0.5+\transfdistance*2,8*\hdist);
\draw[black,rounded corners=2pt, line width=1pt,opacity=1] (\bitsize+\transfoffset+\transfsize*0.5+\transfdistance*2,6*\hdist)--(\bitsize+\transfoffset+\transfdistance*2,6*\hdist)--(\bitsize+\transfoffset+\transfdistance*2,8*\hdist)--(\bitsize+\transfoffset+\transfsize*0.5+\transfdistance*2,8*\hdist);

\draw[fill=myblue2,rounded corners=2pt,draw=white,opacity=1] (\halfleft,3*\hdist)--(-\bitsize-\transfoffset,3*\hdist)--(-\bitsize-\transfoffset,5*\hdist)--(\halfleft,5*\hdist);
\draw[black,rounded corners=2pt, line width=1pt,opacity=1] (-\bitsize-\transfoffset,3*\hdist)--(\halfleft,3*\hdist)--(\halfleft,5*\hdist)--(-\bitsize-\transfoffset,5*\hdist);
\draw[fill=myblue1,rounded corners=2pt,draw=white,opacity=1] (\halfleft,6*\hdist)--(-\bitsize-\transfoffset,6*\hdist)--(-\bitsize-\transfoffset,8*\hdist)--(\halfleft,8*\hdist);
\draw[black,rounded corners=2pt, line width=1pt,opacity=1] (-\bitsize-\transfoffset,6*\hdist)--(\halfleft,6*\hdist)--(\halfleft,8*\hdist)--(-\bitsize-\transfoffset,8*\hdist);
\end{tikzpicture}

%% file: figures/Sepev.tikz
\begin{tikzpicture}[baseline=(O.base)]
\node (O) at (0,2*\hdist) {};
\foreach \t in  {0,...,5}{
\draw[black]  (\bitdistance*\t+\bitsize*\t,10*\hdist) -- (\bitdistance*\t+\bitsize*\t, \bitsize*0.5+2pt) ;
\node[bit]  at (\bitdistance*\t+\bitsize*\t,0)   {};
}
\foreach \t in {0,...,2}{
\node[LR] at (\bitsize+\transfoffset+\transfdistance*\t,0) {\LR};
}
\foreach \t in {0,...,1}{
\node[M2] at (\bitsize+\transfoffset+\transfsize*0.5+\transfdistance*\t,\hdist) {$M_2$};
}
\foreach \t in {0,...,1}{
\node[M1] at (\bitsize+\transfoffset+\transfdistance*\t+\transfsize*0.5,4*\hdist) {$M_1$};
}
\foreach \t in {0,...,2}{
\node[M2] at (\bitsize+\transfoffset+\transfdistance*\t,7*\hdist) {$M_2$};
}

\draw[fill=myblue2,rounded corners=2pt,draw=white,opacity=1] (\bitsize+\transfoffset+\transfsize*0.5+\transfdistance*2,\bitdistance+\halfsize)--(\bitsize+\transfoffset+\transfdistance*2,\bitdistance+\halfsize)--(\bitsize+\transfoffset+\transfdistance*2,2*\hdist)--(\bitsize+\transfoffset+\transfsize*0.5+\transfdistance*2,2*\hdist);
\draw[black,rounded corners=2pt, line width=1pt,opacity=1] (\bitsize+\transfoffset+\transfsize*0.5+\transfdistance*2,\bitdistance+\halfsize)--(\bitsize+\transfoffset+\transfdistance*2,\bitdistance+\halfsize)--(\bitsize+\transfoffset+\transfdistance*2,2*\hdist)--(\bitsize+\transfoffset+\transfsize*0.5+\transfdistance*2,2*\hdist);
\draw[fill=myblue1,rounded corners=2pt,draw=white,opacity=1] (\bitsize+\transfoffset+\transfsize*0.5+\transfdistance*2,3*\hdist)--(\bitsize+\transfoffset+\transfdistance*2,3*\hdist)--(\bitsize+\transfoffset+\transfdistance*2,5*\hdist)--(\bitsize+\transfoffset+\transfsize*0.5+\transfdistance*2,5*\hdist);
\draw[black,rounded corners=2pt, line width=1pt,opacity=1] (\bitsize+\transfoffset+\transfsize*0.5+\transfdistance*2,3*\hdist)--(\bitsize+\transfoffset+\transfdistance*2,3*\hdist)--(\bitsize+\transfoffset+\transfdistance*2,5*\hdist)--(\bitsize+\transfoffset+\transfsize*0.5+\transfdistance*2,5*\hdist);

\draw[fill=myblue2,rounded corners=2pt,draw=white,opacity=1] (\halfleft,\bitdistance+\halfsize)--(-\bitsize-\transfoffset,\bitdistance+\halfsize)--(-\bitsize-\transfoffset,2*\hdist)--(\halfleft,2*\hdist);
\draw[black,rounded corners=2pt, line width=1pt,opacity=1] (-\bitsize-\transfoffset,\bitdistance+\halfsize)--(\halfleft,\bitdistance+\halfsize)--(\halfleft,2*\hdist)--(-\bitsize-\transfoffset,2*\hdist);
\draw[fill=myblue1,rounded corners=2pt,draw=white,opacity=1] (\halfleft,3*\hdist)--(-\bitsize-\transfoffset,3*\hdist)--(-\bitsize-\transfoffset,5*\hdist)--(\halfleft,5*\hdist);
\draw[black,rounded corners=2pt, line width=1pt,opacity=1] (-\bitsize-\transfoffset,3*\hdist)--(\halfleft,3*\hdist)--(\halfleft,5*\hdist)--(-\bitsize-\transfoffset,5*\hdist);
\end{tikzpicture}

%% file: figures/Pichain.tikz
\begin{tikzpicture}[baseline=(O.base)]
\node (O) at (0,2*\hdist) {};
\foreach \t in  {0,...,5}{
\draw[black]  (\bitdistance*\t+\bitsize*\t,10*\hdist) -- (\bitdistance*\t+\bitsize*\t, \bitsize*0.5+2pt) ;
}
\foreach \t in {0,...,2}{
\node[pi] at (\bitsize+\transfoffset+\transfdistance*\t,0) {$\Pi$};
}

\end{tikzpicture}

%% file: figures/M1.tikz.tex
\begin{tikzpicture}
\draw[black]  (\bitdistance+\bitsize,1)--(\bitdistance+\bitsize,-1);
\draw[black]  (0,1)--(0,-1);
\node[M1] at (\bitsize+\transfoffset,0) {$M_1$};
\end{tikzpicture}

%% file: figures/Ggen.tikz
\begin{tikzpicture}
\draw[black]  (\bitdistance+\bitsize,1.5)--(\bitdistance+\bitsize,-1.5);
\draw[black]  (0,1.5)--(0,-1.5);
\node[M2] at (\bitsize+\transfoffset,-0.75) {$M_2$};
\node[M1] at (\bitsize+\transfoffset,0.75) {$M_1$};
\end{tikzpicture}

%% file: figures/Gfig.tikz
\begin{tikzpicture}
\draw[black]  (\bitdistance+\bitsize,1)--(\bitdistance+\bitsize,-1);
\draw[black]  (0,1)--(0,-1);
\node[G] at (\bitsize+\transfoffset,0) {$G$};
\end{tikzpicture}

%% file: figures/Gev.tikz.tex
\begin{tikzpicture}[baseline=(O.base)]
\node (O) at (0,2*\hdist) {};
\foreach \t in  {0,...,5}{
\draw[black]  (\bitdistance*\t+\bitsize*\t,6*\hdist) -- (\bitdistance*\t+\bitsize*\t, \bitsize*0.5+2pt) ;
\node[bit]  at (\bitdistance*\t+\bitsize*\t,0)   {};
}
\foreach \t in {0,...,2}{
\node[LR] at (\bitsize+\transfoffset+\transfdistance*\t,0) {\LR};
}
\foreach \t in {0,...,1}{
\node[G] at (\bitsize+\transfoffset+\transfsize*0.5+\transfdistance*\t,\hdist) {$G$};
}
\foreach \t in {0,...,2}{
\node[G] at (\bitsize+\transfoffset+\transfdistance*\t,4*\hdist) {$G$};
}
\draw[fill=mycolor,rounded corners=2pt,draw=white,opacity=1] (\bitsize+\transfoffset+\transfsize*0.5+\transfdistance*2,\bitdistance+\halfsize)--(\bitsize+\transfoffset+\transfdistance*2,\bitdistance+\halfsize)--(\bitsize+\transfoffset+\transfdistance*2,2*\hdist)--(\bitsize+\transfoffset+\transfsize*0.5+\transfdistance*2,2*\hdist);
\draw[black,rounded corners=2pt, line width=1pt,opacity=1] (\bitsize+\transfoffset+\transfsize*0.5+\transfdistance*2,\bitdistance+\halfsize)--(\bitsize+\transfoffset+\transfdistance*2,\bitdistance+\halfsize)--(\bitsize+\transfoffset+\transfdistance*2,2*\hdist)--(\bitsize+\transfoffset+\transfsize*0.5+\transfdistance*2,2*\hdist);

\draw[fill=mycolor,rounded corners=2pt,draw=white,opacity=1] (\halfleft,\bitdistance+\halfsize)--(-\bitsize-\transfoffset,\bitdistance+\halfsize)--(-\bitsize-\transfoffset,2*\hdist)--(\halfleft,2*\hdist);
\draw[black,rounded corners=2pt, line width=1pt,opacity=1] (-\bitsize-\transfoffset,\bitdistance+\halfsize)--(\halfleft,\bitdistance+\halfsize)--(\halfleft,2*\hdist)--(-\bitsize-\transfoffset,2*\hdist);

\end{tikzpicture}

%% file: figures/LeftG.tikz
\begin{tikzpicture}
\foreach \t in  {0,...,5}{
\draw[black]  (\bitdistance*\t+\bitsize*\t,3*\hdist) -- (\bitdistance*\t+\bitsize*\t, \bitsize*0.5+2pt) ;
\node[bit]  at (\bitdistance*\t+\bitsize*\t,0)   {};
}
\foreach \t in {0,...,2}{
\node[LR] at (\bitsize+\transfoffset+\transfdistance*\t,0) {\LR};
}
\foreach \t in {0,...,1}{
\node[G] at (\bitsize+\transfoffset+\transfsize*0.5+\transfdistance*\t,\hdist) {$G$};
}

\draw[fill=mycolor,rounded corners=2pt,draw=white,opacity=1] (\bitsize+\transfoffset+\transfsize*0.5+\transfdistance*2,\bitdistance+\halfsize)--(\bitsize+\transfoffset+\transfdistance*2,\bitdistance+\halfsize)--(\bitsize+\transfoffset+\transfdistance*2,2*\hdist)--(\bitsize+\transfoffset+\transfsize*0.5+\transfdistance*2,2*\hdist);
\draw[black,rounded corners=2pt, line width=1pt,opacity=1] (\bitsize+\transfoffset+\transfsize*0.5+\transfdistance*2,\bitdistance+\halfsize)--(\bitsize+\transfoffset+\transfdistance*2,\bitdistance+\halfsize)--(\bitsize+\transfoffset+\transfdistance*2,2*\hdist)--(\bitsize+\transfoffset+\transfsize*0.5+\transfdistance*2,2*\hdist);

\draw[fill=mycolor,rounded corners=2pt,draw=white,opacity=1] (\halfleft,\bitdistance+\halfsize)--(-\bitsize-\transfoffset,\bitdistance+\halfsize)--(-\bitsize-\transfoffset,2*\hdist)--(\halfleft,2*\hdist);
\draw[black,rounded corners=2pt, line width=1pt,opacity=1] (-\bitsize-\transfoffset,\bitdistance+\halfsize)--(\halfleft,\bitdistance+\halfsize)--(\halfleft,2*\hdist)--(-\bitsize-\transfoffset,2*\hdist);

\end{tikzpicture}

%% file: figures/Newchain.tikz
\begin{tikzpicture}
\foreach \t in  {0,...,5}{
\draw[black]  (\bitdistance*\t+\bitsize*\t,1) -- (\bitdistance*\t+\bitsize*\t, \bitsize*0.5+2pt) ;
\node[bit]  at (\bitdistance*\t+\bitsize*\t,0)   {};
}
\foreach \t in {0,...,2}{
\node[LR] at (\bitsize+\transfoffset+\transfdistance*\t,0) {\LRbar};
}

\end{tikzpicture}

%% file: figures/NewDiff.tikz
\begin{tikzpicture}
\foreach \t in  {1,...,2}{
\draw[black]  (\bitdistance*\t+\bitsize*\t,1) -- (\bitdistance*\t+\bitsize*\t, \bitsize*0.5+2pt) ;
\node[bit]  at (\bitdistance*\t+\bitsize*\t,0)   {};
}
\node[LR] at (\bitsize+\transfoffset+\transfdistance*0.5,0) {\LRbar};
\end{tikzpicture}

%% file: figures/LeftGconj.tikz
\begin{tikzpicture}
\foreach \t in  {0,...,7}{
\draw[black]  (\bitdistance*\t+\bitsize*\t,3*\hdist) -- (\bitdistance*\t+\bitsize*\t, \bitsize*0.5+2pt) ;
\node[bit]  at (\bitdistance*\t+\bitsize*\t,0)   {};
}
\foreach \t in {0,...,3}{
\node[LR] at (\bitsize+\transfoffset+\transfdistance*\t,0) {\LR};
}

\node[G] at (\bitsize+\transfoffset+\transfsize*0.5+\transfdistance,\hdist) {$G$};

\draw[fill=mycolor,rounded corners=2pt,draw=white,opacity=1] (\bitsize+\transfoffset+\transfsize*0.5+\transfdistance*3,\bitdistance+\halfsize)--(\bitsize+\transfoffset+\transfdistance*3,\bitdistance+\halfsize)--(\bitsize+\transfoffset+\transfdistance*3,2*\hdist)--(\bitsize+\transfoffset+\transfsize*0.5+\transfdistance*3,2*\hdist);
\draw[black,rounded corners=2pt, line width=1pt,opacity=1] (\bitsize+\transfoffset+\transfsize*0.5+\transfdistance*3,\bitdistance+\halfsize)--(\bitsize+\transfoffset+\transfdistance*3,\bitdistance+\halfsize)--(\bitsize+\transfoffset+\transfdistance*3,2*\hdist)--(\bitsize+\transfoffset+\transfsize*0.5+\transfdistance*3,2*\hdist);

\draw[fill=mycolor,rounded corners=2pt,draw=white,opacity=1] (\halfleft,\bitdistance+\halfsize)--(-\bitsize-\transfoffset,\bitdistance+\halfsize)--(-\bitsize-\transfoffset,2*\hdist)--(\halfleft,2*\hdist);
\draw[black,rounded corners=2pt, line width=1pt,opacity=1] (-\bitsize-\transfoffset,\bitdistance+\halfsize)--(\halfleft,\bitdistance+\halfsize)--(\halfleft,2*\hdist)--(-\bitsize-\transfoffset,2*\hdist);

\end{tikzpicture}

%% file: figures/Newchainconj.tikz
\begin{tikzpicture}
\foreach \t in  {0,...,7}{
\draw[black]  (\bitdistance*\t+\bitsize*\t,3*\hdist) -- (\bitdistance*\t+\bitsize*\t, \bitsize*0.5+2pt) ;
\node[bit]  at (\bitdistance*\t+\bitsize*\t,0)   {};
}
\foreach \t in {0,...,3}{
\node[LR] at (\bitsize+\transfoffset+\transfdistance*\t,0) {\LRbar};
}
\foreach \t in {0,...,1}{
\node[G] at (\bitsize+\transfoffset+\transfsize*0.5+\transfdistance*\t*2,\hdist) {$G^\dag$};
}

\end{tikzpicture}

%% file: figures/LeftGbis.tikz
\begin{tikzpicture}
\foreach \t in  {0,...,3}{
\draw[black]  (\bitdistance*\t+\bitsize*\t,3*\hdist) -- (\bitdistance*\t+\bitsize*\t, \bitsize*0.5+2pt) ;
\node[bit]  at (\bitdistance*\t+\bitsize*\t,0)   {};
}
\foreach \t in {0,...,1}{
\node[LR] at (\bitsize+\transfoffset+\transfdistance*\t,0) {\LR};
}

\node[G] at (\bitsize+\transfoffset+\transfsize*0.5,\hdist) {$G$};



\end{tikzpicture}

%% file: figures/Newchainbis.tikz
\begin{tikzpicture}
\foreach \t in  {0,...,3}{
\draw[black]  (\bitdistance*\t+\bitsize*\t,1) -- (\bitdistance*\t+\bitsize*\t, \bitsize*0.5+2pt) ;
\node[bit]  at (\bitdistance*\t+\bitsize*\t,0)   {};
}
\node[LR] at (\bitsize+\transfoffset,0) {\LRtilde};
\node[LR] at (\bitsize+\transfoffset+\transfdistance,0) {\LtilR};

\end{tikzpicture}

%% file: figures/M2.tikz.tex
\begin{tikzpicture}
\draw[black]  (\bitdistance+\bitsize,-1)--(\bitdistance+\bitsize,1);
\draw[black]  (0,-1)--(0,1);
\node[M2] at (\bitsize+\transfoffset,0) {$M_2$};
\end{tikzpicture}

%% file: figures/Uswap.tikz
\begin{tikzpicture}
\draw[black]  (\bitdistance+\bitsize,1.5+2*\hdist)--(\bitdistance+\bitsize,-0.9);
\draw[black]  (0,1.5+2*\hdist)--(0,-0.9);
\node[U] at (0,0) {$J$};
\node[U] at (\bitdistance+\bitsize,0) {$W$};
\node[Cif] at (\bitsize+\transfoffset,0.5*\hdist) {$C$};
\draw[black]  (0,1.3*\nhdist) -- (\bitdistance+\bitsize,-0.9) ;
\draw[black]  (0,-0.9) -- (\bitdistance+\bitsize,1.3*\nhdist) ;
\draw[black]  (\bitdistance+\bitsize,1.3*\nhdist)--(\bitdistance+\bitsize,2*\nhdist);
\draw[black]  (0,1.3*\nhdist)--(0,2*\nhdist);
\end{tikzpicture}

%% file: figures/M2qbit.tikz.tex
\begin{tikzpicture}
\draw[black]  (\bitdistance+\bitsize,1.5)--(\bitdistance+\bitsize,-1);
\draw[black]  (0,1.5)--(0,-1);
\node[Cif] at (\bitsize+\transfoffset,0) {$C_\phi$};
\node[U] at (0,\hdist*0.5) {$U$};
\node[U] at (\bitdistance+\bitsize,\hdist*0.5) {$U$};
\end{tikzpicture}

%% file: figures/def2.tikz
\begin{tikzpicture}
\draw[black]  (\bitdistance+\bitsize,1)--(\bitdistance+\bitsize,-1);
\draw[black]  (0,1)--(0,-1);
\node[Cif] at (\bitsize+\transfoffset,0) {$C_\phi$};
\end{tikzpicture}

%% file: figures/T2noproj.tikz.tex
\begin{tikzpicture}[baseline=(O.base)]
\node (O) at (0,3*\hdist) {};
\foreach \t in  {0,...,5}{
\draw[black]  (\bitdistance*\t+\bitsize*\t,10*\hdist) -- (\bitdistance*\t+\bitsize*\t, \bitsize*0.5+2pt) ;
}
\foreach \t in {0,...,2}{
\node[M1] at (\bitsize+\transfoffset+\transfdistance*\t,\hdist) {$M_1$};
}
\foreach \t in {0,...,1}{
\node[M2] at (\bitsize+\transfoffset+\transfsize*0.5+\transfdistance*\t,4*\hdist) {$M_2$};
}
\foreach \t in {0,...,2}{
\node[M1] at (\bitsize+\transfoffset+\transfdistance*\t,7*\hdist) {$M_1$};
}
\foreach \t in {0,...,1}{
\node[M2] at (\bitsize+\transfoffset+\transfsize*0.5+\transfdistance*\t,10*\hdist) {$M_2$};
}

\draw[fill=myblue2,rounded corners=2pt,draw=white,opacity=1] (\bitsize+\transfoffset+\transfsize*0.5+\transfdistance*2,3*\hdist)--(\bitsize+\transfoffset+\transfdistance*2,3*\hdist)--(\bitsize+\transfoffset+\transfdistance*2,5*\hdist)--(\bitsize+\transfoffset+\transfsize*0.5+\transfdistance*2,5*\hdist);
\draw[black,rounded corners=2pt, line width=1pt,opacity=1] (\bitsize+\transfoffset+\transfsize*0.5+\transfdistance*2,3*\hdist)--(\bitsize+\transfoffset+\transfdistance*2,3*\hdist)--(\bitsize+\transfoffset+\transfdistance*2,5*\hdist)--(\bitsize+\transfoffset+\transfsize*0.5+\transfdistance*2,5*\hdist);
\draw[fill=myblue2,rounded corners=2pt,draw=white,opacity=1] (\bitsize+\transfoffset+\transfsize*0.5+\transfdistance*2,9*\hdist)--(\bitsize+\transfoffset+\transfdistance*2,9*\hdist)--(\bitsize+\transfoffset+\transfdistance*2,11*\hdist)--(\bitsize+\transfoffset+\transfsize*0.5+\transfdistance*2,11*\hdist);
\draw[black,rounded corners=2pt, line width=1pt,opacity=1] (\bitsize+\transfoffset+\transfsize*0.5+\transfdistance*2,9*\hdist)--(\bitsize+\transfoffset+\transfdistance*2,9*\hdist)--(\bitsize+\transfoffset+\transfdistance*2,11*\hdist)--(\bitsize+\transfoffset+\transfsize*0.5+\transfdistance*2,11*\hdist);

\draw[fill=myblue2,rounded corners=2pt,draw=white,opacity=1] (\halfleft,3*\hdist)--(-\bitsize-\transfoffset,3*\hdist)--(-\bitsize-\transfoffset,5*\hdist)--(\halfleft,5*\hdist);
\draw[black,rounded corners=2pt, line width=1pt,opacity=1] (-\bitsize-\transfoffset,3*\hdist)--(\halfleft,3*\hdist)--(\halfleft,5*\hdist)--(-\bitsize-\transfoffset,5*\hdist);
\draw[fill=myblue2,rounded corners=2pt,draw=white,opacity=1] (\halfleft,9*\hdist)--(-\bitsize-\transfoffset,9*\hdist)--(-\bitsize-\transfoffset,11*\hdist)--(\halfleft,11*\hdist);
\draw[black,rounded corners=2pt, line width=1pt,opacity=1] (-\bitsize-\transfoffset,9*\hdist)--(\halfleft,9*\hdist)--(\halfleft,11*\hdist)--(-\bitsize-\transfoffset,11*\hdist);

\end{tikzpicture}

%% file: figures/MargInd1sqr.tikz.tex
\begin{tikzpicture}[baseline=(O.base)]
\node (O) at (0,\hdist) {};
\foreach \t in  {0,...,5}{
\draw[black]  (\bitdistance*\t+\bitsize*\t,6*\hdist) -- (\bitdistance*\t+\bitsize*\t, \bitsize*0.5+2pt) ;
}
\foreach \t in {0,...,2}{
\node[M1] at (\bitsize+\transfoffset+\transfdistance*\t,\hdist) {$M_1^2$};
}
\foreach \t in {0,...,1}{
\node[M2] at (\bitsize+\transfoffset+\transfsize*0.5+\transfdistance*\t,4*\hdist) {$M_2^2$};
}

\draw[fill=myblue2,rounded corners=2pt,draw=white,opacity=1] (\bitsize+\transfoffset+\transfsize*0.5+\transfdistance*2,3*\hdist)--(\bitsize+\transfoffset+\transfdistance*2,3*\hdist)--(\bitsize+\transfoffset+\transfdistance*2,5*\hdist)--(\bitsize+\transfoffset+\transfsize*0.5+\transfdistance*2,5*\hdist);
\draw[black,rounded corners=2pt, line width=1pt,opacity=1] (\bitsize+\transfoffset+\transfsize*0.5+\transfdistance*2,3*\hdist)--(\bitsize+\transfoffset+\transfdistance*2,3*\hdist)--(\bitsize+\transfoffset+\transfdistance*2,5*\hdist)--(\bitsize+\transfoffset+\transfsize*0.5+\transfdistance*2,5*\hdist);

\draw[fill=myblue2,rounded corners=2pt,draw=white,opacity=1] (\halfleft,3*\hdist)--(-\bitsize-\transfoffset,3*\hdist)--(-\bitsize-\transfoffset,5*\hdist)--(\halfleft,5*\hdist);
\draw[black,rounded corners=2pt, line width=1pt,opacity=1] (-\bitsize-\transfoffset,3*\hdist)--(\halfleft,3*\hdist)--(\halfleft,5*\hdist)--(-\bitsize-\transfoffset,5*\hdist);

\end{tikzpicture}

%% file: figures/MargInd12.tikz
\begin{tikzpicture}[baseline=(O.base)]
\node (O) at (0,\hdist) {};
\foreach \t in  {0,...,5}{
\draw[black]  (\bitdistance*\t+\bitsize*\t,6*\hdist) -- (\bitdistance*\t+\bitsize*\t, \bitsize*0.5+2pt) ;
}
\foreach \t in {0,...,2}{
\node[M1] at (\bitsize+\transfoffset+\transfdistance*\t,\hdist) {$M_1'$};
}
\foreach \t in {0,...,1}{
\node[M2] at (\bitsize+\transfoffset+\transfsize*0.5+\transfdistance*\t,4*\hdist) {$M_2'$};
}

\draw[fill=myblue2,rounded corners=2pt,draw=white,opacity=1] (\bitsize+\transfoffset+\transfsize*0.5+\transfdistance*2,3*\hdist)--(\bitsize+\transfoffset+\transfdistance*2,3*\hdist)--(\bitsize+\transfoffset+\transfdistance*2,5*\hdist)--(\bitsize+\transfoffset+\transfsize*0.5+\transfdistance*2,5*\hdist);
\draw[black,rounded corners=2pt, line width=1pt,opacity=1] (\bitsize+\transfoffset+\transfsize*0.5+\transfdistance*2,3*\hdist)--(\bitsize+\transfoffset+\transfdistance*2,3*\hdist)--(\bitsize+\transfoffset+\transfdistance*2,5*\hdist)--(\bitsize+\transfoffset+\transfsize*0.5+\transfdistance*2,5*\hdist);

\draw[fill=myblue2,rounded corners=2pt,draw=white,opacity=1] (\halfleft,3*\hdist)--(-\bitsize-\transfoffset,3*\hdist)--(-\bitsize-\transfoffset,5*\hdist)--(\halfleft,5*\hdist);
\draw[black,rounded corners=2pt, line width=1pt,opacity=1] (-\bitsize-\transfoffset,3*\hdist)--(\halfleft,3*\hdist)--(\halfleft,5*\hdist)--(-\bitsize-\transfoffset,5*\hdist);

\end{tikzpicture}

%% file: figures/M1id.tikz.tex
\begin{tikzpicture}
\foreach \t in  {0,...,5}{
\draw[black]  (\bitdistance*\t+\bitsize*\t,3*\hdist) -- (\bitdistance*\t+\bitsize*\t, -1.5) ;
}
\foreach \t in {0,...,2}{
\node[M1] at (\bitsize+\transfoffset+\transfdistance*\t,0) {$M_1'^\dag$};
\node[M1] at (\bitsize+\transfoffset+\transfdistance*\t,\hdist) {$M_1$};
}

\end{tikzpicture}

%% file: figures/M2id.tikz.tex
\begin{tikzpicture}
\foreach \t in  {0,...,5}{
\draw[black]  (\bitdistance*\t+\bitsize*\t,3*\hdist) -- (\bitdistance*\t+\bitsize*\t, -1.5) ;
}
\foreach \t in {0,...,1}{
\node[M2] at (\bitsize+\transfoffset+\transfsize*0.5+\transfdistance*\t,0) {$M_2'$};
\node[M2] at (\bitsize+\transfoffset+\transfsize*0.5+\transfdistance*\t,\hdist) {$M_2^\dag$};
}

\draw[fill=myblue2,rounded corners=2pt,draw=white,opacity=1] (\bitsize+\transfoffset+\transfsize*0.5+\transfdistance*2,-\bitsize)--(\bitsize+\transfoffset+\transfdistance*2,-\bitsize)--(\bitsize+\transfoffset+\transfdistance*2,\bitsize)--(\bitsize+\transfoffset+\transfsize*0.5+\transfdistance*2,\bitsize);
\draw[fill=myblue2,rounded corners=2pt,draw=white,opacity=1] (\bitsize+\transfoffset+\transfsize*0.5+\transfdistance*2,0.01*\hdist)--(\bitsize+\transfoffset+\transfdistance*2,0.01*\hdist)--(\bitsize+\transfoffset+\transfdistance*2,2*\hdist)--(\bitsize+\transfoffset+\transfsize*0.5+\transfdistance*2,2*\hdist);
\draw[black,rounded corners=2pt, line width=1pt,opacity=1]  (\bitsize+\transfoffset+\transfsize*0.5+\transfdistance*2,-\bitsize)--(\bitsize+\transfoffset+\transfdistance*2,-\bitsize)--(\bitsize+\transfoffset+\transfdistance*2,\bitsize)--(\bitsize+\transfoffset+\transfsize*0.5+\transfdistance*2,\bitsize);
\draw[black,rounded corners=2pt, line width=1pt,opacity=1] (\bitsize+\transfoffset+\transfsize*0.5+\transfdistance*2,0.01*\hdist)--(\bitsize+\transfoffset+\transfdistance*2,0.01*\hdist)--(\bitsize+\transfoffset+\transfdistance*2,2*\hdist)--(\bitsize+\transfoffset+\transfsize*0.5+\transfdistance*2,2*\hdist);

\draw[fill=myblue2,rounded corners=2pt,draw=white,opacity=1] (-\bitsize-\transfoffset,-\bitsize)--(\halfleft,-\bitsize)--(\halfleft,\bitsize)--(-\bitsize-\transfoffset,\bitsize);
\draw[fill=myblue2,rounded corners=2pt,draw=white,opacity=1] (-\bitsize-\transfoffset,0.01*\hdist)--(\halfleft,0.01*\hdist)--(\halfleft,2*\hdist)--(-\bitsize-\transfoffset,2*\hdist);
\draw[black,rounded corners=2pt, line width=1pt,opacity=1] (-\bitsize-\transfoffset,-\bitsize)--(\halfleft,-\bitsize)--(\halfleft,\bitsize)--(-\bitsize-\transfoffset,\bitsize);
\draw[black,rounded corners=2pt, line width=1pt,opacity=1] (-\bitsize-\transfoffset,0.01*\hdist)--(\halfleft,0.01*\hdist)--(\halfleft,2*\hdist)--(-\bitsize-\transfoffset,2*\hdist);

\end{tikzpicture}

%% file: figures/M1unit.tikz.tex
\begin{tikzpicture}
\foreach \t in {0,...,5}{
\draw[black]  (\t*\bitdistance+\t*\bitsize,1)--(\t*\bitdistance+\t*\bitsize,-1);
}
\foreach \t in {0,...,2}{
\node[U] at (\bitdistance*\t*2+\bitsize*\t*2,0) {$U$};
\node[U] at (\bitsize+\bitdistance+\bitdistance*\t*2+\bitsize*\t*2,0) {$V$};
}
\end{tikzpicture}

%% file: figures/Rdnew.tikz
\begin{tikzpicture}
	\begin{pgfonlayer}{nodelayer}
		\node [style=none] (0) at (-1, -0.65) {};
		\node [style=none] (1) at (1, -0.65) {};
		\node [style=none] (2) at (0, 0.35) {};
	\end{pgfonlayer}
	\begin{pgfonlayer}{edgelayer}
		\draw (0.center) to (2.center);
		\draw (1.center) to (2.center);
	\end{pgfonlayer}
\end{tikzpicture}

%% file: figures/Rnew.tikz
\begin{tikzpicture}
	\begin{pgfonlayer}{nodelayer}
		\node [style=none] (0) at (-1, 0.35) {};
		\node [style=none] (1) at (1, 0.35) {};
		\node [style=none] (2) at (0, -0.65) {};
	\end{pgfonlayer}
	\begin{pgfonlayer}{edgelayer}
		\draw (0.center) to (2.center);
		\draw (1.center) to (2.center);
	\end{pgfonlayer}
\end{tikzpicture}

%% file: figures/RRdagnew.tikz
\begin{tikzpicture}
	\begin{pgfonlayer}{nodelayer}
		\node [style=none] (0) at (-1, 1.5) {};
		\node [style=none] (1) at (1, 1.5) {};
		\node [style=none] (2) at (0, 0.5) {};
		\node [style=none] (4) at (0, -0.5) {};
		\node [style=none] (5) at (-1, -1.5) {};
		\node [style=none] (6) at (1, -1.5) {};
	\end{pgfonlayer}
	\begin{pgfonlayer}{edgelayer}
		\draw (0.center) to (2.center);
		\draw (2.center) to (1.center);
		\draw (2.center) to (4.center);
		\draw (4.center) to (5.center);
		\draw (4.center) to (6.center);
	\end{pgfonlayer}
\end{tikzpicture}

%% file: figures/RdagR.tikz
\begin{tikzpicture}
	\begin{pgfonlayer}{nodelayer}
		\node [style=none] (0) at (-1, -0.05) {};
		\node [style=none] (1) at (1, -0.05) {};
		\node [style=none] (2) at (0, -1.05) {};
		\node [style=none] (3) at (-1, -0.025) {};
		\node [style=none] (4) at (1, -0.025) {};
		\node [style=none] (5) at (0, 0.975) {};
	\end{pgfonlayer}
	\begin{pgfonlayer}{edgelayer}
		\draw (0.center) to (2.center);
		\draw (1.center) to (2.center);
		\draw (3.center) to (5.center);
		\draw (4.center) to (5.center);
	\end{pgfonlayer}
\end{tikzpicture}

%% file: figures/Renev.tikz
\begin{tikzpicture}[baseline=(O.base)]
\node (O) at (0,2*\hdist) {};
\foreach \t in  {0,1,4,5}{
\draw[black]  (\bitdistance*\t+\bitsize*\t,8*\hdist) -- (\bitdistance*\t+\bitsize*\t, \bitsize*0.5+2pt) ;
\node[bit]  at (\bitdistance*\t+\bitsize*\t,0)   {};
}
\foreach \t in  {2,3}{
\draw[black]  (\bitdistance*\t+\bitsize*\t,8*\hdist) -- (\bitdistance*\t+\bitsize*\t, 1.8*\hdist) ;
}
\foreach \t in {0,2,4}{
\tikzmath{
\j=\t+1;
\i=\t+0.5;
}
\draw[black]  (\bitdistance*\t+\bitsize*\t,8*\hdist) -- (\bitdistance*\i+\bitsize*\i, 9*\hdist) ;
\draw[black]  (\bitdistance*\j+\bitsize*\j,8*\hdist) -- (\bitdistance*\i+\bitsize*\i, 9*\hdist) ;
}
\foreach \t in {0,...,5}{
\node at (\bitdistance*\t+\bitsize*\t,-1){\t};
}
\foreach \t in {0,2}{
\node[LR] at (\bitsize+\transfoffset+\transfdistance*\t,0) {\LR};
}
\node at (\bitdistance*0.5+\bitsize*0.5,9.5*\hdist){a};
\node at (\bitdistance*2.5+\bitsize*2.5,9.5*\hdist){b};
\node at (\bitdistance*4.5+\bitsize*4.5,9.5*\hdist){c};
\draw[black]  (\bitdistance*2.5+\bitsize*2.5,0.5*\hdist) -- (\bitdistance*2.5+\bitsize*2.5, -\bitsize*0.5-2pt) ;
\draw[black]  (\bitdistance*2.5+\bitsize*2.5,0.5*\hdist) -- (\bitdistance*2+\bitsize*2, 1.8*\hdist) ;
\draw[black]  (\bitdistance*2.5+\bitsize*2.5,0.5*\hdist) -- (\bitdistance*3+\bitsize*3, 1.8*\hdist) ;
\node[Op] at (\bitsize*1.5+\bitdistance*1.5+\transfdistance*0.5,0) {$B$};
\foreach \t in {0,...,1}{
\node[G] at (\bitsize+\transfoffset+\transfsize*0.5+\transfdistance*\t,3*\hdist) {$G$};
}
\foreach \t in {0,...,2}{
\node[G] at (\bitsize+\transfoffset+\transfdistance*\t,6*\hdist) {$G$};
}
\draw[fill=mycolor,rounded corners=2pt,draw=white,opacity=1] (\bitsize+\transfoffset+\transfsize*0.5+\transfdistance*2,2.25*\bitdistance+\halfsize)--(\bitsize+\transfoffset+\transfdistance*2,2.25*\bitdistance+\halfsize)--(\bitsize+\transfoffset+\transfdistance*2,4*\hdist)--(\bitsize+\transfoffset+\transfsize*0.5+\transfdistance*2,4*\hdist);
\draw[black,rounded corners=2pt, line width=1pt,opacity=1] (\bitsize+\transfoffset+\transfsize*0.5+\transfdistance*2,2.25*\bitdistance+\halfsize)--(\bitsize+\transfoffset+\transfdistance*2,2.25*\bitdistance+\halfsize)--(\bitsize+\transfoffset+\transfdistance*2,4*\hdist)--(\bitsize+\transfoffset+\transfsize*0.5+\transfdistance*2,4*\hdist);

\draw[fill=mycolor,rounded corners=2pt,draw=white,opacity=1] (\halfleft,2.25*\bitdistance+\halfsize)--(-\bitsize-\transfoffset,2.25*\bitdistance+\halfsize)--(-\bitsize-\transfoffset,4*\hdist)--(\halfleft,4*\hdist);
\draw[black,rounded corners=2pt, line width=1pt,opacity=1] (-\bitsize-\transfoffset,2.25*\bitdistance+\halfsize)--(\halfleft,2.25*\bitdistance+\halfsize)--(\halfleft,4*\hdist)--(-\bitsize-\transfoffset,4*\hdist);

\end{tikzpicture}

%% file: figures/Renevconj.tikz
\begin{tikzpicture}[baseline=(O.base)]
\node (O) at (0,2*\hdist) {};
\foreach \t in  {0,1,4,5}{
\draw[black]  (\bitdistance*\t+\bitsize*\t,8*\hdist) -- (\bitdistance*\t+\bitsize*\t, \bitsize*0.5+2pt) ;
\node[bit]  at (\bitdistance*\t+\bitsize*\t,0)   {};
}
\foreach \t in  {2,3}{
\draw[black]  (\bitdistance*\t+\bitsize*\t,8*\hdist) -- (\bitdistance*\t+\bitsize*\t, 1.8*\hdist) ;
}
\foreach \t in {0,2,4}{
\tikzmath{
\j=\t+1;
\i=\t+0.5;
}
\draw[black]  (\bitdistance*\t+\bitsize*\t,8*\hdist) -- (\bitdistance*\i+\bitsize*\i, 9*\hdist) ;
\draw[black]  (\bitdistance*\j+\bitsize*\j,8*\hdist) -- (\bitdistance*\i+\bitsize*\i, 9*\hdist) ;
}

\node[LR] at (\bitsize+\transfoffset,0) {\LtilR};
\node[LR] at (\bitsize+\transfoffset+2*\transfdistance,0) {\LRtilde};
\node at (\bitdistance*0.5+\bitsize*0.5,9.5*\hdist){a};
\node at (\bitdistance*2.5+\bitsize*2.5,9.5*\hdist){b};
\node at (\bitdistance*4.5+\bitsize*4.5,9.5*\hdist){c};
\draw[black]  (\bitdistance*2.5+\bitsize*2.5,0.5*\hdist) -- (\bitdistance*2.5+\bitsize*2.5, -\bitsize*0.5-2pt) ;
\draw[black]  (\bitdistance*2.5+\bitsize*2.5,0.5*\hdist) -- (\bitdistance*2+\bitsize*2, 1.8*\hdist) ;
\draw[black]  (\bitdistance*2.5+\bitsize*2.5,0.5*\hdist) -- (\bitdistance*3+\bitsize*3, 1.8*\hdist) ;
\node[Op] at (\bitsize*1.5+\bitdistance*1.5+\transfdistance*0.5,0) {$B$};
\foreach \t in {0,...,1}{
\node[G] at (\bitsize+\transfoffset+\transfsize*0.5+\transfdistance*\t,3*\hdist) {$G$};
}
\foreach \t in {0,...,2}{
\node[G] at (\bitsize+\transfoffset+\transfdistance*\t,6*\hdist) {$G$};
}

\end{tikzpicture}